\documentclass[11pt,twoside]{amsart}

\usepackage[english]{babel}
\usepackage{csquotes}
\usepackage{amsfonts}
\usepackage{hyperref}
\usepackage{etoolbox}
\usepackage{amsmath}
\usepackage{amssymb}
\usepackage{amsthm}
\usepackage{stmaryrd}
\usepackage{dsfont}
\usepackage{pifont}
\usepackage{caption}
\usepackage{subcaption}
\usepackage{mathrsfs}
\usepackage{array,multirow,graphicx}
\usepackage{enumerate}
\usepackage[all]{xy}
\usepackage{enumitem}
\usepackage{geometry}
\usepackage{amsfonts}
\usepackage{fancyhdr}
\usepackage{calc}
\usepackage{cancel}
\usepackage{accents}
\usepackage{abraces}
\usepackage[new]{old-arrows}
\usepackage{tikz,tikz-3dplot}
\usepackage{lscape}
\usepackage{algorithm}
\usepackage{algpseudocode}
\usepackage{tabularx}
\usepackage{slashbox}
\usepackage{float}
\usetikzlibrary{decorations.markings}
\usetikzlibrary{shapes,arrows,plotmarks}
\usepackage{orcidlink}
\usepackage[backend=biber,style=alphabetic,sorting=nty]{biblatex}
\usepackage{color}
\definecolor{linkColor}{rgb}{0.0,0.0,0.554}
\definecolor{citeColor}{rgb}{0.0,0.0,0.554}
\definecolor{fileColor}{rgb}{0.0,0.0,0.554}
\definecolor{urlColor}{rgb}{0.0,0.0,0.554}
\definecolor{promptColor}{rgb}{0.0,0.0,0.589}
\definecolor{brkpromptColor}{rgb}{0.589,0.0,0.0}
\definecolor{gapinputColor}{rgb}{0.589,0.0,0.0}
\definecolor{gapoutputColor}{rgb}{0.0,0.0,0.0}
\usepackage{fancyvrb}
\usepackage{adjustbox}
\usepackage{helvet}

\definecolor{cof}{RGB}{219,144,71}
\definecolor{pur}{RGB}{186,146,162}
\definecolor{greeo}{RGB}{91,173,69}
\definecolor{greet}{RGB}{52,111,72}

\tdplotsetmaincoords{70}{165}

\newcommand{\changefont}{%
    \fontsize{8}{8}\selectfont
}
\mathchardef\mhyphen="2D 

\makeatletter
\def\BState{\State\hskip-\ALG@thistlm}
\makeatother

\title[Some integrators for the KNdS geodesic equation and black holes shadowing]{Motion equations in a Kerr--Newman--de Sitter spacetime: some methods of integration and application to black holes shadowing in Scilab}

\author{Arthur Garnier \orcidlink{0000-0003-4069-3203}}
\address{\newline
Universit\'e de Picardie,
\newline D\'epartement de Math\'ematiques et LAMFA (UMR 7352 du CNRS),
\newline 33 rue St Leu,
\newline F-80039 Amiens Cedex 1,
\newline France}
\email{arthur.garnier@math.cnrs.fr}

\theoremstyle{plain}
\newtheorem{prop}{Proposition}[subsection]
\newtheorem{prop-def}[prop]{Proposition-Definition}

\newtheorem{theo}[prop]{Theorem}
\newtheorem{cor}[prop]{Corollary}
\newtheorem{rem}[prop]{Remark}
\newtheorem{definition}[prop]{Definition}

\newtheorem*{prop*}{Proposition}
\newtheorem*{prop-def*}{Proposition-Definition}
\newtheorem*{propri*}{Property}
\newtheorem*{lem*}{Lemma}
\newtheorem*{theo*}{Theorem}
\newtheorem*{cor*}{Corollary}
\newtheorem*{rem*}{Remark}
\newtheorem*{definition*}{Definition}
\newtheorem*{exemple*}{Example}
\newtheorem*{notation*}{Notation}

\newcommand{\lra}{\longrightarrow}

\newcommand{\ra}{\rightarrow}
\newcommand{\sdp}{\times\kern-.2em\vrule height1.1ex depth-.05ex}
\newcommand{\epi}{\lra \kern-.8em\ra}
\newcommand{\C}{{\mathbb C}}

\newcommand{\R}{{\mathbb R}}

\newcommand{\Sph}{\mathbb{S}}

 \normalbaselines
\makeatletter
\newlength\@SizeOfCirc%
\newcommand{\CircleArrowRight}[1]{%
    \setlength{\@SizeOfCirc}{\maxof{\widthof{#1}}{\heightof{#1}}}%
    \tikz [x=1.0ex,y=1.0ex,line width=.12ex]%
        \draw [->,anchor=center]%
            node (0,0) {#1}%
            (0,0.8\@SizeOfCirc) arc (85:-240:0.8\@SizeOfCirc);%
}%
\newcommand{\CircleArrowLeft}[1]{%
    \setlength{\@SizeOfCirc}{\maxof{\widthof{#1}}{\heightof{#1}}}%
    \tikz [x=1.0ex,y=1.0ex,line width=.12ex]%
        \draw [<-,anchor=center]%
            node (0,0) {#1}%
            (0,0.8\@SizeOfCirc) arc (85:-240:0.8\@SizeOfCirc);%
}%
\makeatother
\tikzset{
    set arrow inside/.code={\pgfqkeys{/tikz/arrow inside}{#1}},
    set arrow inside={end/.initial=>, opt/.initial=},
    /pgf/decoration/Mark/.style={
        mark/.expanded=at position #1 with
        {
            \noexpand\arrow[\pgfkeysvalueof{/tikz/arrow inside/opt}]{\pgfkeysvalueof{/tikz/arrow inside/end}}
        }
    },
    arrow inside/.style 2 args={
        set arrow inside={#1},
        postaction={
            decorate,decoration={
                markings,Mark/.list={#2}
            }
        }
    },
}


\date{April 14, 2023.}

\begin{document}

\begin{abstract}
In this paper, we recall some basic facts about the Kerr--Newman--(anti) de Sitter (KNdS) spacetime and review several formulations and integration methods for the geodesic equation of a test particle in such a spacetime. In particular, we introduce some basic general symplectic integrators in the Hamiltonian formalism and we re-derive the separated motion equations using Carter's method.

After this theoretical background, we explain how to ray-trace a KNdS black hole, equipped with a thin accretion disk, using Scilab. We compare the accuracy and execution time of the previous methods, concluding that the Carter equations is the best one. Then, inspired by Hagihara, we apply Weierstrass' elliptic functions to the non-rotating case, yielding a fairly fast shadowing program for such a spacetime.

We provide some illustrations of the code, including a depiction of the effects of the cosmological constant on shadows and accretion disk, as well as a simulation of M87*.
\end{abstract}

\maketitle

\subsection*{Copyright statement} This Accepted Manuscript is available for reuse under a CC BY-NC-ND licence after the 12 month embargo period provided that all the terms of the licence are adhered to. This is a peer-reviewed, un-copyedited version of an article published in Classical and Quantum Gravity. IOP Publishing Ltd is not responsible for any errors or omissions in this version of the manuscript or any version derived from it. The Published Version is available online at \url{https://doi.org/10.1088/1361-6382/accbfe}.
\vspace{1cm}

\section*{Introduction and motivation}

\indent The numerical computation of shadows and images of black holes and related relativistic objects is a crucial tool in understanding the effects of a strong (non-Newtonian) gravity field. This has been an extensive area of research for the last four decades, with significant progress in the last few years, due to an always increasing computational power and related observations of actual black holes, such as Sgr A* or M87* \cite{chargeconstraints,M87}.

The literature regarding the subject is quite extensive and many ray-tracing codes were produced, with various aspects: the appearance of a star orbiting a black hole \cite{luminet,cunningham-bardeen,perez-giz-levin}, images of accretion structures \cite{fanton-et-al,fuerst-wu,agol-dexter,karas-polnarev,marck,schnittman-krolik-hawley}, modelizations related to existing black holes \cite{broderick-loeb,M87}. Moreover, a lot of free codes is available \cite{agol-dexter,GRay,gyoto,PYYY}. See also \cite{shadow_scalar_hair,parametrised,osiris}.

Given so numerous and various works, why yet a new paper on the subject? We have three main reasons.

First, to the knowledge of the author, no ray-tracing code takes cosmological effects into account, that is, the assumption that the cosmological constant $\Lambda$ vanishes is always made. Moreover, the charge $Q$ of the black hole is also assumed to be zero. These are reasonable simplifications, since $\Lambda$ and $Q$ are expected to be negligible in the case of the observable black holes of our universe. Indeed, according to \cite[\S 3.2]{planck18}, the physical value of $\Lambda$ should be $\Lambda=(1.090\pm0.029)\cdot10^{-52}\rm{m}^{-2}$ in SI units and the charge should be small due to the plasma orbiting the object, see \cite[\S 6]{teukolsky}. The latter claim is confirmed in \cite[\S 4]{chargeconstraints} for Sgr A*. However, as pointed out in \cite{charge-plasma}, even a small charge could, in certain cases, have a great influence on electrons and thus on the plasma motion (provided that the electromagnetic field of the plasma is small). Moreover, to introduce a cosmological constant allows to visualise the properties of black holes in a faster-expanding (or even contracting) universe. We chose to add the charge term $Q$ for completeness and because it doesn't complicate the calculations too much, especially in comparison to the introduction of $\Lambda$. As an illustration of our code, the Figures \ref{lambda_on_shadows} and \ref{lambda_on_accretion} depict the visual influence of $\Lambda$ on shadows and accretion disks.

Furthermore, our code is freely available\footnote{at \url{https://github.com/arthur-garnier/knds_orbits_and_shadows.git}} and, again to the knowledge of the author, is the only black hole vizualising tool developed for \texttt{scilab}\footnote{Version 6.1.1, equipped with the \texttt{IPCV} package, see \url{https://www.scilab.org/} and \url{https://ipcv.scilab-academy.com/}}, a free software providing efficient routines for matrix manipulations and elementary image processing. This makes the code relatively transparent easy to explore and modify and makes the formulae of the paper easy to track in the code. We also designed the code in a way that the user may tune each parameter of the simulation, including the choice between the different redshifts to apply to the accretion disk, the brightness rescaling, etc. Moreover, a single geodesic tracer code is provided, so that the user may plot and compare various orbits, including that of a charged particle around a KNdS black hole. All this could make the codes useful for educational purposes. See \S\ref{rem_implem} and the documentation of the package for more details.

Finally, we wanted to derive and make explicit all the formulae involved in the process, so that the reader may easily create its own code out of them. Indeed, the statement of elementary formulae giving the motion constants in terms of prescribed initial conditions is rare in the literature (though \cite{PYYY} is an example). We tried to make the formulae as readable as possible, with conventions that are as close as possible from the existing references. For the convenience of the reader, the tedious proofs are put in Appendix \ref{some_proofs}.

The paper is organized as follows: after a reminder on Einstein's general theory of relativity, we introduce the KNdS metric as in \cite{gibbons-hawking} and re-prove in Theorem \ref{EFE_ok} that it maximally extends to an analytic metric satisfying the Maxwell-Einstein field equation.

Then, \S\ref{formulations} focuses on the geodesic equation of a (possibly charged) test particle in the KNdS spacetime. It also considers some of the formulations that can be used to numerically solve it, such as the Lagrangian and Hamiltonian formalisms. The latter is rather efficient, since it features some nice symplectic geometric properties. We then recall some classical general symplectic schemes which we implement. As we will later integrate the geometric equation \emph{backwards}, the symplectic schemes that are \emph{reversible} are of particular interest. However, we shall see that all of them will show some instabilities around the symmetry axis and moreover, these methods can be quite long to process as the stable ones are implicit. To get rid of this issue, we use the method from \cite{carter}.

Carter's method consists in identifying a fourth motion constant that makes the geodesic equation \emph{integrable}. We apply this method to our context in \S\ref{mot_con}. The resulting differential system is much simpler than the original one and can be solved quite easily using the routine \texttt{lsode} for Scilab \cite{hindmarsh}. For more details, see Theorem \ref{carter_equations} and Corollary \ref{with_conj_mom}. In Proposition \ref{find_constants}, we derive the motion constants from the \emph{rest mass} and the initial data of the geodesic.

In \S\ref{wp}, we treat the particular case of a non-rotating black hole. Following the original idea of \cite{hagihara}, we consider planar geodesics, parametrized in polar coordinates. In the case of a photon orbit in the \emph{Reissner--Nordstr\"{o}m--(anti) de Sitter black hole} (i.e. a non-rotating KNdS black hole), the geodesic equation can be reduced to the Weierstrass equation $\dot{\wp}^2=4\wp^3-g_2\wp-g_3$, whose solution is a \emph{Weierstrass elliptic function}; see Proposition \ref{into_weierstrass} and Corollary \ref{ini_conds_weierstrass}. Coupled with Carlson's algorithm for elliptic integrals (\cite{carlson}) and an elementary Newton approximation method, this provides an efficient way to shadow an RNdS black hole which is much faster than numerical integration of motion equations.

Then, we explain how we choose our model for the thin accretion disk, based on \cite{shakura-sunyaev} and \cite{spruit}. We assume that the matter in the accretion disk radiates as a blackbody and we use (a rescaled version of) Planck's law for the brightness. We also include the gravitational and Doppler redshift effects to the implementation. See \S\ref{mod_acc} for more details.

In \S\ref{rem_implem}, we make some remarks on the implementation process and provide details about the \emph{backward ray tracing} algorithm we use. We compare the different integration methods introduced earlier, regarding conservation of motion constants and execution times. Among others, we explain how the Weierstrass functions can be used to make an efficient program in the case of a non-rotating black hole. In the general case, the Carter equations are by far the best integration method. Among other illustrations, we display the effect of the cosmological constant on shadows and accretion disks in Figures \ref{lambda_on_shadows} and \ref{lambda_on_accretion}. We finish by giving a simulation of the M87 black hole in Figures \ref{m87_pict} and \ref{blurred}.

\section{The Kerr--Newman--(anti) de Sitter spacetime}\label{reminders}

\subsection{Reminders on Einstein's field equation and electromagnetic stress-energy tensor}
We start by recalling some very general facts and notation on Lorentzian manifolds and Einstein--Maxwell equations.

Consider a Lorentzian 4-manifold $(\mathcal{M},\mathbf{g})$ and let $\mathbf{R}$ be its Ricci tensor. Let $R:=\mathrm{tr}_{\mathbf{g}}(\mathbf{R})$ be the Ricci (scalar) curvature and $\mathbf{G}:=\mathbf{R}-\tfrac{1}{2}R\mathbf{g}$ be the associated \emph{Einstein tensor}. Then, the \emph{Einstein field equation (EFE)} is the following equality
\begin{equation}\label{intrin_EFE}%%\tag{$\dagger$}
\mathbf{G}+\Lambda\mathbf{g}=\kappa\mathbf{T},
\end{equation}
where $\mathbf{T}$ is a symmetric 2-tensor on $\mathcal{M}$, $\kappa:=8\pi G/c^4$ is the \emph{Einstein gravitational constant} and $\Lambda\in\R$ is called the \emph{cosmological constant}. In this case, notice that the Bianchi identity implies that the covariant derivative of $\mathbf{T}$ vanishes. If $(x^\mu)_{\mu=0,1,2,3}$ is a (local) coordinate frame on $\mathcal{M}$, then the (EFE) can be (locally) rewritten as
\begin{equation}\label{coor_EFE}
R_{\mu\nu}-\frac{1}{2}Rg_{\mu\nu}+\Lambda g_{\mu\nu}=\frac{8\pi G}{c^4}T_{\mu\nu},
\end{equation}
with $R=g^{\mu\nu}R_{\mu\nu}$ (using Einstein's summation convention), the matrix $(g^{\mu\nu})_{\mu,\nu}$ being the inverse of the Gram matrix $\mathrm{Mat}_{x^\mu}(\mathbf{g})=(g(\partial_{x^\mu},\partial_{x^\nu}))=:(g_{\mu\nu})$. To simplify the notation, we also denote partial derivatives (resp. covariant derivatives) using a comma (resp. a semicolon) low index. \textbf{In the following, we choose the signature $(-,+,+,+)$ for Lorentzian metrics and we use natural (Stoney) units where $G=c=4\pi\epsilon_0=1$.} Notice that this implies that $\mu_0=4\pi$.

Recall that given a metric $\mathbf{g}=(g_{\mu\nu})$, a divergence-free contravariant vector $\mathbf{J}=(J^\mu)$ (i.e. such that ${J^\mu}_{;\mu}:=\nabla_\mu J^\mu=0$) and a totally antisymmetric 2-tensor $\mathbf{F}=(F_{\mu\nu})$, seen as a differentiable 2-form $\mathbf{F}=\tfrac{1}{2}F_{\mu\nu}\mathrm{d}x^\mu\wedge\mathrm{d}x^\nu$, we say that $\mathbf{F}$ satisfies the \emph{covariant Maxwell equations} if
\begin{equation}\label{ME}\tag{ME}
\mathrm{d}\mathbf{F}=0=\mathrm{d}{}^{\ast}\!{\mathbf{F}}+\mu_0{}^{\ast}\!{\mathbf{J}},
%\left\{\begin{array}{l}
%\mathrm{d}\mathbf{F}=0, \\[.5em]
%\mathrm{d}{}^{\ast}\!{\mathbf{F}}=-\mu_0{}^{\ast}\!{\mathbf{J}},\end{array}\right.
\end{equation}
where ${}^{\ast}\!{(-)}$ denotes the Hodge dual. In this case the vector $J^\mu$ is called the \emph{current 1-form} and $\mathbf{F}$ is the \emph{electromagnetic field tensor}. We can translate these equations in coordinates:
\[\left\{\begin{array}{l}
%\partial_\lambda F_{\mu\nu}+\partial_\mu F_{\nu\lambda}+\partial_\nu F_{\lambda\mu}=0,\\[.5em]
F_{\mu\nu,\lambda}+F_{\nu\lambda,\mu}+F_{\lambda\mu,\nu}=0,\\[.5em]
{F^{\mu\nu}}_{;\mu}=-4\pi J^\nu.\end{array}\right.\]
Moreover, on a contractible open subset of $\mathcal{M}$, the Poincar\'e lemma ensures the existence of a 1-form $\mathbf{A}=A_\mu\mathrm{d}x^\mu$, called the \emph{electromagnetic vector potential}, such that $\mathbf{F}=\mathrm{d}\mathbf{A}$. In coordinates, this reads
\[F_{\mu\nu}=A_{\nu,\mu}-A_{\mu,\nu}=A_{\nu;\mu}-A_{\mu;\nu}.\]
Finally, the \emph{electromagnetic stress-energy tensor} $\mathbf{T}$ associated to the field $\mathbf{F}$ is given in local coordinates by\footnote{to be precise, this expression is valid only once a gauge where $A_\sigma J^\sigma=0$ has been chosen, but we don't need to be that subtle as we are interested only in vacuum solutions.}
\[T_{\mu\nu}=\frac{1}{\mu_0}\left(g^{\alpha\beta}F_{\alpha\mu}F_{\beta\nu}-\frac{1}{4}g_{\mu\nu}F_{\alpha\beta}F^{\alpha\beta}\right).\]
Then, the resulting EFE is called the \emph{Einstein--Maxwell equation (EME)} associated to $(\mathbf{g},\mathbf{J},\mathbf{F})$. In the case where $\mathbf{J}=0$, we call it the \emph{electro-vacuum} Einstein-Maxwell equation.

\subsection{The Kerr--Newman--(anti) de Sitter solution}
We now recall what the Kerr--Newman--de Sitter metric is. For more details, see \cite[\S 1.1]{shanjit}, \cite[\S 5, 6]{red-carter} or \cite[\S II]{gibbons-hawking}.

Consider the manifold $\mathcal{M}:=\R^2\times\Sph^2$, equipped with \emph{Boyer-Lindquist coordinates} $(t,r,\theta,\phi)$, where $(\theta,\phi)\in[0,\pi]\times[0,2\pi[$ describe spherical coordinates on $\Sph^2$. Fix four constants $\Lambda,M,Q,J\in\R\times\R_+^3$ and define $a:=J/M$ if $M\ne 0$ and $a:=J$ otherwise. Let $\lambda:=\Lambda/3$ and $\chi:=1+\lambda a^2$ (we assume $\chi\ne0$) and consider the following globally defined functions
\[\Sigma:=r^2+a^2\cos^2\theta,~~\Delta_r:=(1-\lambda r^2)(r^2+a^2)-2Mr+Q^2,~~\Delta_\theta:=1+\lambda a^2\cos^2\theta.\]
The \emph{Kerr--Newman--(anti)de Sitter (KNdS) metric} is the metric defined on the open subset $\{\Sigma\Delta_r\Delta_\theta\sin\theta\ne0\}$ by the line element
\begin{equation}\label{knds}\tag{KNdS}
\mathrm{d}s^2=-\frac{\Delta_r}{\chi^2\Sigma}(\mathrm{d}t-a\sin^2\theta\mathrm{d}\phi)^2+\frac{\Delta_\theta\sin^2\theta}{\chi^2\Sigma}(a\mathrm{d}t-(r^2+a^2)\mathrm{d}\phi)^2+\Sigma\left(\frac{\mathrm{d}r^2}{\Delta_r}+\frac{\mathrm{d}\theta^2}{\Delta_\theta}\right).
\end{equation}
It may be convenient to have the metric written in terms of matrices. Ordering the coordinates as $(t,r,\theta,\phi)$, we have
\[\mathbf{g}=\begin{pmatrix}\frac{a^2\sin^2\theta\Delta_\theta-\Delta_r}{\chi^2\Sigma} & 0 & 0 & \frac{a\sin^2\theta(\Delta_r-(r^2+a^2)\Delta_\theta)}{\chi^2\Sigma} \\ 0 & \frac{\Sigma}{\Delta_r} & 0 & 0 \\ 0 & 0 & \frac{\Sigma}{\Delta_\theta} & 0 \\ \frac{a\sin^2\theta(\Delta_r-(r^2+a^2)\Delta_\theta)}{\chi^2\Sigma} & 0 & 0 & \frac{\sin^2\theta((r^2+a^2)^2\Delta_\theta-a^2\sin^2\theta\Delta_r)}{\chi^2\Sigma}\end{pmatrix}\]
and
\[\mathbf{g}^{-1}=\begin{pmatrix}\frac{\chi^2(a^2\sin^2\theta\Delta_r-(r^2+a^2)^2\Delta_\theta)}{\Sigma\Delta_r\Delta_\theta} & 0 & 0 & \frac{a\chi^2(\Delta_r-(r^2+a^2)\Delta_\theta)}{\Sigma\Delta_r\Delta_\theta} \\ 0 & \frac{\Delta_r}{\Sigma} & 0 & 0 \\ 0 & 0 & \frac{\Delta_\theta}{\Sigma} & 0 \\ \frac{a\chi^2(\Delta_r-(r^2+a^2)\Delta_\theta)}{\Sigma\Delta_r\Delta_\theta} & 0 & 0 & \frac{\chi^2(\Delta_r-a^2\sin^2\theta\Delta_\theta)}{\Sigma\Delta_r\Delta_\theta\sin^2\theta}\end{pmatrix}.\]

The following result is well-known (see for instance \cite[\S 6]{red-carter} or \cite{BL}) and is recalled here for completeness (for a detailed proof, see Appendices \ref{some_proofs_1} and \ref{proofEFE}):
\begin{theo}\label{EFE_ok}
Assume that $\chi>0$ and consider the electromagnetic vector potential $\mathbf{A}=A_\mu\mathrm{d}x^\mu$ defined on the open submanifold $\mathcal{U}:=\mathcal{M}\setminus\{\Sigma=0\}$ by
\[\mathbf{A}=\frac{Qr}{\chi\Sigma}(\mathrm{d}t-a\sin^2\theta\mathrm{d}\phi).\]
Then the metric (\ref{knds}) maximally extends to a smooth Lorentzian metric on $\mathcal{U}$ and the electromagnetic field $\mathbf{F}:=\mathrm{d}\mathbf{A}$ verifies the associated vacuum Maxwell equations. Moreover, the KNdS metric solves the electro-vacuum Einstein--Maxwell equation on $\mathcal{U}$.
\end{theo}

\section{Several formulations and numerical schemes for the geodesic equation}\label{formulations}

Here, we first recall two of the main formulations of the geodesic equation namely, the Euler--Lagrange and Hamilton equations. Then, we review some of the general elementary symplectic integrators that can be used.

Throughout this section, we consider a geodesic $\gamma=(t,r,\theta,\phi)$ in the KNdS spacetime, corresponding to the trajectory of a test particle with rest mass $\mu\in\{-1,0\}$ and electric charge $e$. Recall that $\gamma$ satisfies the \emph{geodesic equation}
\begin{equation}\label{GE}
\ddot{\gamma}^\mu+{\Gamma^\mu}_{\alpha\beta}\dot{\gamma}^\alpha\dot{\gamma}^\beta=e{F^\mu}_{\alpha}\dot{\gamma}^\alpha,
\end{equation}
where ${\Gamma^\mu}_{\alpha\beta}=g^{\mu\nu}\Gamma_{\nu\alpha\beta}:=\tfrac{1}{2}g^{\mu\nu}(g_{\nu\beta,\alpha}+g_{\nu\alpha,\beta}-g_{\alpha\beta,\nu})$ are the Christoffel symbols and ${F^\mu}_{\alpha}=g^{\mu\nu}F_{\nu\alpha}$ is the electromagnetic tensor (in mixed form). We assume that $\gamma$ is a maximal solution of this equation, defined on an open interval $I\subset\R$, say, with affine parameter $\ell\in I$ (the dot of course represents the derivative with respect to the affine parameter).

\subsection{Lagrangian and Hamiltonian formalisms}
Consider the relativistic \emph{Lagrangian} $\mathcal{L} : T\mathcal{M}\to\R$, defined by
\[\mathcal{L}(\gamma,\dot{\gamma}):=\tfrac{1}{2}g_{\mu\nu}\dot{\gamma}^\mu\dot{\gamma}^\nu+eA_\mu\dot{\gamma}^\mu,\]
as well as the related \emph{action integral}
\[S:=\int \mathcal{L}(\gamma,\dot{\gamma})\mathrm{d}\ell,\]
where we integrate on a compact sub-interval of $I$. Hamilton's principle asserts that $\gamma$ is a stationary point
of the action $S$, and this is equivalent to the \emph{Euler--Lagrange equation}
\begin{equation}\label{EL}
\frac{\mathrm{d}}{\mathrm{d}\ell}\left(\frac{\partial\mathcal{L}}{\partial\dot{\gamma}}\right)=\frac{\partial\mathcal{L}}{\partial\gamma}.
\end{equation}
Developing, we find that for all $\mu\in\{0,1,2,3\}$,
\[g_{\mu\nu}\ddot{\gamma}^\nu+g_{\mu\alpha,\beta}\dot{\gamma}^\alpha\dot{\gamma}^\beta+eA_{\mu,\alpha}\dot{\gamma}^\alpha=\frac{\mathrm{d}}{\mathrm{d}\ell}\left(g_{\mu\alpha}\dot{\gamma}^\alpha+eA_\mu\right)=\frac{1}{2}g_{\alpha\beta,\mu}\dot{\gamma}^\alpha\dot{\gamma}^\beta+eA_{\alpha,\mu}\dot{\gamma}^\alpha\] 
and rearranging this yields
\[g_{\mu\nu}\ddot{\gamma}^\nu+\frac{1}{2}(2g_{\mu\alpha,\beta}-g_{\alpha\beta,\mu})\dot{\gamma}^\alpha\dot{\gamma}^\beta+e(A_{\mu,\alpha}-A_{\alpha,\mu})\dot{\gamma}^\alpha=0,\]
or, equivalently,
\begin{equation}\label{ELe}
\ddot{\gamma}^\mu+{}^{'}{\Gamma^\mu}_{\alpha\beta}\dot{\gamma}^\alpha\dot{\gamma}^\beta-e{F^\mu}_\alpha\dot{\gamma}^\alpha=0,
\end{equation}
where
\[{}^{'}{\Gamma^\mu}_{\alpha\beta}=g^{\mu\nu}\left(g_{\alpha\nu,\beta}-\tfrac{1}{2}g_{\alpha\beta,\nu}\right).\]
This is indeed equivalent to (\ref{GE}) since the difference ${}^{'}{\Gamma^\mu}_{\alpha\beta}-{\Gamma^\mu}_{\alpha\beta}=\tfrac{1}{2}g^{\mu\nu}(g_{\nu\alpha,\beta}-g_{\nu\beta,\alpha})$ is anti-symmetric in the indices $\alpha$ and $\beta$. However, we implement\footnote{For the Euler--Lagrange equation and Hamiltonian methods, we assume that $e=0$ for simplicity.} the geodesic equation in Euler--Lagrange form, as it requires a bit less heavy calculations than the genuine Christoffel symbols. To solve the equations (\ref{ELe}), we simply use the internal solver from Scilab that implements Adams methods (see \cite{hindmarsh}).

Instead of the Lagrangian, one may look at the \emph{Hamiltonian}. First, we introduce the \emph{conjugate momenta}:
\[p_\mu:=g_{\mu\nu}\dot{\gamma}^\nu+eA_\mu.\]
The Hamiltonian $\mathcal{H} : T^*\mathcal{M}\to\R$ is then defined as the Legendre transform of $\mathcal{L}$, namely
\[\mathcal{H}(\gamma,p):=\tfrac{1}{2}g^{\mu\nu}(p_\mu-eA_\mu)(p_\nu-eA_\nu)=p_\mu\dot{\gamma}^\mu-\mathcal{L}(\gamma,\dot{\gamma}).\]
Then, the Euler--Lagrange equation is equivalent to \emph{Hamilton's equations}
\begin{equation}\label{H}
\left\{\begin{array}{l}
\dfrac{\mathrm{d}\gamma}{\mathrm{d}\ell}=\dfrac{\partial\mathcal{H}}{\partial p},\\[1em]
\dfrac{\mathrm{d}p}{\mathrm{d}\ell}=-\dfrac{\partial\mathcal{H}}{\partial\gamma}.\end{array}\right.
\end{equation}
Unravelling this, we obtain the following system of order 1
\begin{equation}\label{He}
\left\{\begin{array}{l}
\dot{\gamma}^\mu=g^{\mu\alpha}(p_\alpha-eA_\alpha),\\[.5em]
\dot{p}_\mu=\frac{e}{2}g^{\alpha\beta}(A_{\alpha,\mu}(p_\alpha-eA_\alpha)+A_{\beta,\mu}(p_\beta-eA_\beta))-\frac{1}{2}{g^{\alpha\beta}}_{,\mu}(p_\alpha-eA_\alpha)(p_\beta-eA_\beta).\end{array}\right.
\end{equation}
In the case of a particle without charge ($e=0$), this reduces to
\[\left\{\begin{array}{l}
\dot{\gamma}^\mu=g^{\mu\alpha}p_\alpha,\\[.5em]
\dot{p}_\mu=-\frac{1}{2}{g^{\alpha\beta}}_{,\mu}p_\alpha p_\beta.\end{array}\right.\]
As we shall see in the comparison section, the equations are a bit faster to integrate (with the Adams solver from \cite{hindmarsh}) than the Euler--Lagrange ones. Moreover, they are more efficient in preserving the Hamiltonian.

\subsection{Symplectic schemes for Hamilton's equations}\label{symp_schemes}
In view of integrating the system (\ref{He}), we may use general algorithms that apply to any Hamiltonian $\mathcal{H} : T^*\mathcal{X}\to\R$, called \emph{symplectic integrators}. A detailed exposition can be found in \cite{feng-qin} and \cite{hairer}. See also \cite{sanz-serna-calvo}.

First, we remind some basics of symplectic geometry (see \cite[\S 3.1]{feng-qin}). If $q=(q^1,\dotsc,q^N)$ are local coordinates on an $N$-manifold $\mathcal{X}$ and $p=(p_1,\dotsc,p_N)$ the associated coordinates on $T^*_q\mathcal{X}$, then $(q,p)$ are local coordinates on $T^*\mathcal{X}$ and we may define a symplectic form on it:
\[\omega:=\mathrm{d}p\wedge\mathrm{d}q=\mathrm{d}p_i\wedge\mathrm{d}q^i.\]
If $\mathcal{H} : T^*\mathcal{X}\to\R$ is a smooth function, then there exists a vector field $X_{\mathcal{H}}\in\Gamma(T(T^*\mathcal{X}))$ on $T^*\mathcal{X}$ such that $\omega(X_{\mathcal{H}},-)=\mathrm{d}\mathcal{H}$. Then, given $(q,p)\in T^*\mathcal{X}$, there is a unique maximal curve $\gamma_{q,p} : ]-\varepsilon,\varepsilon[\to T^*\mathcal{X}$ such that
\[\left\{\begin{array}{l}
\gamma_{q,p}(0)=(q,p),\\[.5em]
\gamma_{q,p}'=X_{\mathcal{H}}\circ\gamma_{q,p}.\end{array}\right.\]
Then, the \emph{Hamiltonian flow} $\Phi_s$ is defined as $\Phi_s(q,p):=\gamma_{q,p}(s)$, when this makes sense. Citing \cite[\S 3.2.1, Theorem 2.4]{feng-qin}, this flow is \emph{symplectic}, meaning that the pull-back $\Phi_s^*\omega=\omega$. In other words, if $\Phi'_s(q,p)$ denotes the Jacobian $(\partial\Phi_s/\partial(q,p))$ of $\Phi_s$, then we have
\[{}^{t}{\Phi'_s(q,p)}\cdot J\cdot \Phi'_s(q,p)=J,~~\text{where}~~J:=\begin{pmatrix}0 & I_n \\ -I_n & 0\end{pmatrix}\]
Roughly, this means that Hamilton's equations (\ref{H}) (or rather the flow of $\mathcal{H}$) preserves the symplectic structure on $T^*\mathcal{X}$. As we would like to solve the system numerically, it would be nice to have schemes that also preserve this geometric structure.

Consider a smooth curve $\xi : s\mapsto \xi(s)=(q(s),p(s))$ satisfying Hamilton's equations
\begin{equation}\label{GH}
\left\{\begin{array}{l}
\dot{q}=\partial_p\mathcal{H}(q,p),\\[.5em]
\dot{p}=-\partial_q\mathcal{H}(q,p).\end{array}\right.
\end{equation}
A one-step numerical scheme with step $h\ne0$ can be represented by its \emph{numerical flow} $\Phi_h : (q_n,p_n)\mapsto(q_{n+1},p_{n+1})$. As for the Hamiltonian, this flow reflects the geometric properties of the scheme.

\begin{definition}\label{scheme_flows}
\begin{enumerate}
Define the involution $\psi : (q,p)\mapsto(q,-p)$ on $T^*\mathcal{X}$ and consider a numerical scheme with flow $\Phi_h : (q_n,p_n)\to(q_{n+1},p_{n+1})$.
\item The Hamiltonian $\mathcal{H}$ is said to be \emph{time-reversible} if its flow $\Phi_s$ satisfies
\[\psi\circ\Phi_s\circ\psi=\Phi_{-s}.\]
In other words, this means that $(\widehat{q},\widehat{p})=\Phi_s(q,p)$ iff $\Phi_s(\widehat{q},-\widehat{p})=(q,-p)$.
\item Similarly, if $\mathcal{H}$ is time-reversible, then the scheme is \emph{reversible} if its flow satisfies
\[\psi\circ\Phi_h\circ\psi=\Phi_{-h}.\]
\item The scheme is \emph{symmetric} if we have $\Phi_h^{-1}=\Phi_{-h}$.
\item Finally, the scheme is \emph{symplectic} if its flow is, i.e. if
\[{}^{t}{\Phi_h'(q,p)}\cdot J\cdot \Phi_h'(q,p)=J.\]
\end{enumerate}
\end{definition}

\begin{rem}
To say that $\mathcal{H}$ is reversible is equivalent to the following conditions
\[\partial_p\mathcal{H}(q,-p)=-\partial_p\mathcal{H}(q,p)~~\text{and}~~\partial_q\mathcal{H}(q,-p)=\partial_q\mathcal{H}(q,p).\]
From this we see that, for instance, the Hamiltonian of an uncharged particle in the KNdS spacetime is reversible.
\end{rem}

We now give the symplectic schemes we have implemented. As is well-known, explicit schemes are unstable and the approximations they produce may blow-up, especially with problems like our geodesic one, where some (coordinate) singularities appear in the metric. However, the (velocity-)Verlet is a relatively good explicit alternative for our setting. With that being said, it turns out that all the schemes we present here do blow-up near the axis of rotation $\{\sin\theta=0\}\subset\mathcal{M}$.

The simplest methods are the \emph{semi-implicit Euler schemes}. These are given as follows:
\begin{algorithm}
\caption{$q$-implicit Euler scheme}\label{SEq}
\begin{algorithmic}[1]
\Require $h>0$, $(q_0,p_0)$
\For{$n=0,\dotsc,$}
	\State $q_{n+1}=q_n+h\partial_p\mathcal{H}(q_{n+1},p_n)$
	\vspace{2mm}
	\State $p_{n+1}=p_n-h\partial_q\mathcal{H}(q_{n+1},p_n)$
\EndFor
\end{algorithmic}
\end{algorithm}
\begin{algorithm}
\caption{$p$-implicit Euler scheme}\label{SEp}
\begin{algorithmic}[1]
\Require $h>0$, $(q_0,p_0)$
\For{$n=0,\dotsc,$}
	\State $p_{n+1}=p_n-h\partial_q\mathcal{H}(q_n,p_{n+1})$
	\vspace{2mm}
	\State $q_{n+1}=q_n+h\partial_p\mathcal{H}(q_n,p_{n+1})$
\EndFor
\end{algorithmic}
\end{algorithm}
As we shall see later, the $p$-implicit method is roughly twice as fast as the $q$-implicit one in our setting. This comes from the fact that our (uncharged) Hamiltonian $\mathcal{H}=\tfrac{1}{2}g^{\mu\nu}(q)p_\mu p_\nu$ is way easier to differentiate with respect to $p$ (it is quadratic in $p$) than with respect to $q$ and thus the equation $p_{n+1}=p_n-h\partial_q\mathcal{H}(q_n,p_{n+1})$ is more easily solved than the equation $q_{n+1}=q_n+h\partial_p\mathcal{H}(q_{n+1},p_n)$.

A relatively strong explicit method is the \emph{velocity Verlet} (or Verlet--leapfrog) scheme. As in \cite[\S 3.3]{bou-rabee-sanz-serna}, the scheme with step size $h$ is written in Algorithm \ref{VV}.
\begin{algorithm}
\caption{Velocity Verlet scheme}\label{VV}
\begin{algorithmic}[1]
\Require $h>0$, $(q_0,p_0)$
\For{$n=0,\dotsc,$}
	\State $p_{n+\frac{1}{2}}=p_n-\frac{h}{2}\partial_q\mathcal{H}(q_n,p_n)$
	\vspace{2mm}
	\State $q_{n+1}=q_n+h\partial_p\mathcal{H}\left(q_n,p_{n+\frac{1}{2}}\right)$
	\vspace{2mm}
	\State $p_{n+1}=p_{n+\frac{1}{2}}-\frac{h}{2}\partial_q\mathcal{H}(q_{n+1},p_n)$
\EndFor
\end{algorithmic}
\end{algorithm}

Following \cite[\S 1.8, (1.25)]{hairer}, a more stable method is the \emph{St\"ormer--Verlet scheme}, detailed in Algorithm \ref{SV} (there's a dual version of it, roughly by exchanging $q$ and $p$ and the signs accordingly).
\begin{algorithm}[h!]
\caption{St\"ormer--Verlet scheme}\label{SV}
\begin{algorithmic}[1]
\Require $h>0$, $(q_0,p_0)$
\For{$n=0,\dotsc,$}
	\State $q_{n+\frac{1}{2}}=q_n+\frac{h}{2}\partial_p\mathcal{H}\left(q_{n+\frac{1}{2}},p_n\right)$
	\vspace{2mm}
	\State $p_{n+1}=p_n-\frac{h}{2}\left(\partial_q\mathcal{H}\left(q_{n+\frac{1}{2}},p_n\right)+\partial_q\mathcal{H}\left(q_{n+\frac{1}{2}},p_{n+1}\right)\right)$
	\vspace{2mm}
	\State $q_{n+1}=q_{n+\frac{1}{2}}+\frac{h}{2}\partial_p\mathcal{H}\left(q_{n+\frac{1}{2}},p_{n+1}\right)$
\EndFor
\end{algorithmic}
\end{algorithm}

Because of its stability, this is the most efficient method, but it requires much more time to numerically solve the implicit equation for $q_{n+1/2}$.

\normalsize{We} may summarize the properties of the above schemes in the following result:
\begin{theo}[\cite{hairer}, \cite{grmonty}]\label{schemes_props}
The Euler schemes are of order 1 and symplectic but not symmetric (inverting the flow exchanges the two schemes) and not reversible (time-reversion takes each one to its explicit analogue).

The Verlet scheme is symplectic, reversible, symmetric and of order 2. 

Finally, the St\"ormer--Verlet scheme is symplectic, reversible, symmetric and of order 2 as well, but it is also stable.
\end{theo}

\section{Motion constants and Carter's equations}\label{mot_con}

In this section, we take advantage of the form of the metric (in Boyer--Lindquist coordinates) and apply Carter's method \cite{carter} to derive the motion equations in the KNdS spacetime. More precisely, the Hamilton--Jacobi equation is separable and yields four constants of motion, making the geodesic equations separable. Then, we explain how to find the four constants from genuine initial conditions.

\subsection{Motion equations}
Consider the trajectory of charged particle, with electric charge $e\in\R$, and let $\gamma$ be the corresponding (time-like or light-like) geodesic, defined on an open interval $0\in I\subset\R$ with affine parameter $\ell\in I$ and assume $\gamma$ has values in $\{\Sigma\Delta_r\sin\theta\ne0\}$. Recall the Hamiltonian
\[\frac{\mu}{2}:=\mathcal{H}(\gamma,p)=\frac{1}{2}g^{\mu\nu}(p_\mu-eA_\mu)(p_\nu-eA_\nu)\]
which is constant along $\gamma$ and equals $-\tfrac{1}{2}m^2$, where $m$ is the rest mass of the particle\footnote{$m=0$ for a photon}. Also, as $\partial_t$ and $\partial_\phi$ are Killing vectors, the \emph{total energy} $E:=-p_t$ and the \emph{total (azimuthal) angular momentum} $L:=p_\phi$ are constant along $\gamma$ too. It turns out that there is a fourth constant $\kappa$, called the \emph{Carter constant}, which allows to write the geodesic equations in a separable form. This is the point of the following well-known result, the formulation and proof (Appendix \ref{some_proofs_2}) of which are inspired by \cite{balek}, \cite{heisnam} and \cite{shanjit}:

\begin{theo}\label{carter_equations}
Given a geodesic $\gamma$ as above, define the following functions on $I$:
\[W_r:=\chi(E(r^2+a^2)-aL)+eQr~~\text{and}~~W_\theta:=\chi(aE\sin\theta-L/\sin\theta).\]
Then, the quantity
\[\kappa:=\Delta_\theta p_\theta^2+\frac{W_\theta^2}{\Delta_\theta}-\mu a^2\cos^2\theta=-\Delta_rp_r^2+\frac{W_r^2}{\Delta_r}+\mu r^2\]
is constant along $\gamma$ and moreover, $\gamma=(t,r,\theta,\phi)$ satisfies the following differential system on $I$:
\begin{equation}\label{motion}
\left\{\begin{array}{l}
\dfrac{\Sigma}{\chi}\dot{t}=\dfrac{W_r(r^2+a^2)}{\Delta_r}-\dfrac{aW_\theta\sin\theta}{\Delta_\theta},\\[1.5em]
\Sigma^2\dot{r}^2=W_r^2-\Delta_r(\kappa-\mu r^2),\\[1em]
\Sigma^2\dot{\theta}^2=-W_\theta^2+\Delta_\theta(\kappa+\mu a^2\cos^2\theta),\\[1em]
\dfrac{\Sigma}{\chi}\dot{\phi}=\dfrac{aW_r}{\Delta_r}-\dfrac{W_\theta}{\Delta_\theta\sin\theta}.\end{array}\right.
\end{equation}
\end{theo}

The set of equations (\ref{motion}) is unusable in numerical computations due to the squares in the equations for $\dot{r}$ and $\dot{\theta}$. Indeed, at \emph{turning points} (points where the sign of $\dot{r}$ or $\dot{\theta}$ changes), we cannot choose what sign to put in front of the square root when these get smaller and smaller. We get rid of this difficulty using the method of \cite{fuerst-wu} (see also \cite{PYYY}) and derivate the equations for $\dot{r}^2$ and $\dot{\theta}^2$ again. It turns out the formulation is more elegant when dealing with the derivate conjugate momenta $\dot{p_r}$ and $\dot{p_\theta}$ rather that with $\ddot{r}$ and $\ddot{\theta}$.

\begin{cor}\label{with_conj_mom}
With the same notation as in Theorem \ref{carter_equations}, the geodesic $\gamma$ with motion constants $(\mu,E,L,\kappa)$ satisfies the following first order autonomous differential system with variables $(t,r,p_r,\theta,p_\theta,\phi)$:
\begin{equation}\label{motion_with_momenta}
\left\{\begin{array}{l}
\dfrac{\Sigma}{\chi}\dot{t}=\dfrac{W_r(r^2+a^2)}{\Delta_r}-\dfrac{aW_\theta\sin\theta}{\Delta_\theta},\\[1.5em]
\Sigma\dot{r}=\Delta_r p_r,\\[1em]
\Sigma\dot{p_r}=\dfrac{\frac{\partial W_r^2}{\partial r}-\Delta_r'(\kappa-\mu r^2)}{2\Delta_r}+\mu r-\Delta_r'p_r^2,\\[1.5em]
\Sigma\dot{\theta}=\Delta_\theta p_\theta,\\[1em]
\Sigma\dot{p_\theta}=\dfrac{-\frac{\partial W_\theta^2}{\partial \theta}+\Delta_\theta'(\kappa+\mu a^2\cos^2\theta)}{2\Delta_\theta}-\mu a^2\cos\theta\sin\theta-\Delta_\theta'p_\theta^2,\\[1.5em]
\dfrac{\Sigma}{\chi}\dot{\phi}=\dfrac{aW_r}{\Delta_r}-\dfrac{W_\theta}{\Delta_\theta\sin\theta},\end{array}\right.
\end{equation}
where, of course, for $\nu=r,\theta$, the symbol $\Delta_\nu'$ means ${\partial\Delta_\nu}/{\partial\nu}$.
\end{cor}
\begin{proof}
We only carry the calculations out for $\dot{p_r}$, the case of $\dot{p_\theta}$ being similar. Define $f(r):=\tfrac{W_r^2}{\Delta_r^2}-\tfrac{\kappa-\mu r^2}{\Delta_r}$ so that the second equation from (\ref{motion}) reads $p_r^2=f(r)$ and differentiating this equation with respect to $\ell$ gives
\begin{align*}
2p_r\dot{p_r}=\frac{\partial f}{\partial r}\frac{\mathrm{d}r}{\mathrm{d}\ell}&~\Longleftrightarrow~\frac{2\Sigma\dot{p_r}}{\Delta_r}=\frac{\partial f}{\partial r}=2\frac{W_r(W_r'\Delta_r-W_r\Delta_r')}{\Delta_r^3}-\frac{-2\mu r\Delta_r-(\kappa-\mu r^2)\Delta_r'}{\Delta_r^2}\\
&~\Longleftrightarrow~\Sigma\dot{p_r}=\frac{W_r(W_r'\Delta_r-W_r\Delta_r')}{\Delta_r^2}+\mu r+(\kappa-\mu r^2)\frac{\Delta_r'}{2\Delta_r} \\
&~\Longleftrightarrow~\Sigma\dot{p_r}=\frac{2W_r'W_r-\Delta_r'(\kappa-\mu r^2)}{2\Delta_r}+\mu r+\frac{\Delta_r'}{\Delta_r}\left(\kappa-\mu r^2-\frac{W_r^2}{\Delta_r}\right)\\
&~\Longleftrightarrow~\Sigma\dot{p_r}=\frac{\partial_r(W_r^2)-\Delta_r'(\kappa-\mu r^2)}{2\Delta_r}+\mu r-\Delta_r'p_r^2.
\end{align*}
\end{proof}

\subsection{Expressions for the motion constants}
In order to implement the set of equations (\ref{motion_with_momenta}), we need to find the constants $(\mu,E,L,\kappa)$ from initial values for the geodesic $\gamma$. We have the following result:

\begin{prop}\label{find_constants}
Given a geodesic $\gamma=(t,r,\theta,\phi)$ as in Theorem \ref{carter_equations}, the energy, angular momentum and Carter's constant are given as follows:
\[\left\{\begin{array}{l}
E=-\dfrac{eQr}{\chi\Sigma}+\dfrac{1}{\chi}\sqrt{(a^2\sin^2\theta\Delta_\theta-\Delta_r)\left(\dfrac{\mu}{\Sigma}-\dfrac{\dot{r}^2}{\Delta_r}-\dfrac{\dot{\theta}^2}{\Delta_\theta}\right)+\dfrac{\dot{\phi}^2\sin^2\theta\Delta_r\Delta_\theta}{\chi^2}},\\[1.5em]
L=\dfrac{\sin^2\theta}{\chi^2(\Delta_r-a^2\sin^2\theta\Delta_\theta)}\left[aE\chi^2\Delta_r+\Delta_\theta\left(\Sigma\Delta_r\dot{\phi}-a\chi(\chi E(r^2+a^2)+eQr)\right)\right],\\[1.5em]
\kappa=\dfrac{W_\theta^2+\Sigma^2\dot{\theta}^2}{\Delta_\theta}-\mu a^2\cos^2\theta=\dfrac{W_r^2-\Sigma^2\dot{r}^2}{\Delta_r}+\mu r^2,\end{array}\right.\]
where $\mu=-1$ for a massive test particle and $\mu=0$ for a photon.
\end{prop}
\begin{proof}
The expressions for $\kappa$ are straightforwardly obtained from those in Theorem \ref{carter_equations}. To compute $L$, we simply invert the azimuthal equation from the system (\ref{motion}). We write
\[\frac{\Sigma}{\chi}\dot{\phi}=\frac{aW_r}{\Delta_r}-\frac{W_\theta}{\Delta_\theta\sin\theta}=\chi L\left(\frac{1}{\Delta_\theta\sin^2\theta}-\frac{a^2}{\Delta_r}\right)-\frac{a\chi E}{\Delta_\theta}+\frac{a(\chi E(r^2+a^2)+eQr)}{\Delta_r}\]
so that, multiplying both sides by $\sin^2\theta\Delta_r\Delta_\theta$ yields
\[\chi L(\Delta_r-a^2\sin^2\theta\Delta_\theta)=\frac{\sin^2\theta}{\chi}\left[aE\chi^2\Delta_r+\Delta_\theta\left(\Sigma\Delta_r\dot{\phi}-a\chi(\chi E(r^2+a^2)+eQr\right)\right],\]
as claimed. Now for the energy, it is determined by $E\ge0$ and the fact that $2\mathcal{H}(\gamma,p)\equiv \mu$. Recalling the equation (\ref{hamil}) and using the above expression for $L$, we compute
\[2\mathcal{H}=\mu~\Longleftrightarrow~\Sigma(\Delta_r-a^2\sin^2\theta\Delta_\theta)(\mu-2\mathcal{H})=0~\Longleftrightarrow~\alpha_2 E^2+\alpha_1 E+\alpha_0=0,\]
where
\[\alpha_2=\chi^2\Sigma^2,~\alpha_1=2\chi eQr\Sigma,~\alpha_0=(\Delta_r-a^2\sin^2\theta\Delta_\theta)\left(\mu\Sigma-\frac{\Sigma^2\dot{r}^2}{\Delta_r}-\frac{\Sigma^2\dot{\theta}^2}{\Delta_\theta}\right)-\frac{\Sigma^2\sin^2\theta\Delta_r\Delta_\theta\dot{\phi}^2}{\chi^2}+e^2Q^2r^2.\]
Therefore, the positive solution $E$ of $\alpha_i E^i=0$ reads
\[E=-\frac{\alpha_1}{2\alpha_2}+\sqrt{\frac{\alpha_1^2}{4\alpha_2^2}-\frac{\alpha_0}{\alpha_2}}=-\frac{eQr}{\chi\Sigma}+\sqrt{\frac{a^2\sin^2\theta\Delta_\theta-\Delta_r}{\chi^2\Sigma}\left(\mu-\frac{\Sigma\dot{r}^2}{\Delta_r}-\frac{\Sigma\dot{\theta}^2}{\Delta_\theta}\right)+\frac{\sin^2\theta\Delta_r\Delta_\theta\dot{\phi}^2}{\chi^4}},\]
and this is exactly the stated formula.
\end{proof}

\begin{rem}
From the set of equations (\ref{motion}), we see that trajectories $\gamma=(t,r,\theta,\phi)$ for which $\theta(\ell_0)=\pi/2$ and $\dot{\theta}(\ell_0)=0$ for some $\ell_0\in I$ are confined in the equatorial plane $\theta=\pi/2$. In this case, Carter's constant reduces to $\kappa=\chi^2(aE-L)^2$. Therefore, Carter's constant sometimes refers rather to the constant $C:=\kappa-\chi^2(aE-L)^2$ so that $C=0$ for orbits in the plane $\theta=\pi/2$. More explicitly, the constant $C$ can be written as
\[C=\Delta_\theta p_\theta^2+\frac{\chi^2\cos^2\theta}{\Delta_\theta}\left[\frac{L^2}{\sin^2\theta}-a^2\left(E^2+\frac{\mu\Delta_\theta}{\chi^2}+{3\lambda^2}(aE-L)^2\right)\right].\]
Notice that this expression agrees with the one from \cite[\S 2.1]{PYYY} when $\lambda\to0$.
\end{rem}

\section{Polar formulation for RNdS trajectories and the Weierstrass elliptic function}\label{wp}

In this entire section, we assume that $a=0$, that is, we work with the \emph{Reissner--Nordstr\"{o}m-(anti) de Sitter (RNdS) metric} which is given, in Boyer--Lindquist (spherical) coordinates by
\begin{equation}\label{RNdS}\tag{RNdS}
\mathrm{d}s^2=-\underline{\Delta}\mathrm{d}t^2+\frac{\mathrm{d}r^2}{\underline{\Delta}}+r^2(\mathrm{d}\theta^2+\sin^2\theta\mathrm{d}\phi^2)
%\mathrm{d}s^2=-\left(1-\lambda r^2-\frac{2M}{r}+\frac{Q^2}{r^2}\right)\mathrm{d}t^2+\frac{\mathrm{d}r^2}{1-\lambda r^2-\frac{2M}{r}+\frac{Q^2}{r^2}}+r^2(\mathrm{d}\theta^2+\sin^2\theta\mathrm{d}\phi^2)
%-\frac{\Delta_r}{r^2}\mathrm{d}t^2+\frac{r^2}{\Delta_r}\mathrm{d}r^2+r^2(\mathrm{d}\theta^2+\sin^2\theta\mathrm{d}\phi^2)=
\end{equation}
where we let $\underline{\Delta}:=\Delta_r/r^2=1-\lambda r^2-2M/r+Q^2/r^2$ to lighten the notation. Since this metric is spherically symmetric, the geodesics are planar. Therefore, in order to study geodesics (and to implement them afterwards), we only need to focus on the equatorial ones. More precisely, if we have any geodesic, we may apply a linear rotation (i.e. an element of $1\times SO(3)\subset \mathrm{Isom}(\mathcal{M}\setminus\{r=0\},g_{\rm RNdS})$) to force its velocity vector to lie on the equatorial plane, solve the equations and then go back with the inverse rotation.

\subsection{Polar geodesic equation}
Consider then an equatorial geodesic $\gamma=(t,r,\pi/2,\phi)$ with Hamiltonian $\mu$, energy $E$ and angular momentum $L$.
The set of equations (\ref{motion}) becomes
\begin{equation}\label{motion_RNdS}
\left\{\begin{array}{l}
\underline{\Delta}\dot{t}=E,\\[.5em]
r^4\dot{r}^2=(Er^2+eQr)^2-\Delta_r(L^2-\mu r^2),\\[.5em]
r^2\dot{\phi}=L.\end{array}\right.
\end{equation}
From this we see that if $\dot{\phi}$ evaluates to zero somewhere, then $L=0$ and $\dot{\phi}\equiv0$ and the motion is then radial. Suppose it is not the case, then ${\phi}$ is a diffeomorphism onto its image and we may express $r=r(\phi)$ as a function of $\phi$. We write
\[\left(\frac{\mathrm{d}r}{\mathrm{d}\phi}\right)^2=\left(\frac{\dot{r}}{\dot{\phi}}\right)^2=\left(\frac{r^2\dot{r}^2}{r^2\dot{\phi}}\right)^2=\frac{r^4\dot{r}^2}{L^2}=\frac{(Er+eQ)^2}{L^2}r^2-\Delta_r\left(1-\frac{\mu}{L^2}r^2\right)\]
and after calculations,
\begin{equation}\label{pre-pre-weierstrass}
\left(\frac{\mathrm{d}r}{\mathrm{d}\phi}\right)^2=-\frac{\lambda\mu}{L^2}r^6+\left(\lambda+\frac{E^2+\mu}{L^2}\right)r^4+\frac{2}{L^2}(EeQ-M\mu)r^3+\left(\frac{Q^2}{L^2}(e^2+\mu)-1\right)r^2+2Mr-Q^2.
\end{equation}
Now, considering the Binet variable $u:=1/r$, we obtain the equation (from now on, the dot means differentiation with respect to $\phi$)
\begin{equation}\label{polar_binet}
\dot{u}^2=\frac{\dot{r}^2}{r^4}=-\frac{\lambda\mu}{L^2u^2}+\left(\lambda+\frac{E^2+\mu}{L^2}\right)+\frac{2}{L^2}(EeQ-M\mu)u+\left(\frac{Q^2}{L^2}(e^2+\mu)-1\right)u^2+2Mu^3-Q^2u^4.
\end{equation}
Finally, we can get rid of the square by differentiating again. We find
\begin{equation}\label{polar_binet_2}
\ddot{u}=\frac{\lambda\mu}{L^2u^3}+\frac{EeQ-M\mu}{L^2}+\left(\frac{Q^2}{L^2}(e^2+\mu)-1\right)u+3Mu^2-2Q^2u^3
\end{equation}
and this equation is much easier to (numerically) solve than the system (\ref{motion_with_momenta}).

\subsection{Use of Weierstrass' function $\wp$ for photon orbits}
The striking observation that the Weierstrass elliptic function $\wp$ solves the polar equatorial motion equation was first made by Hagihara in \cite{hagihara}. Here, inspired by the method from \cite[\S 3.1]{gibbons-vyska}, we show that we can still use the function $\wp$ to describe null geodesics in the RNdS metric.

In the case of a photon (whose world-line is a null geodesic with $\mu=0$ and $e=0$), the equation (\ref{pre-pre-weierstrass}) reduces to
\begin{equation}\label{pre-weiertrass}
%\dot{u}^2=\lambda+\frac{E^2}{L^2}-u^2+2Mu^3-Q^2u^4
\dot{r}^2=\left(\lambda+\frac{E^2}{L^2}\right)r^4-r^2+2Mr-Q^2
\end{equation}
This equation can be further reduced to the Weierstrass equation $\dot{y}^2=4y^3-g_2y-g_3$ as follows: suppose that $\lambda L^2+E^2\ge0$, then the depressed quartic $(\lambda+E^2/L^2)x^4-x^2+2Mx-Q^2$ has a real root\footnote{In practice, we choose $\overline{r}$ with minimal norm so that $\overline{r}=0$ when $Q=0$.} $\overline{r}\in \R$ and let $\widetilde{r}:=r-\overline{r}$. We have
\begin{align*}
\dot{\widetilde{r}}^2&=\dot{r}^2=\left(\lambda+\frac{E^2}{L^2}\right)(\widetilde{r}+\overline{r})^4-(\widetilde{r}+\overline{r})^2+2M(\widetilde{r}+\overline{r})-Q^2\\
&=\widetilde{r}\left[\left(\lambda+\frac{E^2}{L^2}\right)\widetilde{r}^3+4\overline{r}\left(\lambda+\frac{E^2}{L^2}\right)\widetilde{r}^2+\left(6\overline{r}^2\left(\lambda+\frac{E^2}{L^2}\right)-1\right)\widetilde{r}+\left(4\overline{r}^3\left(\lambda+\frac{E^2}{L^2}\right)-2\overline{r}+2M\right)\right]
\end{align*}
and considering the new Binet variable $u:=1/\widetilde{r}=(r-\overline{r})^{-1}$, we get
\[\dot{u}^2=\left(\lambda+\frac{E^2}{L^2}\right)+4\overline{r}\left(\lambda+\frac{E^2}{L^2}\right)u+\left(6\overline{r}^2\left(\lambda+\frac{E^2}{L^2}\right)-1\right)u^2+\left(4\overline{r}^3\left(\lambda+\frac{E^2}{L^2}\right)-2\overline{r}+2M\right)u^3\]
and it is now straightforward to put this cubic in depressed form and then rewrite it in Weierstrass' form. We summarize the discussion in the following result:
\begin{prop}\label{into_weierstrass}
Let $\gamma=(t,r,\pi/2,\phi)$ be a non-circular, non-radial equatorial null geodesic in the RNdS metric, with energy $E$ and angular momentum $L$. The map $\ell\mapsto\phi(\ell)$ is a diffeomorphism onto its image so that we may re-parametrize $\gamma$ using $\phi$ and we abusively denote by $r$ the re-parametrized coordinate $\phi\mapsto r(\phi)$.

If $\lambda\ge-E^2/L^2$, then we may choose a root $\overline{r}\in\R$ of the quartic 
\[(\lambda+E^2/L^2)x^4-x^2+2Mx-Q^2\]
and if we let
\[\left\{\begin{array}{l}
\delta=\lambda+{E^2}/{L^2},\\[.5em]
\gamma=4\overline{r}\delta,\\[.5em]
\beta=6\overline{r}^2\delta-1,\\[.5em]
\alpha=4\overline{r}^3\delta-2\overline{r}+2M.
\end{array}\right.~~~~\text{as well as}~~~~\left\{\begin{array}{l}
g_2:=\frac{1}{4}\left(\frac{\beta^2}{3}-\alpha\gamma\right),\\[.5em]
g_3:=\frac{1}{8}\left(\frac{\alpha\beta\gamma}{6}-\frac{\alpha^2\delta}{2}-\frac{\beta^3}{27}\right),\\[.5em]
P:=\frac{\alpha}{4(r-\overline{r})}+\frac{\beta}{12},\end{array}\right.\]
then the function $P$ satisfies the Weierstrass equation
\[\dot{P}^2=4P^3-g_2P-g_3.\]

In other words, if the \emph{discriminant} $g_2^3-27g_3^2\ne0$, then the polar radial motion is given by
\[r(\phi)=\overline{r}+\frac{\alpha}{4\wp(\phi)-\beta/3},\]
where $\wp=\wp_{g_2,g_3}$ is the Weierstrass function associated to $(g_2,g_3)\in\R^2$.
\end{prop}

\begin{rem}
Differentiating the radial equation from (\ref{motion_RNdS})
we obtain
\[\ddot{r}=\frac{2Q^2L^2}{r^5}-\frac{3ML^2}{r^4}+\frac{L^2-Q^2(e^2+\mu)}{r^3}+\frac{M\mu-EeQ}{r^2}-\lambda\mu r.\]
Fixing an initial value for $r$ and $\dot{r}$, we obtain a second order Cauchy problem. Hence, if $r : I\to \R$ is a maximal solution of this problem, then we either have $|r|\to+\infty$ or $r\to0$ on $\partial I$. This says that ultimately, every geodesic is either always defined (stable orbit), or goes to $\infty$ (escape path) or dies at the singularity.

Qualitatively, the previous result says that the phase portrait, in Binet variable, of a generic null RNdS orbit describes (a connected component of) an elliptic curve.
\end{rem}

In practice, given a (polar) initial condition $(r_0,\dot{r}_0):=(r(\phi_0),\dot{r}(\phi_0))$, we have to find $z_0\in\C$ such that $\wp(z_0)=\tfrac{\alpha}{4(r_0-\overline{r})}+\tfrac{\beta}{12}$ and this can be done using the \emph{Carlson integrals} (see \cite{carlson})
\[R_F(x,y,z):=\frac{1}{2}\int_0^\infty \frac{\mathrm{d}\zeta}{\sqrt{(\zeta+x)(\zeta+y)(\zeta+z)}}.\]
More precisely, we have the following result:
\begin{cor}\label{ini_conds_weierstrass}
Fix $(L,E,r_0,\dot{r}_0)\in\R^*\times(\R^*_+)^2\times\R$ such that $\lambda L^2+E^2\ge0$ and let $\gamma$ be the unique maximal non-circular, non-radial equatorial null RNdS geodesic with energy $E$, angular momentum $L$ and such that $r(0)=r_0$ and $\dot{r}(0)=\dot{r}_0$ in polar parametrization $r=r(\phi)$. Recall also the constants $\overline{r},\alpha,\beta,\gamma,\delta,g_2,g_3$ from Proposition \ref{into_weierstrass}.

If $g_2^3-27g_3^2\ne0$, then the function $r$ is given (on its definition domain) by
\[r(\phi)=\overline{r}+\frac{\alpha}{4\wp_{g_2,g_3}(z_0+\phi)-\beta/3},~~\text{where}~~z_0:=R_F(\wp_0-z_1,\wp_0-z_2,\wp_0-z_3)\in\C,\]
with $z_{1,2,3}\in\C$ the roots of the Weierstrass cubic $4z^3-g_2z-g_3$ and $\wp_0:=\tfrac{\alpha}{4(r_0-\overline{r})}+\tfrac{\beta}{12}$.
\end{cor}
Numerically, we approach $\wp$ with the Coquereaux--Grossmann--Lautrup algorithm\footnote{based on the duplication formula and the Laurent expansion of $\wp$ at $0$} from \cite[\S 3]{coquereaux} and $R_F$ is approximated using the Carlson algorithm from \cite[\S 2]{carlson}.

\section{Model for the accretion disk}\label{mod_acc}

We now detail how we modelled the (thin steady nearly Keplerian opaque) accretion disk, radiating as a blackbody. For detailed treatments of accretion disks, see \cite{pringle,spruit}.

\subsection{Angular velocity of circular massive orbits}
First, we have to find the angular velocity of a circular equatorial orbit. This is done in the following result:
\begin{prop}\label{ang_vel}
Let $\gamma=(t,r,\theta,\phi) : I\to \mathcal{M}$ be a geodesic such that $\theta\equiv\pi/2$ and $\dot{r}=0$. Then, the angular velocity $\omega:=\dot{\phi}/\dot{t}$ is given by
\[\omega=\frac{1}{a+{r^2}/{\rho}},\]
where $\rho:=\sqrt{-\lambda r^4+Mr-Q^2}$.
\end{prop}
\begin{proof}
Consider the Lagrangian
\[\mathcal{L}=\tfrac{1}{2}g_{\mu\nu}\dot{\gamma}^\mu\dot{\gamma}^\nu=p_\mu\dot{\gamma}^\mu-\mathcal{H}(\gamma,p).\]
Since $\dot{r}=0$, we have $\tfrac{\partial\mathcal{L}}{\partial\dot{r}}=0$ and the radial Euler--Lagrange equation is
\[0=2\frac{\mathrm{d}}{\mathrm{d}\ell}\left(\frac{\partial\mathcal{L}}{\partial\dot{r}}\right)=2\frac{\partial\mathcal{L}}{\partial r}=\frac{\partial g_{tt}}{\partial r}\dot{t}^2+2\frac{\partial g_{t\phi}}{\partial r}\dot{t}\dot{\phi}+\frac{\partial g_{\phi\phi}}{\partial r}\dot{\phi}^2~\Longleftrightarrow~g_{tt,r}+2g_{t\phi,r}\omega+g_{\phi\phi,r}\omega^2=0,\]
where $\omega:=\dot{\phi}/\dot{t}$. Computing the derivatives, we obtain
\begin{align*}
&~g_{tt,r}+2g_{t\phi,r}\omega+g_{\phi\phi,r}\omega^2=0\\
\Longleftrightarrow &~\frac{2(r^4+a^2\Delta_r-a^4)-a^2r\Delta_r'}{\chi^2r^3}\omega^2+\frac{2a(r\Delta_r'+2(a^2-\Delta_r))}{\chi^2r^3}\omega+\frac{2(\Delta_r-a^2)-r\Delta_r'}{\chi^2r^3}=0\\
\Longleftrightarrow &~(\chi r^4-a^2(Mr+Q^2))\omega^2-2a(\lambda r^4-Mr+Q^2)\omega+\lambda r^4-Mr+Q^2=0\\
\Longleftrightarrow &~\omega=\frac{a\rho^2\pm r^2\rho}{a^2\rho^2-r^4}=\rho\frac{a\rho\pm r^2}{(a\rho-r^2)(a\rho+r^2)}=\frac{\rho}{a\rho\mp r^2}.
\end{align*}
But when $\lambda=a=Q=0$, we must find $\omega=+\sqrt{{M}/{r^3}}$ and thus the above sign is a plus.
\end{proof}

\subsection{Blackbody radiation temperature and brightness}
As mentioned above, we assume that the matter in the accretion disk radiates as a blackbody. To compute its surface temperature $T_s=T_s(r)$, we use the \emph{Shakura--Sunyaev formula} (see \cite[\S 2a]{shakura-sunyaev} or \cite[formula (26)]{spruit}). In SI units, it reads
\begin{equation}\label{SS}
\sigma_BT_s(r)^4=\frac{3GM\dot{M}}{8\pi r^3}\left(1-\sqrt{\frac{r_{\rm int}}{r}}\right),
\end{equation}
where $r_{\rm int}$ is the interior radius of the disk, $\dot{M}$ is the \emph{accretion rate} of matter into the disk and $\sigma_B$ is the Stefan-Boltzmann constant.

Now, for the brightness, we use Planck's law
\[B_\lambda(T)=\frac{2hc^2}{\lambda^5}\frac{1}{e^{\frac{hc}{k_B\lambda T}}-1},\]
$h$ is Planck's constant and $k_B$ is Boltzmann's constant; coupled with the Wien law $\lambda=b/T$, where $b$ is Wien's displacement constant. This yields, after evaluating the constants,
\[B(r):=B_{b/T}(T_s(r))=\frac{2hc^2}{b^5}\frac{T^5}{e^{\frac{hc}{k_Bb}}-1}\approx 4.086\cdot10^{-6}\times T^5.\]
This is the value by which we shall multiply the pixel's RGB triple corresponding to the temperature $T_s$, according to the conversion table by M. Charity\footnote{\url{http://www.vendian.org/mncharity/dir3/blackbody/}}. However, it turns out that implementing these values gives an over-bright disk, hence we found useful to rescale the brightness by $10^{-15}$ so that $B(r)\approx 4.086\cdot 10^{-21}\times T^5$. The user is then invited to give a value $B_0\ge0$, typically $B_0\le 10^4$, so that the disk becomes visible, as changing the inner (outer) radius or the accretion rate dramatically affects the brightness. The rescaled brightness is then $\widetilde{B}(r)=B_0T^5\times4.086\cdot10^{-21}$. If $B_0=0$ is chosen, the formula for $B(r)$ is ignored and a linear scaling of brightness is taken, from the outer to the inner radius.

\subsection{Gravitational redshift and Doppler effect}
Last, we have to take the Doppler effect and gravitational redshift into account for the temperature and the brightness, as we deal with relativistic speeds and strong gravitational fields. More precisely, we will rescale the temperature and brightness by factors $\alpha_{\rm Grav}^{-1}=(1+z_{\rm Grav})^{-1}$ and $\alpha_{\rm Dop}^{-1}=(1+z_{\rm Dop})^{-1}$ corresponding the the gravitational and Doppler shifts, respectively.

The gravitational redshift is easily computed from the matrix $(g_{\mu\nu})$. Indeed, for a stationary observer (a test particle with $\dot{r}=\dot{\theta}=\dot{\phi}=0$), the KNdS metric reduces to $\mathrm{d}s^2=-c^2\mathrm{d}\tau^2=g_{tt}c^2\mathrm{d}t^2$, where $\tau$ is the proper time of the observer. Therefore, the gravitational redshift for such an observer is simply given by
\[\alpha_{\rm Grav}=\frac{\mathrm{d}t}{\mathrm{d}\tau}=\frac{1}{\sqrt{-g_{tt}}}=\chi\sqrt{\frac{\Sigma}{\Delta_r-a^2\sin^2\theta\Delta_\theta}}.\]

For the Doppler shift, take a circular massive orbit with constant radius $r$, angle $\phi=\phi(\ell)$ and four-velocity $l^\mu=(\dot{t},0,0,\dot{\phi})$, and a photon path leaving the point $(t,r,\pi/2,\phi)$ with angle $\vartheta$ with respect to $l^\mu$. Following \cite[\S 48]{landau-lifshitz}, in natural units, the Doppler shift reads
\[\alpha_{\rm Dop}=\frac{1-v\cos\vartheta}{\sqrt{1-v^2}},\]
where $v=\omega\sqrt{r^2+a^2}$ is the velocity of the orbit which, using Proposition \ref{ang_vel}, reads
\[v=\frac{\sqrt{r^2+a^2}}{a+\frac{r^2}{\sqrt{-\lambda r^4+Mr-Q^2}}}.\]

\section{Implementation and comparison of the methods}\label{rem_implem}

In this section, we give some details on the Scilab functions we created to solve the geodesic equations and to draw the shadow of a KNdS black hole, with an accretion disk. The functions are designed to allow the user to tune parameters (cosmological constant mass, charge, angular momentum, accretion rate, brightness...) as desired and to draw a shadow accordingly. The full scripts and documentation may be found at \url{https://github.com/arthur-garnier/knds_orbits_and_shadows.git}.

In all our programs, we systematically rescale the initial data so that $G=c=M=4\pi\epsilon_0=1$ and go back to SI units after computations.

The programs \texttt{auxi.sci} and \texttt{orbit.sci} are intended to solve the geodesic equations. The first one is simply a library of useful functions, such as the conversion between Cartesian and Boyer--Lindquist coordinates, the (inverse and derivatives of the) metric matrices, Christoffel symbols, etc. The second one is the solver itself: it takes as input the cosmological constant, the three parameters of the black hole, the mass of the particle ($0$ or $1$), the discretized affine parameter (maximal value and step-size) and the initial conditions of the geodesic, in Boyer--Lindquist coordinates. It also lets the user choose between the different integration methods we discussed above, as well as the method to use in the \texttt{ode} routine\footnote{such as RK4, RK45, BDF, Adams... see \url{https://help.scilab.org/docs/6.1.1/en_US/ode.html}}. As output, it returns the trajectory in Boyer--Lindquist coordinates and the Hamiltonian along the trajectory, that is, the values of $(r,\theta,\phi)$ and of $\mathcal{H}$ at each node.

\subsection{Shadowing and the backward ray tracing method}\label{brt}
The method we use to create the shadow of the black hole is quite standard: the \emph{backward ray tracing method}. For a detailed and illustrated explanation of this method, we refer to \cite{osiris}. The function that ray-traces the black hole is \texttt{shadow.sci}; it takes as input the parameters of the black hole and the cosmological constant, the image to use for the shadowing and the accretion data\footnote{inner and outer radii, accretion rate, angle of view (from the equatorial plane) and brightness. But it also allows the user to force the temperature at extremal radii and to choose between the different shifts (gravitational, Doppler, both, none) described in \S\ref{mod_acc}. For more details, see \url{https://github.com/arthur-garnier/knds_orbits_and_shadows.git}.}. Though doable with any integration method, we used the Carter equations for \texttt{shadow.sci}, as it is by far the fastest method available (see \S \ref{comparison}). For a non-rotation black hole ($a=0$), the Weierstrass $\wp$ functions yield a more efficient process and they advantageously replace Carter's equations in this particular case, as explained below.

The basic idea is as follows: consider a static point in the KNdS spacetime, far from the center, representing the ``eye'' of our observer. Consider also a screen between our observer and the black hole, orthogonal to the segment joining the center and the observer. The celestial sphere emits light in every direction and some of it will eventually reach the observer, passing through the screen and the point where it hits the screen gives the pixel to draw at this point, depending on where it left the celestial sphere. However, as light will not propagate in straight lines, it is hard to know which ray will cross the screen in advance.

Therefore, we work backwards: suppose the observer emits light in every direction and keep only those rays that hit the screen at some point. As we are far from the source, we assume that light travels in straight lines between the camera and the screen. Then, we let the light ray trace backward in time, in the KNdS geometry, and see where it eventually lands (actually, where it came from): if it dies in the black hole, no pixel is displayed on the screen and if it crosses the celestial sphere, then the pixel is coloured in accordance with where it touches the sphere.

More precisely, first, we consider an artificial celestial hemisphere on which we project our original image, seeing it as a portion of its tangent plane parallel to the screen (and on the other side of the black hole). As a projection, we simply choose the standard and widely used \emph{equirectangular projection}, which has the advantage of taking the celestial hemisphere to a square, which we may rescale to fit our image. However, as we are dealing with black holes, a light ray may land on the other hemisphere (see Figure \ref{rays}), which we therefore choose to fill with a mirrored version of the original image. This avoids pixel loss and too much distortion of the original picture, which is assumed to be flat.

Next, for each pixel of the screen, we consider the null geodesic starting at this point and with velocity directed by the line from the point observer. We then solve the geodesic equations (backwards) and we see if the ray ends in (came from) the black hole or touches the sphere somewhere. If so, the RGB value of the pixel on the screen is given by the value of the landing pixel on the sphere and we carry this process on until every pixel has been worked out. We illustrate the process in Figure \ref{rt}.

In the case of an RNdS (non-rotating) black hole, the metric is spherically symmetric and, as described in Proposition \ref{into_weierstrass} and Corollary \ref{ini_conds_weierstrass}, a photon path is explicitly described in terms of the Weierstrass $\wp$ function, for which efficient approximation algorithms exist \cite{coquereaux,carlson}. Moreover, because of the symmetry, we don't have to compute every geodesic: given an initial datum, use a linear rotation to bring the initial velocity (and hence the full orbit) in the plane $\{\theta=\pi/2\}$. Then, we give values to the various constants involved in the expression of the polar radial geodesic and, instead of computing the full orbit, we simply solve the equation $r=r_{\Sph}$ where $r_{\Sph}$ is the radius of the celestial sphere. This can be done rather easily, precisely and quickly: we compute some values until we cross the sphere and the first such point is used as an initial value for the Newton method\footnote{We also use this procedure for the accretion disk, rather than a naive interpolation.}. We finally rotate the result back and find our landing pixel. Thus, no full orbit calculation nor ODE solving is required, making the resulting program rather fast.

\begin{center}
\begin{figure}[h!]
\hspace{-0.75cm}
\begin{subfigure}[c]{0.5\columnwidth}
\centering
\includegraphics[scale=0.22]{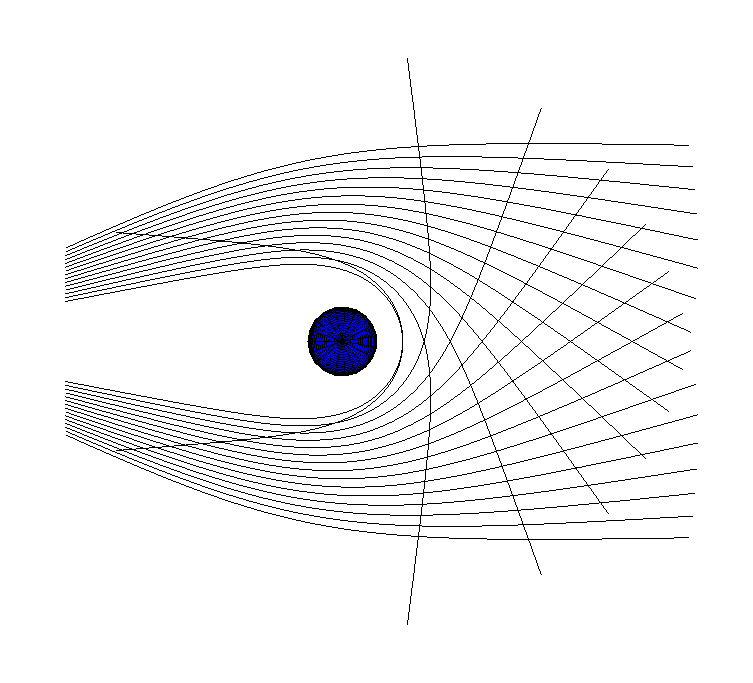}
\caption{Schwarzschild}
\end{subfigure}
\hfill
\begin{subfigure}[c]{0.5\columnwidth}
\includegraphics[scale=0.22]{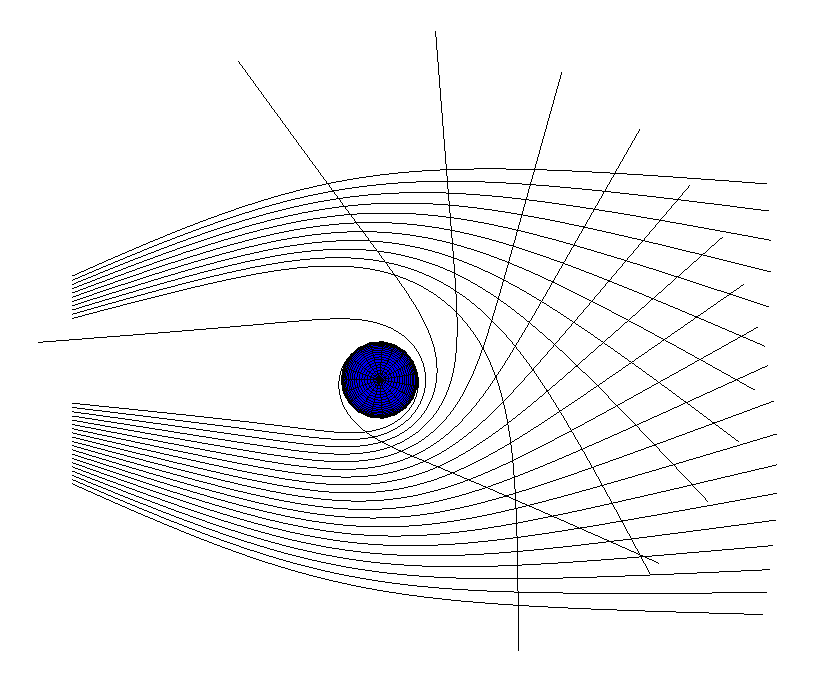}
\caption{Extremal Kerr}
\end{subfigure}
\caption{A pencil of equatorial rays near a black hole.}\label{rays}
\end{figure}
\end{center}

Concerning the accretion disk, we simply interpolate the plane $\{\theta=\pi/2\}$: if the geodesic ray hits the plane (up to some fixed threshold) at a point whose radius is between the extremal radii of the disk, then we compute the radiation temperature at this point, as well as the gravitational and Doppler effects described in \S\ref{mod_acc}. We then give the corresponding colors and brightness to the associated pixel on the screen.

All this requires a Scilab package for processing images. The package \texttt{IPCV 4.1.2}\footnote{See \url{https://atoms.scilab.org/toolboxes/IPCV} and \url{https://ipcv.scilab-academy.com}} is well-suited for this purpose. The command \texttt{imread} loads an image (\texttt{.jpg}, \texttt{.png}, etc) with $N\times M$ pixels and encodes it as an $N\times M\times 3$ hypermatrix with, for each $(i,j)\in\{0,\dotsc,N\}\times\{0,\dotsc,M\}$, the three RGB values of the pixel in position $(i,j)$. Then, we produce the pixels for the shadowed image as described above and put them in a similar $N\times M\times 3$ hypermatrix, which we can display as an image using the command \texttt{imshow}.

\begin{figure}[h!]
\begin{tikzpicture}[scale=0.9]
  \coordinate (z) at (0,0);
  \coordinate (x) at (-3.64,0);
  \coordinate (p) at (0,2);
  \coordinate (m) at (0,-2);
  \coordinate (mp) at (-1.67,1.08);
  \coordinate (mm) at (-1.67,-1.08);
  \coordinate (ppa) at (2,3.1416);
  \coordinate (mma) at (2,-3.1416);
  
  \draw[fill=black] (0,0.235) arc[start angle=90, end angle=-270,radius=0.235cm];
  
  \draw[draw=green,ultra thick] (mma)--(ppa);
  \draw[draw=orange,ultra thick] (mp)--(mm);
  \draw[draw=orange,dotted,ultra thick] (mm)--($(mp)!1.2!(mm)$) (mp)--($(mm)!1.2!(mp)$);
  
  \draw[dashed,<-,very thick] (1.2,2)--(1.8,2.5);
  \draw[dashed,->,very thick] (1.8,-2.5)--(1.2,-2);
  
  \draw[draw=red] (p) arc[start angle=90, end angle=-90,radius=2cm];
  \draw[draw=red,dotted] (p) arc[start angle=90, end angle=110,radius=2cm];
  \draw[draw=red,dotted] (m) arc[start angle=-90, end angle=-110,radius=2cm];
  \draw (x)--(p) (x)--(m);
  \draw[dotted] (p)--($(x)!1.7!(p)$) (m)--($(x)!1.7!(m)$);
  
  \draw[dashed,opacity=0.3,->] (-4,0)--(4,0);
  \draw[dashed,opacity=0.3,->] (0,-4)--(0,4);
  
  \draw[opacity=0.3] (4,0) node[below right]{$x$};
  \draw[opacity=0.3] (0,4) node[above right]{$y$};
  
  \draw [cyan] plot [smooth, tension=1] coordinates {(1.7,-2) (0,0.7) (x)} [arrow inside={end=stealth,opt={cyan,scale=2}}{0.2,0.4,0.66}];
  
  \matrix [draw] at (current bounding box.south west) {
  \node [draw=black, shape=circle, fill=black,label=right:Black hole] {}; \\
  \node [draw=red,shape=circle,line width=1,label=right:Celestial sphere] {}; \\
  \node [draw=green,shape=rectangle,line width=2,label=right:Original image] {}; \\
  \node [draw=orange,shape=rectangle,line width=2,label=right:Final screen] {}; \\
  \node [draw=cyan,shape=circle,line width=1,label=right:Photon path] {}; \\
  \node [draw=black,shape=rectangle,line width=1,dashed,label=right:Projection] {}; \\
};
 
  \draw[opacity=0.3,purple] plot [smooth,tension=1] coordinates {($(p)!0.35!(ppa)$) (2,0) ($(m)!0.35!(mma)$)};
  \draw[opacity=0.3,purple] plot [smooth,tension=1] coordinates {($(p)!0.55!(ppa)$) (2,0) ($(m)!0.55!(mma)$)};
  \draw[opacity=0.3,purple] plot [smooth,tension=1] coordinates {($(p)!0.75!(ppa)$) (2,0) ($(m)!0.75!(mma)$)};
\end{tikzpicture}
\caption{Schematics of our shadowing method (in the $xy$-plane).}\label{rt}
\end{figure}
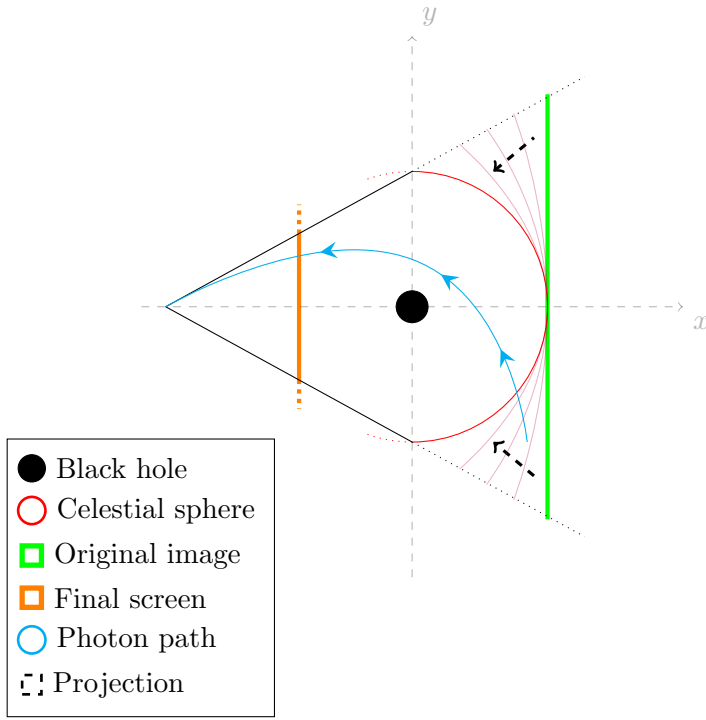

\subsection{Comparison of the accuracy and execution time of the integrators}\label{comparison}
We can now compare the different formulations and schemes. Regarding the accuracy, we compare the conservation of the Hamiltonian and Carter's constant along trajectories. For the execution times, we will shadow a KNdS black hole in low resolutions, as some of the methods are quite long and the differences are obvious, even in low resolution. We shall observe that the fastest and most accurate method is, as one could expect, the Carter equations. In the case of a RNdS black hole, the analytic method using $\wp$ is the fastest one.

For all calculations and illustrations, except otherwise stated, we consider a KNdS space with parameters\footnote{the corresponding unit-less cosmological constant is $\Lambda=3\cdot10^{-4}$} $\Lambda=3.4\cdot10^{-11}~{\rm m}^{-2}$, $M=4\cdot10^{30}~{\rm kg}\approx2M_\odot$, $a=0.95$ and $Q=0.3$. It should be mentioned that our shadowing program \texttt{shadow.sci} is designed for a black hole with mass around $2M_\odot$ (typically between $10^{-3}M_\odot$ and $10M_\odot$), because of the choices we had to make in the actual code: the integration interval for the affine parameter, the numerical tolerance on the intersection of a ray with the accretion disk or the celestial sphere, the position and size of the virtual screen, etc. Every computation was made on a 8-core 3.00 GHz CPU with 16 Go of RAM.

Consider a massive orbit ($\mu=-1$) with $(r_0,\theta_0,\phi_0,\dot{r}_0,\dot{\theta}_0,\dot{\phi}_0)=(12.3,\pi/2,0,0,0.014,0.014)$. This is a non-planar orbit so the Carter constant is not $0$. The orbit is depicted in Figure \ref{orbit}. The evolution of the Hamiltonian and Carter constant are depicted in Figures \ref{hams} and \ref{cars}.

To be more quantitative, the maximal deviations and execution times for each method are summarized in the Table \ref{compa}. We observe that the two symplectic Euler schemes are comparable in terms of conservation, but the $q$-implicit one is, as expected, almost twice as fast as the other one. The implicit St\"{o}rmer--Verlet scheme is far more efficient that the Verlet scheme, but it is also the slowest method. The Euler--Lagrange formulation is rather fast, but leaves the Hamiltonian far from being constant. The Hamilton equations are very efficient and rather fast, and seem to form the most reasonable method, except for the Carter equations (\ref{motion_with_momenta}). This last method is definitely the fastest and most accurate one.

Figure \ref{leaves} depicts some remarkable planar leaf-orbits. These all have $\theta_0=\pi/2$, $\phi_0=\dot{r}_0=\dot{\theta}_0=0$ and $\dot{\phi}_0=0.05$. These orbits are consistent with the ones from the figures in \cite{perez-giz-levin}.

In order to compare the different methods of integration regarding the shadowing process, we make several shadows of the same black hole, using the simple coloured grid displayed in Figure \ref{figure_0}. This picture is chosen so that the reader may easily compare our figures with the already existing ones in the literature, such as \cite[Fig. 11]{bacchini-ripperda}, \cite[Fig. 12]{osiris} or \cite[Fig. 9]{wang-chen-jing}. The resulting comparison images have a resolution of $100\times100$ and are depicted in Figure \ref{compa_shadows}. The corresponding execution times are in Table \ref{exec_times}. We can see that all the symplectic schemes display a singularity at the rotation axis.

We may also see how the execution times grow with the number of pixels. As the answer is pretty clear, we only made the computation for the same image and with resolutions from $10^2$ to $30^2$ pixels. The resulting graph (in log-log scale) is in Figure \ref{pixels_and_times}. As expected, the best method is, by far, the Carter equations (numerically integrated with the Adams methods from \cite{hindmarsh}). It should be mentioned that the function we used to draw these shadows is longer than the program \texttt{shadow.sci} using a specific method, since it is designed to work with all methods at once. We finish our comparison by considering the differences between the methods using Carter's equations and Weierstrass' functions, in the case $a=0$. We see in Figure \ref{slow_and_fast} that the resulting images are the same, while the analytic method using $\wp$ is highly faster than the ODE one. To make the time comparisons more quantitative, in Table \ref{expo-regs} we provide an exponential regression on the number of displayed pixels for each method.

\vspace{-2mm}

\begin{center}
\begin{figure}[h!]
\includegraphics[scale=0.242]{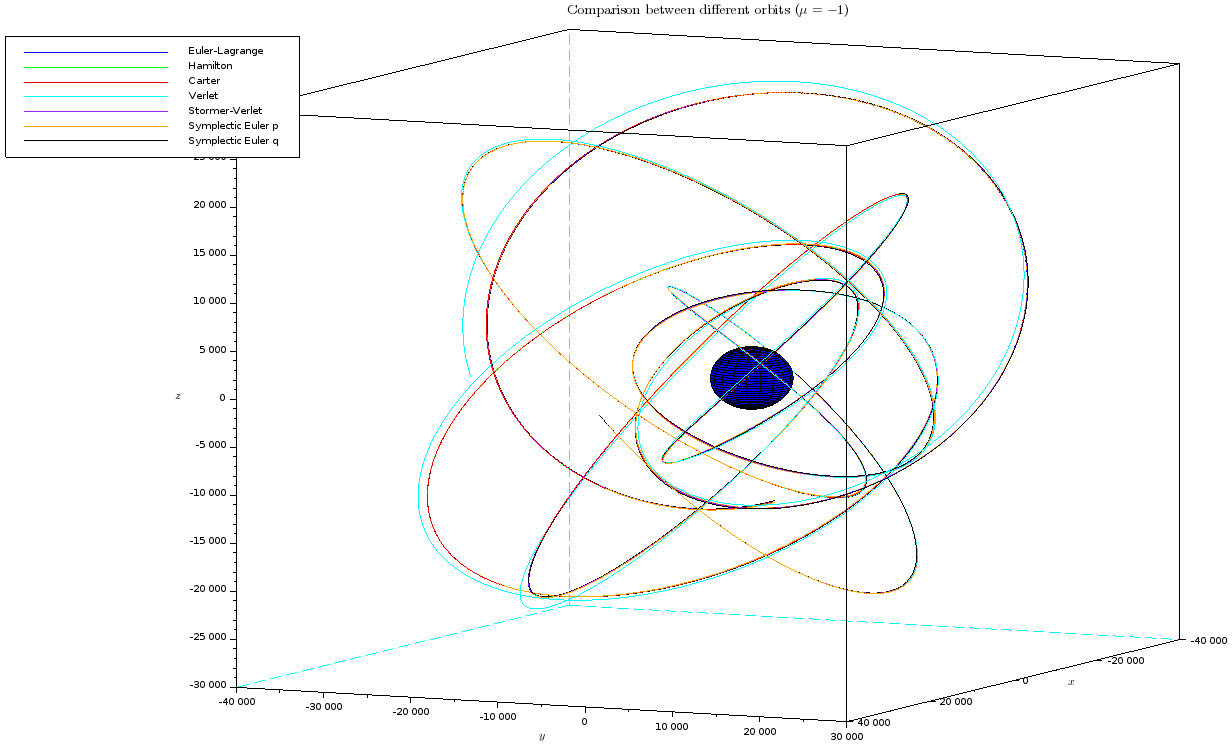}
\caption{A non-planar prograde orbit around a KNdS black hole. The central ellipsoid represents the outer horizon.}\label{orbit}
\end{figure}

\begin{table}[h!]
\begin{tabular}{c|c|c|c|}
\cline{2-4} & Max deviation on $\mathcal{H}$ & Max deviation on $C$ & Execution time (sec) \\
 \hline
\multicolumn{1}{|c|}{Carter} & $7.53\cdot10^{-7}$ & $1.67\cdot10^{-5}$ & $0.481$ \\
\multicolumn{1}{|c|}{Hamilton} & $8.91\cdot10^{-7}$ & $1.46\cdot10^{-5}$ & $1.083$ \\
\multicolumn{1}{|c|}{St\"{o}rmer--Verlet} & $5.71\cdot10^{-3}$ & $4.92\cdot10^{-4}$ & $4.818$ \\
\multicolumn{1}{|c|}{Euler--Lagrange} & $9.99\cdot10^{-1}$ & $1.25\cdot10^{-5}$ & $1.065$ \\
\multicolumn{1}{|c|}{Verlet} & $9.74\cdot10^{-1}$ & $1.19\cdot10^{-2}$ & $1.876$ \\
\multicolumn{1}{|c|}{$q$-symplectic Euler} & $9.74\cdot10^{-1}$ & $2.29\cdot10^{-2}$ & $2.505$ \\
\multicolumn{1}{|c|}{$p$-symplectic Euler} & $9.74\cdot10^{-1}$ & $2.36\cdot10^{-2}$ & $4.346$ \\
\hline
\end{tabular}
\captionof{table}{Comparison of the methods (the deviations are expressed as absolute percentages of the initial value).}
\label{compa}
\end{table}

\begin{figure}[h!]
\includegraphics[scale=0.242]{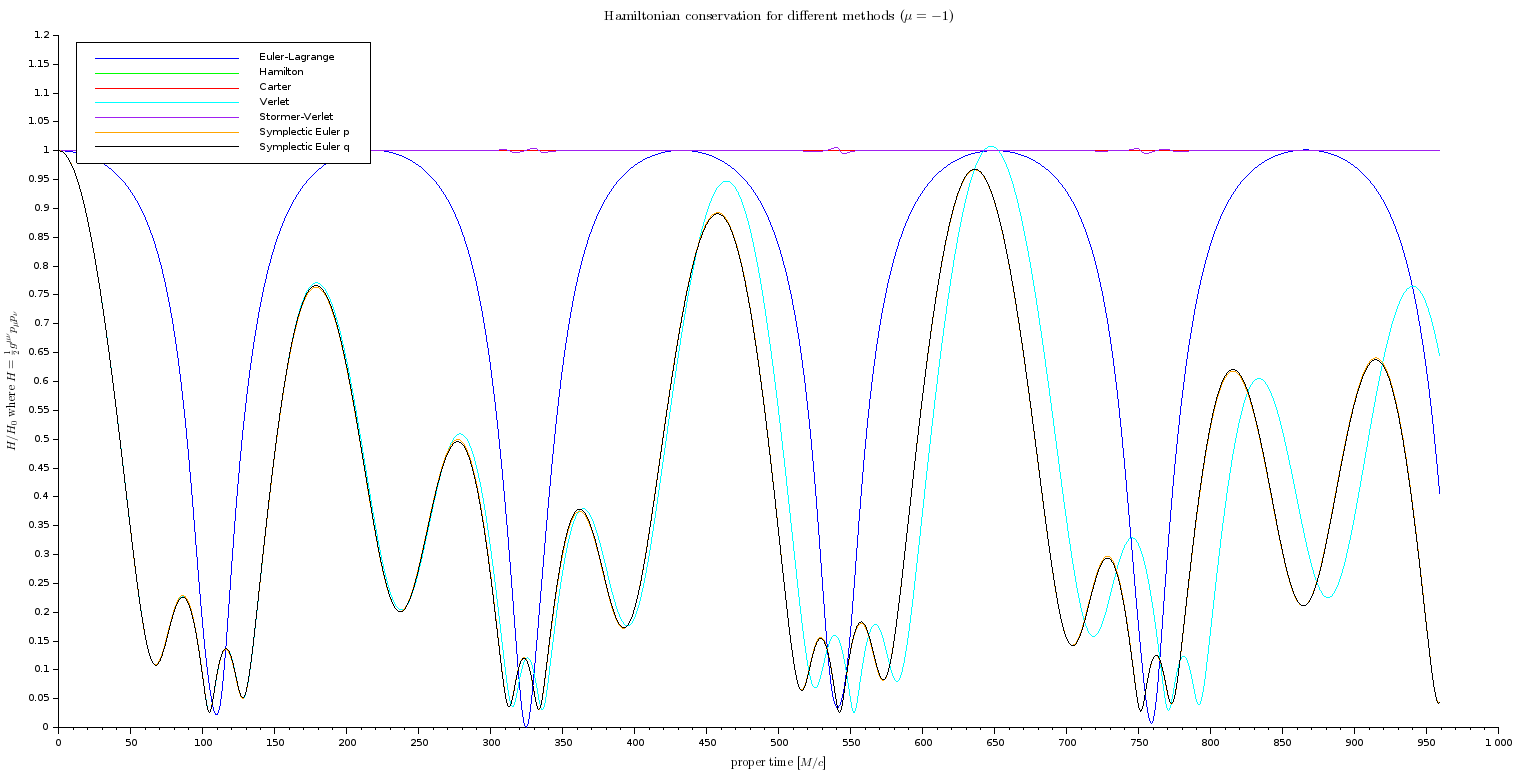}
\caption{Hamiltonian evolution.}\label{hams}
\end{figure}

\begin{figure}[h!]
\includegraphics[scale=0.242]{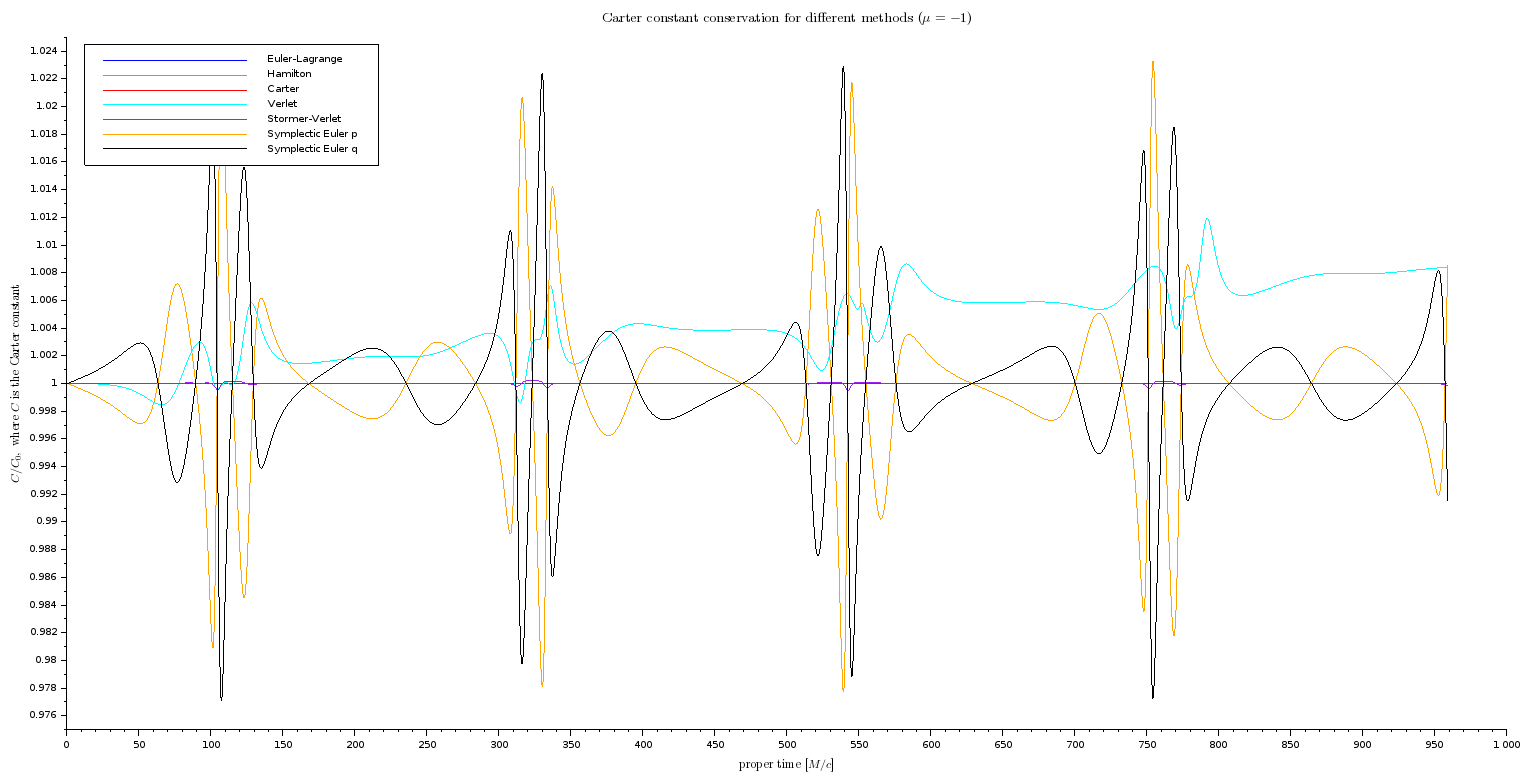}
\caption{Carter's constant evolution.}\label{cars}
\end{figure}
\end{center}

\begin{center}
\begin{figure}[h!]
\includegraphics[scale=0.242]{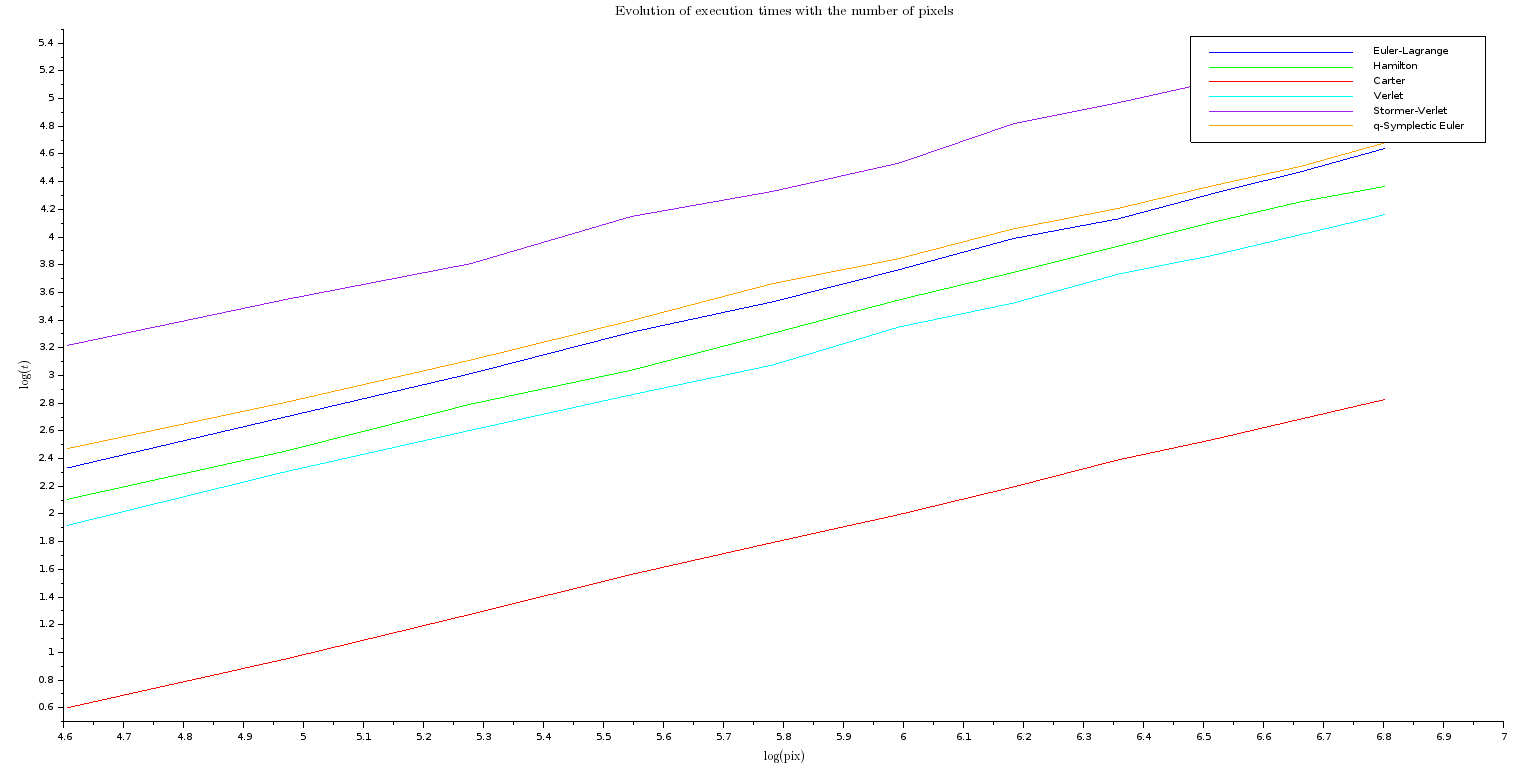}
\caption{Growth of execution times with the number of pixels to shadow.}\label{pixels_and_times}
\end{figure}
\end{center}

\vspace{-8mm}

\begin{center}
\begin{figure}[h!]
\begin{subfigure}[c]{0.4\textwidth}
\centering
\includegraphics[scale=0.135]{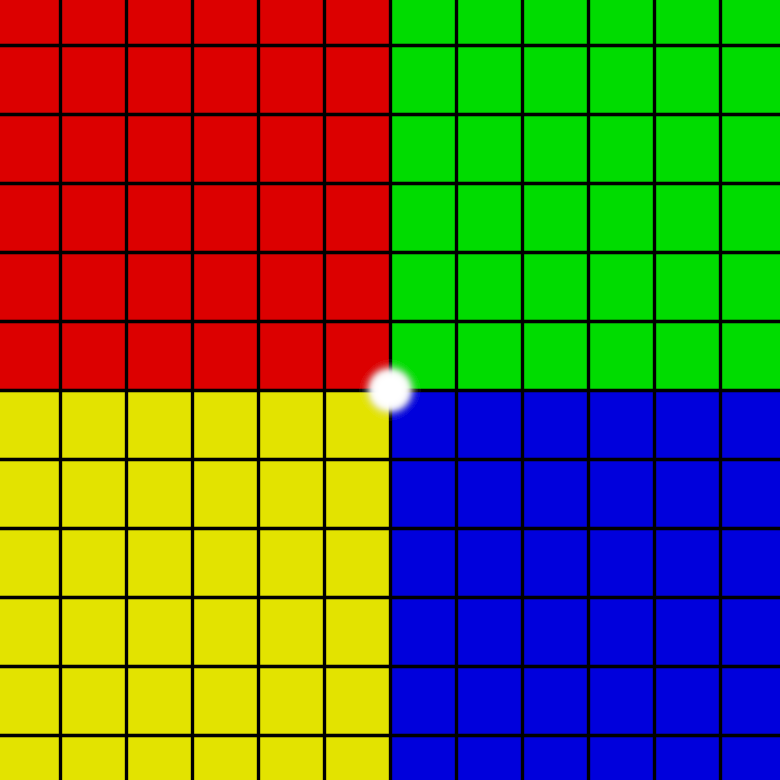}
\caption{The original grid used in the programs.}
\end{subfigure}
\hfill
\begin{subfigure}[c]{0.4\textwidth}
\centering
\includegraphics[scale=0.15]{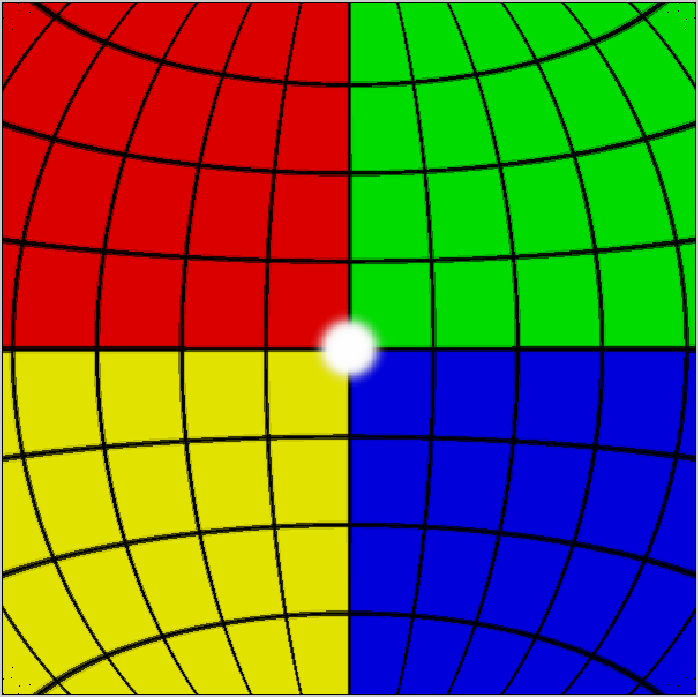}
\caption{Picture obtained with \texttt{shadow.sci}, setting $\Lambda=M=a=Q=0$.}
\end{subfigure}
\caption{The base pictures.}\label{figure_0}
\end{figure}
\end{center}

\vspace{-8mm}

\begin{center}
\begin{figure}[h!]
\begin{subfigure}[c]{0.3\columnwidth}
\centering
\includegraphics[width=0.7\linewidth]{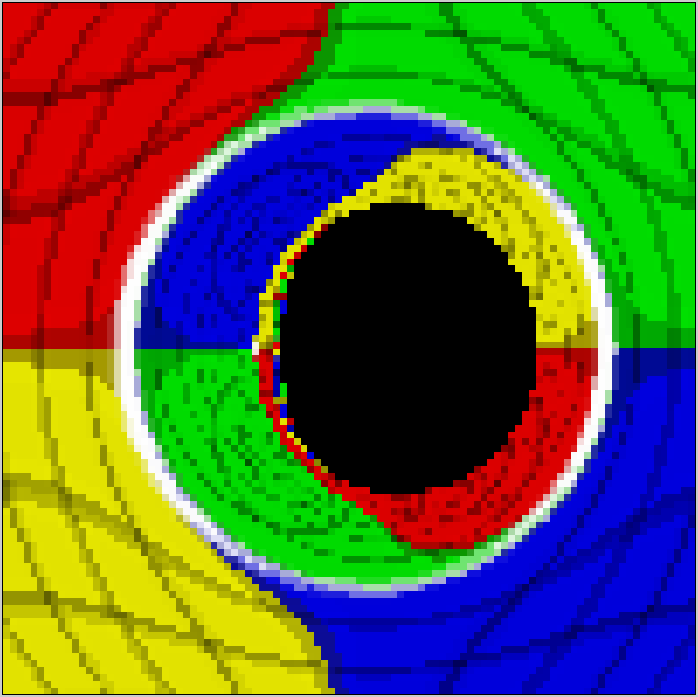}
\caption{Euler--Lagrange}
\end{subfigure}
\hfill
\begin{subfigure}[c]{0.3\columnwidth}
\centering
\includegraphics[width=0.7\linewidth]{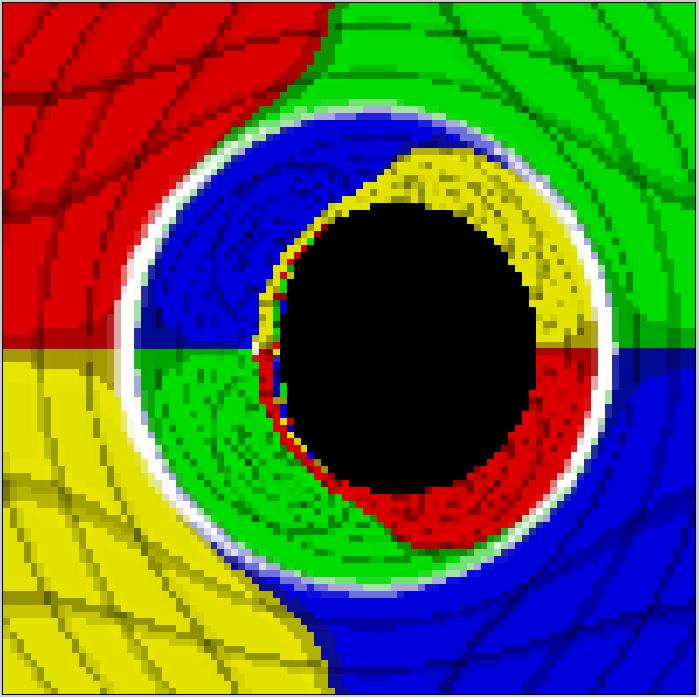}
\caption{Hamilton}
\end{subfigure}
\hfill
\begin{subfigure}[c]{0.3\columnwidth}
\centering
\includegraphics[width=0.7\linewidth]{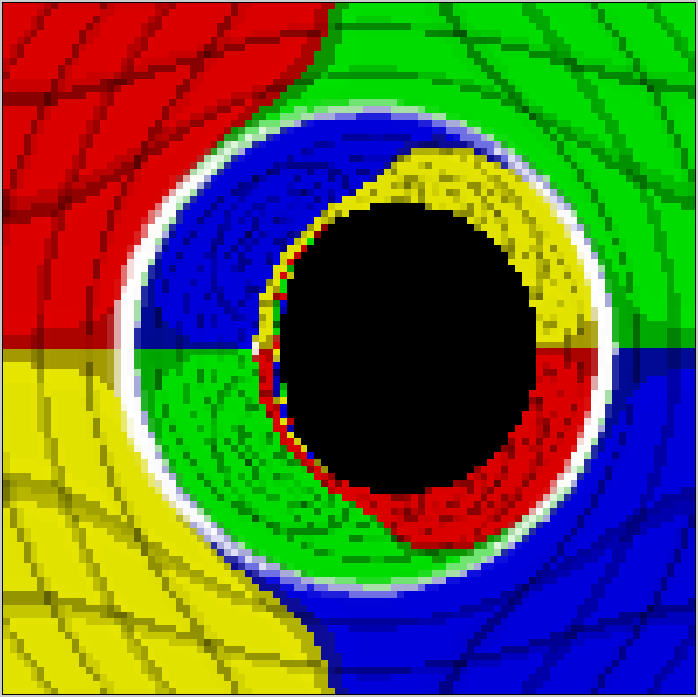}
\caption{Carter}
\end{subfigure}\\[1em]

\begin{subfigure}[c]{0.3\columnwidth}
\centering
\includegraphics[width=0.7\linewidth]{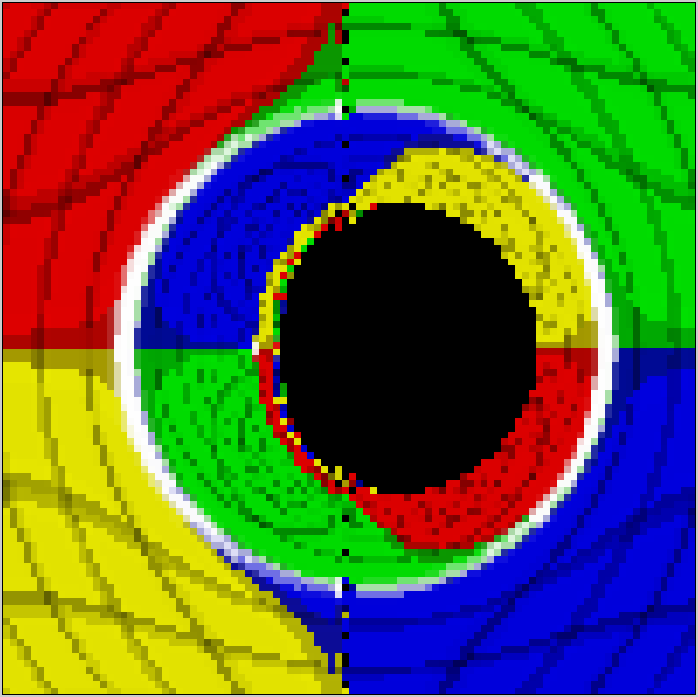}
\caption{Verlet}
\end{subfigure}
\hfill
\begin{subfigure}[c]{0.3\columnwidth}
\centering
\includegraphics[width=0.7\linewidth]{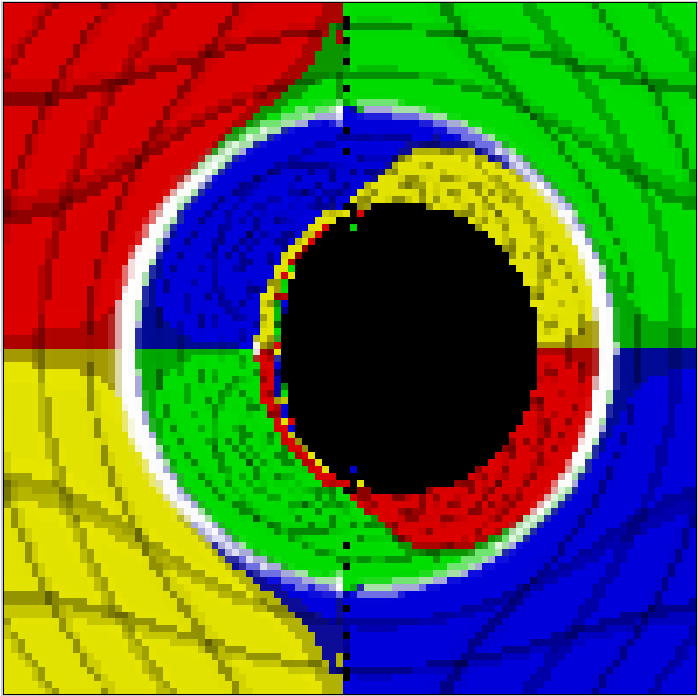}
\caption{St\"{o}rmer--Verlet}
\end{subfigure}
\hfill
\begin{subfigure}[c]{0.3\columnwidth}
\centering
\includegraphics[width=0.7\linewidth]{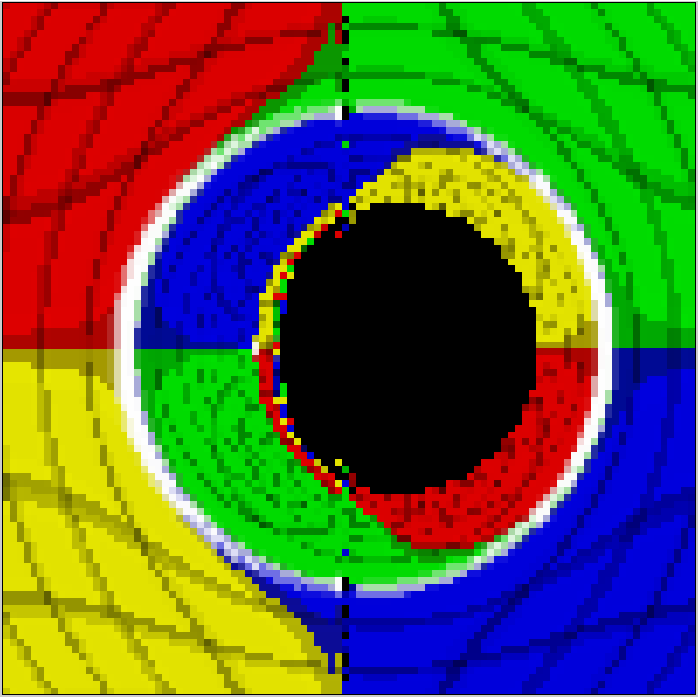}
\caption{$q$-symplectic Euler}
\end{subfigure}
\caption{Shadows obtained with the different methods ($100^2$ pixels).}
\label{compa_shadows}
\end{figure}
\end{center}

\vspace{-5mm}

%\vspace{-5mm}

\begin{table}[h!]
\begin{tabular}{c|c|c|c|c|c|c|}
\cline{2-7} & Carter & Verlet & Hamilton & Euler--Lagrange & $q$-symplectic Euler & St\"{o}rmer--Verlet \\ 
 \hline
\multicolumn{1}{|c|}{Times (sec)} & 190 & 690 & 887 & 1116 & 1198 & 2391 \\
\hline
\end{tabular}
\captionof{table}{Execution times for the shadows of Figure \ref{compa_shadows}.}\label{exec_times}
\end{table}

\begin{table}[h!]
\begin{tabular}{c|c||c|c|c|}
\cline{2-5} & \backslashbox{Method}{\raise-.3ex\hbox{Parameters}} & $a$ & $k$ & $\sigma$ \\
\hline
\multicolumn{1}{|c|}{\multirow{6}{*}{\parbox{5.5cm}{General shadowing program, $\text{pix}=\{10,12,14,\dotsc,30\}$.}}} & Carter & 1.0181 & 4.0953 & 0.0077 \\
\multicolumn{1}{|c|}{} & Verlet & 1.0233 & 2.7985 & 0.0173 \\
\multicolumn{1}{|c|}{} & Hamilton & 1.05 & 2.7525 & 0.017 \\
\multicolumn{1}{|c|}{} & Euler--Lagrange & 1.047 & 2.5028 & 0.0136 \\
\multicolumn{1}{|c|}{} & $q$-symplectic Euler & 1.0072 & 2.1848 & 0.0151 \\
\multicolumn{1}{|c|}{} & St\"{o}rmer--Verlet & 1.0062 & 1.4515 & 0.0343 \\
\hline
\multicolumn{1}{|c|}{\multirow{2}{*}{\parbox{5.5cm}{Dedicated program \texttt{shadow.sci}, $\text{pix}=\{20,30,40,\dotsc,120\}$.}}} & Weierstrass & 0.8037 & 5.7222 & 0.0487 \\
\multicolumn{1}{|c|}{} & Carter & 1.0293 & 4.9464 & 0.0239 \\
\hline
\end{tabular}
\captionof{table}{Regression: ${\rm time}\approx e^{-k}({\rm pix})^{2a}$, with standard deviation $\sigma$.}
\label{expo-regs}
\end{table}

\begin{center}
\begin{figure}[h!]
\hspace{-.5cm}
\begin{subfigure}[c]{0.25\columnwidth}
\centering
\includegraphics[width=1\linewidth]{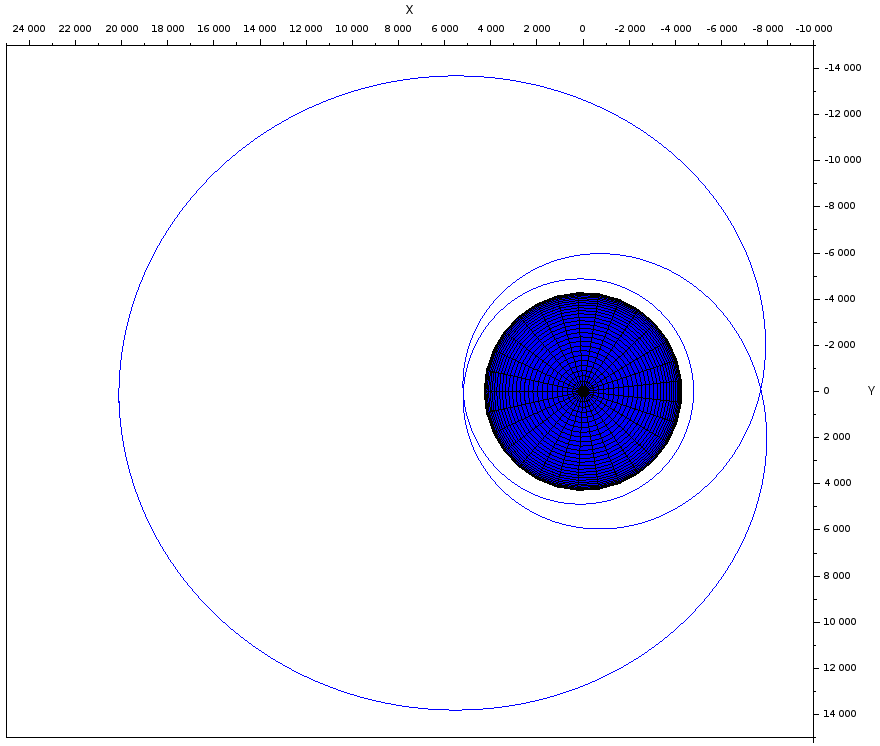}
\caption{$r_0=6.77253$}
\end{subfigure}
\hfill
\begin{subfigure}[c]{0.25\columnwidth}
\centering
\includegraphics[width=1\linewidth]{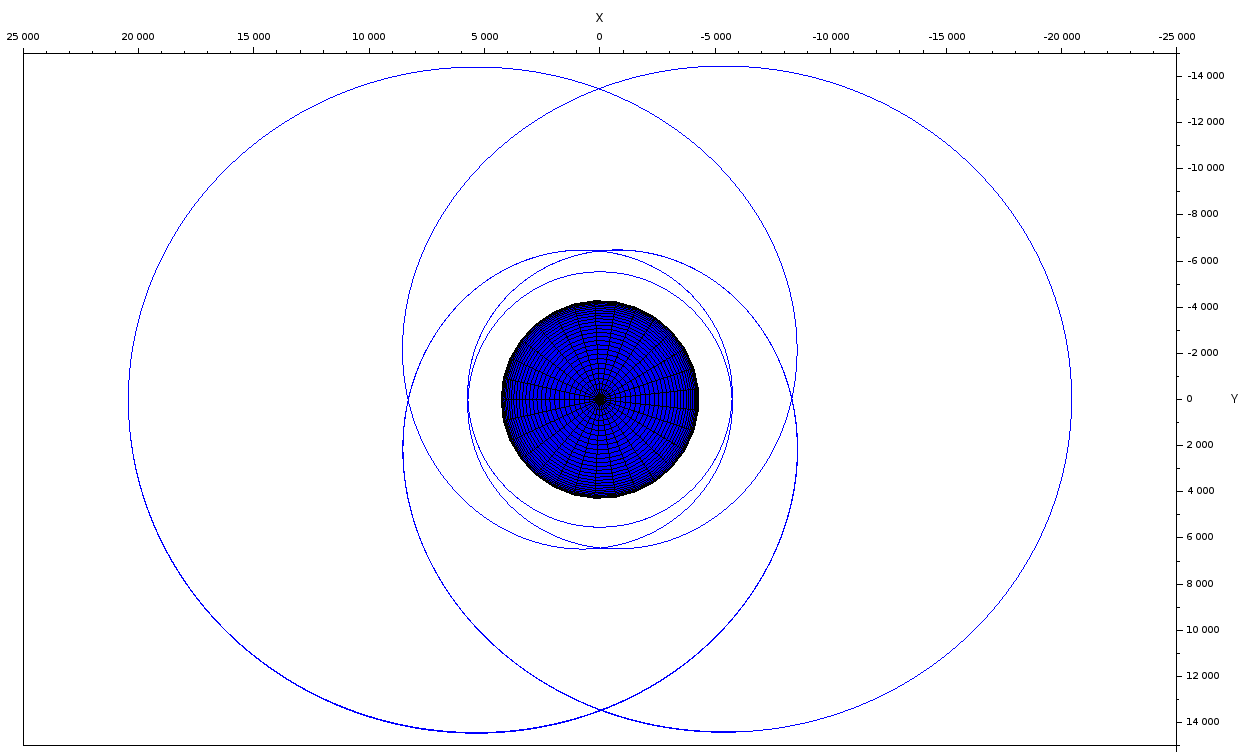}
\caption{$r_0=6.88102$}
\end{subfigure}
\hfill
\begin{subfigure}[c]{0.25\columnwidth}
\centering
\includegraphics[width=1\linewidth]{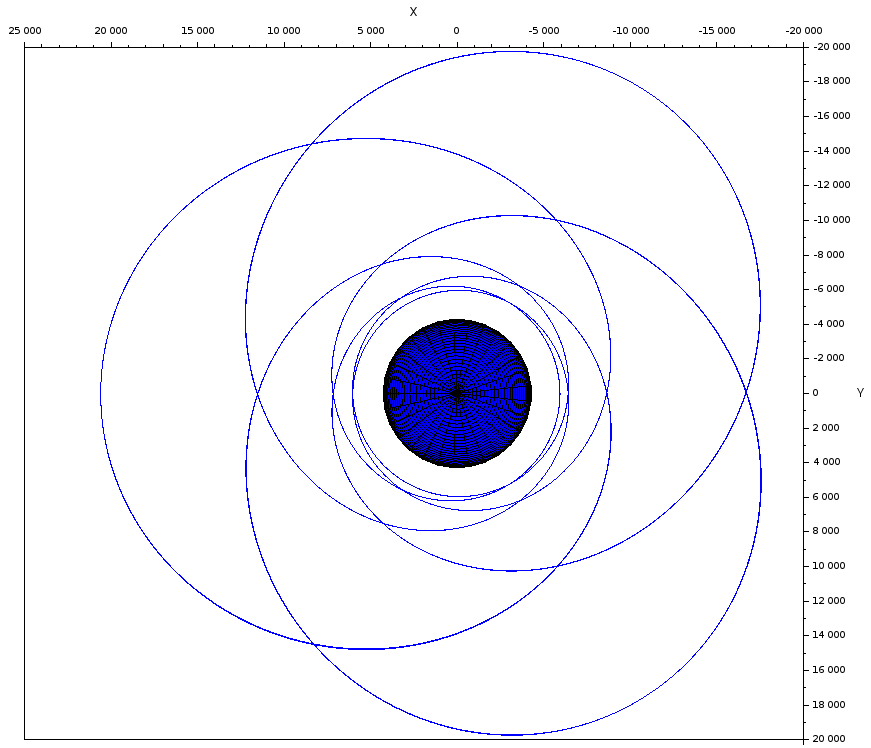}
\caption{$r_0=6.9361$}
\end{subfigure}
\hfill
\begin{subfigure}[c]{0.25\columnwidth}
\centering
\includegraphics[width=1\linewidth]{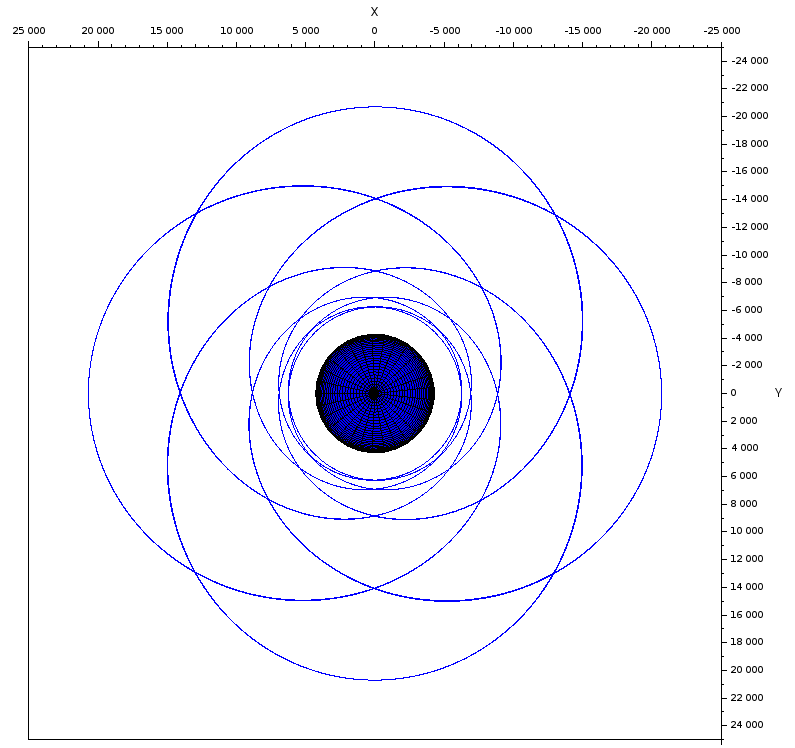}
\caption{$r_0=6.96938$}
\end{subfigure}\\
\hspace{-.5cm}
\begin{subfigure}[c]{0.25\columnwidth}
\centering
\includegraphics[width=1\linewidth]{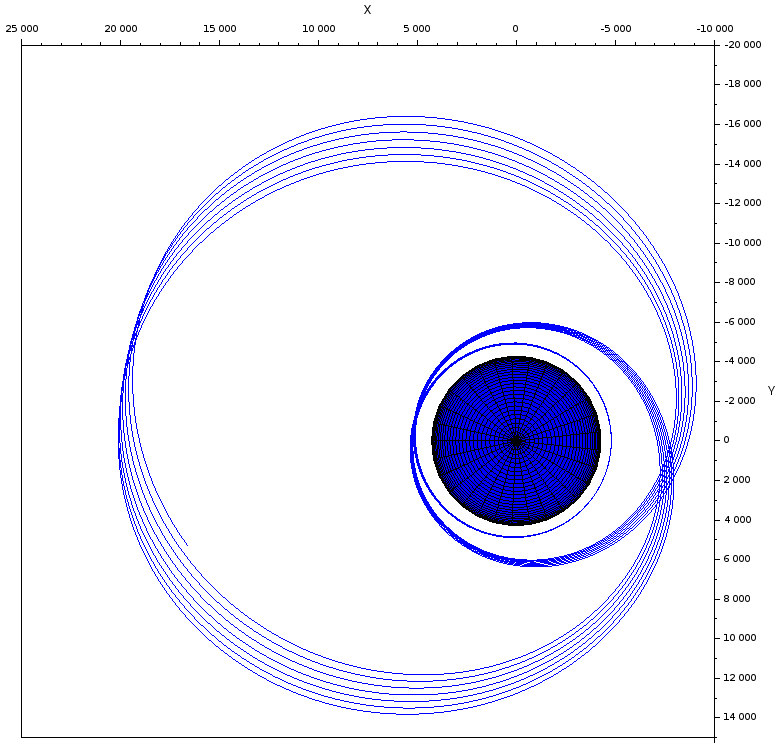}
\caption{$r_0=6.77253+10^{-3}$}
\end{subfigure}
\hfill
\begin{subfigure}[c]{0.25\columnwidth}
\centering
\includegraphics[width=1\linewidth]{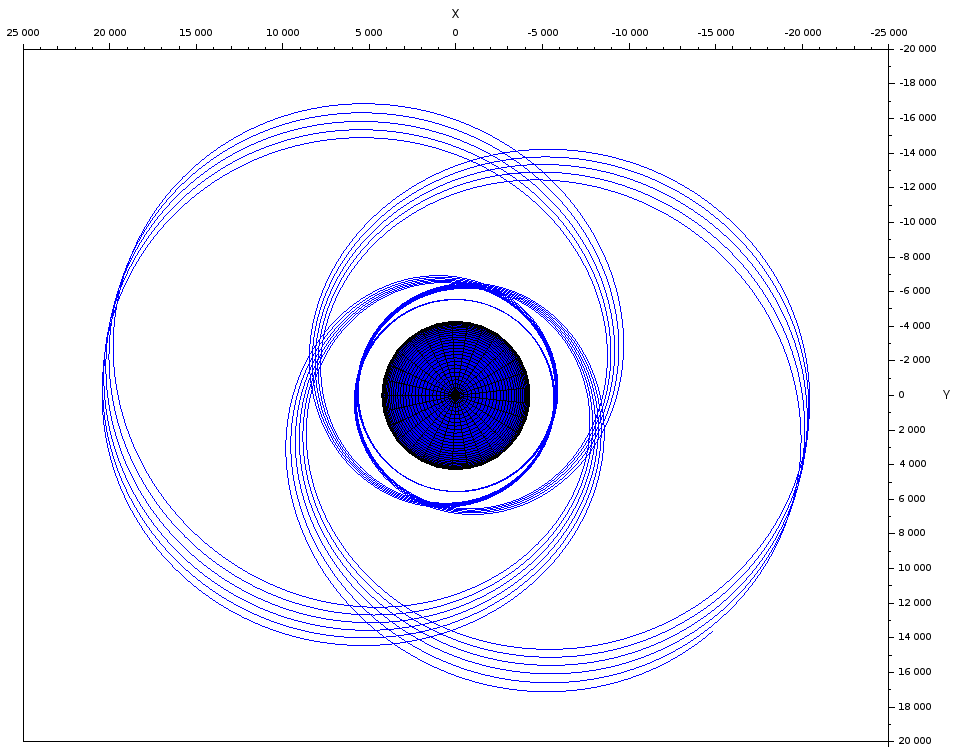}
\caption{$r_0=6.88102+10^{-3}$}
\end{subfigure}
\hfill
\begin{subfigure}[c]{0.25\columnwidth}
\centering
\includegraphics[width=1\linewidth]{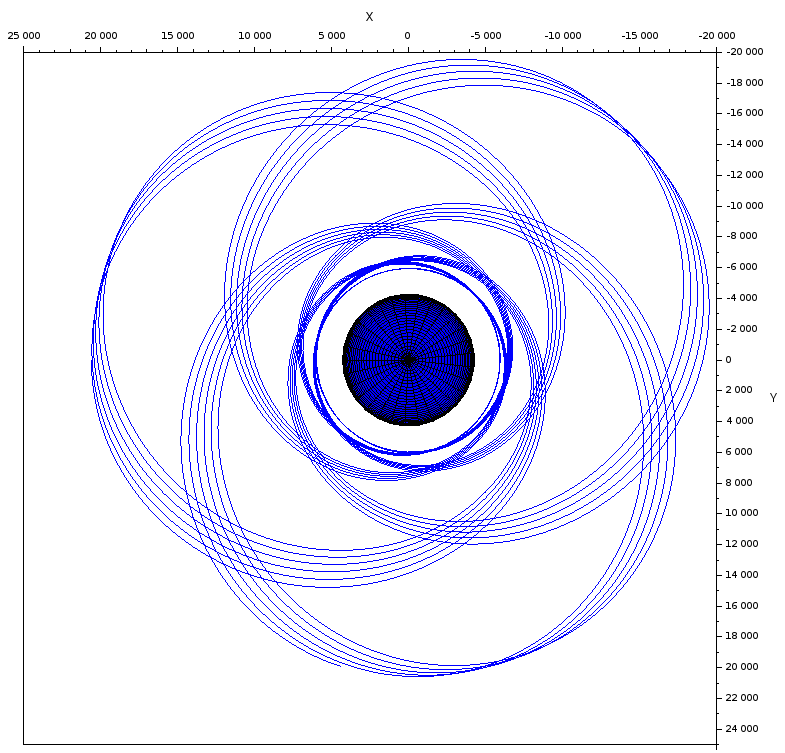}
\caption{$r_0=6.9361+10^{-3}$}
\end{subfigure}
\hfill
\begin{subfigure}[c]{0.25\columnwidth}
\centering
\includegraphics[width=1\linewidth]{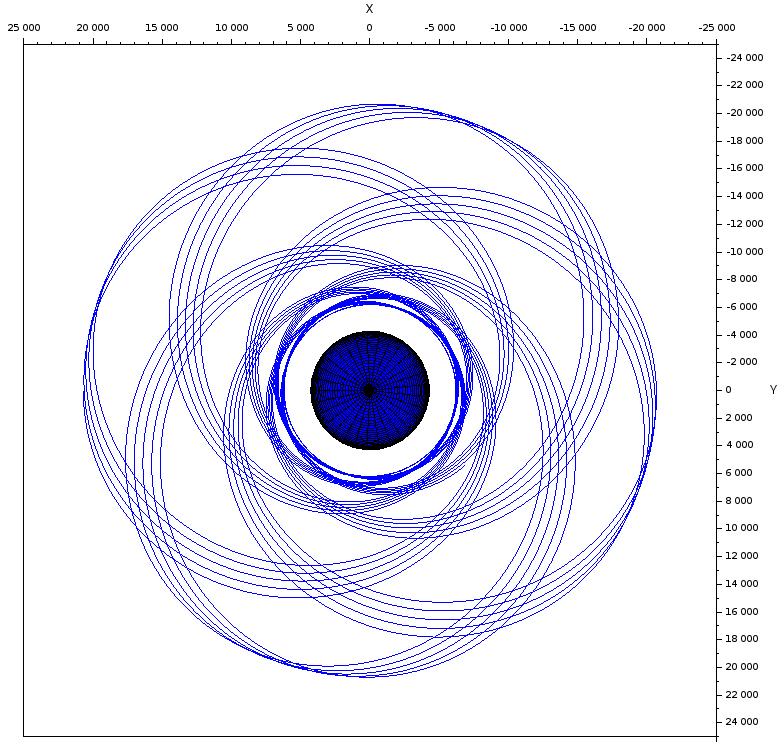}
\caption{$r_0=6.96938+10^{-3}$}
\end{subfigure}
\caption{Orbits with leaves (unit-less initial data).}\label{leaves}
\end{figure}
\end{center}

\begin{center}
\begin{figure}[h!]
\hspace{-.3cm}
\begin{subfigure}[c]{0.35\textwidth}
\centering
\includegraphics[scale=0.58]{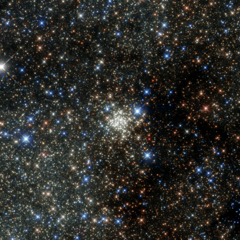}
\caption{The Arches cluster.}
\end{subfigure}
\hspace{-.5mm}
\begin{subfigure}[c]{0.3\textwidth}
\centering
\includegraphics[scale=0.2]{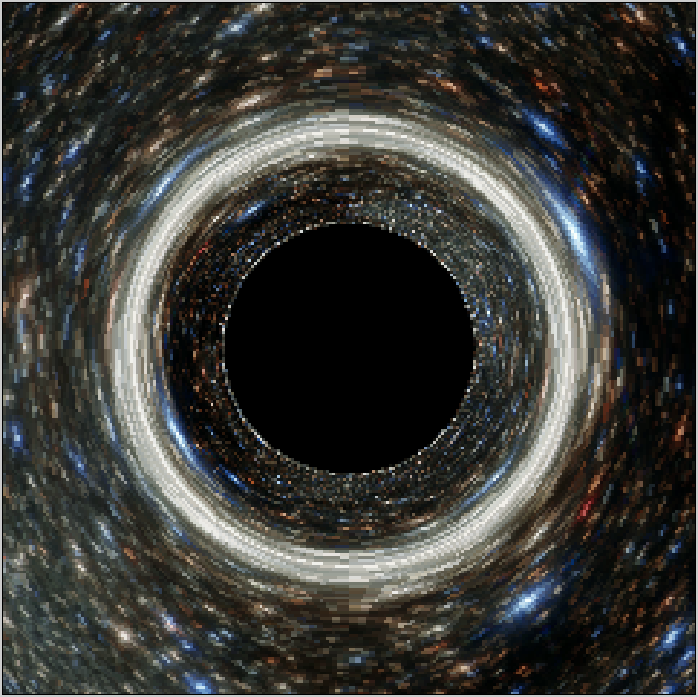}
\caption{Carter (521 s).}
\end{subfigure}
\hspace{5mm}
\begin{subfigure}[c]{0.3\textwidth}
\centering
\includegraphics[scale=0.2]{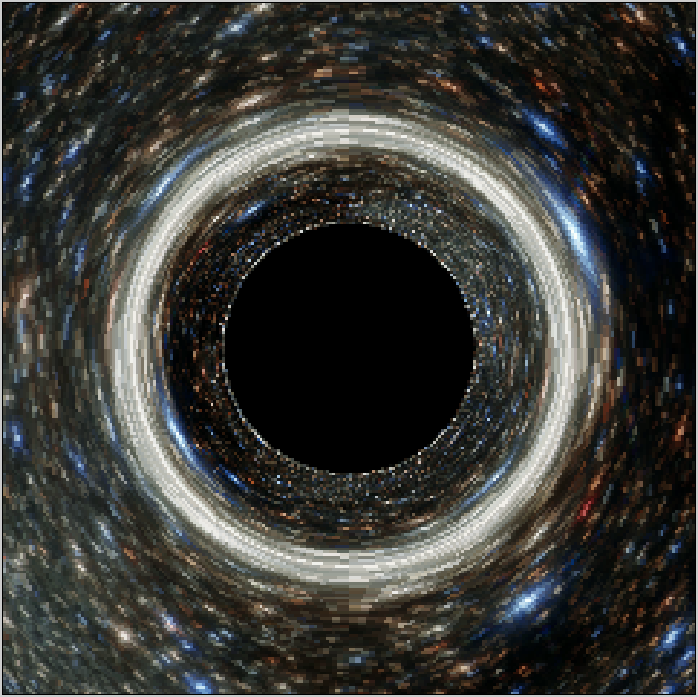}
\caption{Weierstrass (32 s).}
\end{subfigure}
\caption{Carter's and Weierstrass' methods ($a=0$), on a $240^2$ picture from the NASA (original image: \url{https://images.nasa.gov/details/GSFC_20171208_Archive_e000717}).}\label{slow_and_fast}
\end{figure}
\end{center}

\vspace{-7mm}

\begin{center}
\begin{figure}[h!]
\begin{subfigure}[b]{0.45\textwidth}
\centering
\includegraphics[scale=0.18]{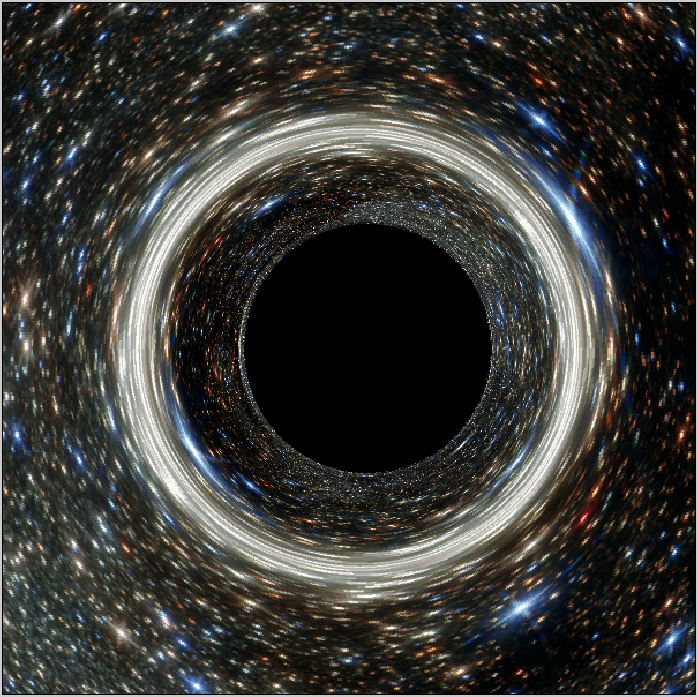}
\caption{$a=0.4$}
\end{subfigure}
\hfill
\begin{subfigure}[b]{0.45\textwidth}
\centering
\includegraphics[scale=0.18]{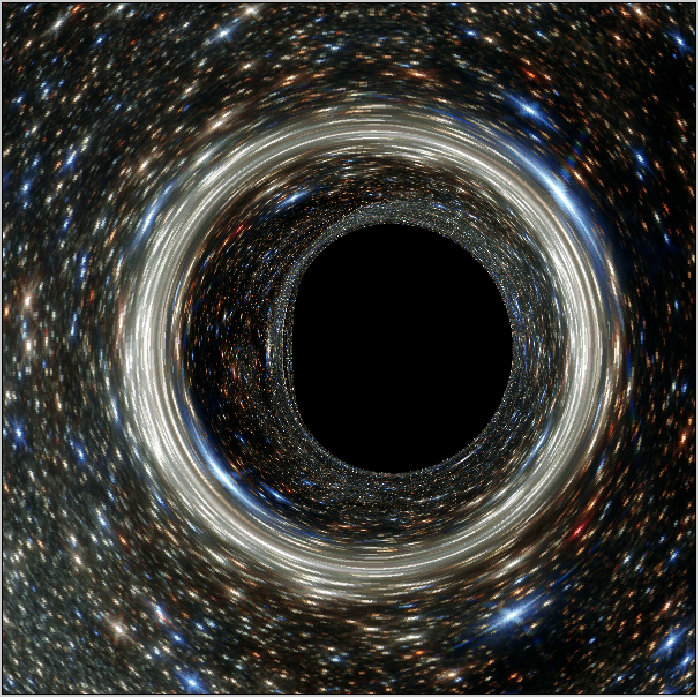}
\caption{$a=0.95$}
\end{subfigure}
\caption{Different Kerr parameters, with $\mathrm{pix}=600^2$, $t_e\approx 5200s$.}\label{KN_4_9}
\end{figure}
\end{center}

\vspace{-7mm}

\newpage
\subsection{Some illustrations and a simulation of M87*}
We finish by giving some figures using \texttt{shadow.sci} and we illustrate the model for the accretion disk described in \S\ref{mod_acc}. We include a simulation of the M87 black hole, using the data from \cite{M87}.

In each case, we label the figure with the resolution $\mathrm{pix}$ and average execution time $t_e$. We also indicate the inclination angle $i$ (from the symmetry axis), the inner (resp. outer) radius $r_i$ (resp. $r_e$) of the accretion disk, the accretion rate $\dot{M}$ used in the equation (\ref{SS}) and the brightness scaling $B_0$. Figure \ref{KN_4_9} displays the well-known influence of the Kerr parameter on the shadow. Figure \ref{effect2} shows the Doppler and gravitational redshifts of the light emitted by an accretion disk radiating as a blackbody, the resulting radiation temperatures being depicted in figure \ref{temper}.

Concerning the modelling of M87*, the black hole and accretion parameters are inspired by \cite{M87} and gathered in Table \ref{M87_para}. As explained in \S\ref{comparison}, the typical order of magnitude of the mass should be $2M_\odot$ for our programs; the mass of M87* is thus rescaled to $1.5M_\odot$. Next, as the physical value $\Lambda\sim10^{-52}{\rm m}^{-2}$ will not visibly affect the picture, we choose to take $\Lambda=0$. Also, we changed some values, marked with a star, which we had to arbitrarily choose (or change from the reference) for the implementation. Besides, we included the combination of the gravitational and Doppler effects in the computation. The result is depicted in Figure \ref{m87_pict}. To make the photon rings even more visible, we also took another picture of the same black hole. The only changed values are $r_i=5.82M$, $r_e=16M$, $\dot{M}=10$, $B_0=5000$. These figures can be compared with \cite[Figs. 8, 10, 11]{M87}.

\begin{table}[h!]
\begin{tabular}{|c|c|}
\hline
\textbf{Parameter} & \textbf{Value} \\
 \hline
 \hline
Mass* & $M=1.5M_{\odot}$ (value in \cite{M87}: $M\approx6.2\cdot10^{9}M_{\odot}$) \\
Cosmological constant & $\Lambda=0$ \\
Charge & $Q=0$ \\
Kerr parameter & $a=0.8$ \\
Accretion rate* & $\dot{M}=3$ \\
Inner radius & $r_i=2.91M$ \\ 
Outer radius* & $r_e=10M$ \\
Brightness rescaling & $B_0=3800$ \\
Angle from symmetry axis & $i=160^\circ$ \\
Resolution & $720$p \\
\hline
\end{tabular}
\captionof{table}{Parameters for the picture of M87* and its accretion disk.}
\label{M87_para}
\end{table}

\vspace{-5mm}

\begin{center}
\begin{figure}[h!]
\begin{subfigure}[c]{0.32\textwidth}
\centering
\includegraphics[scale=0.2]{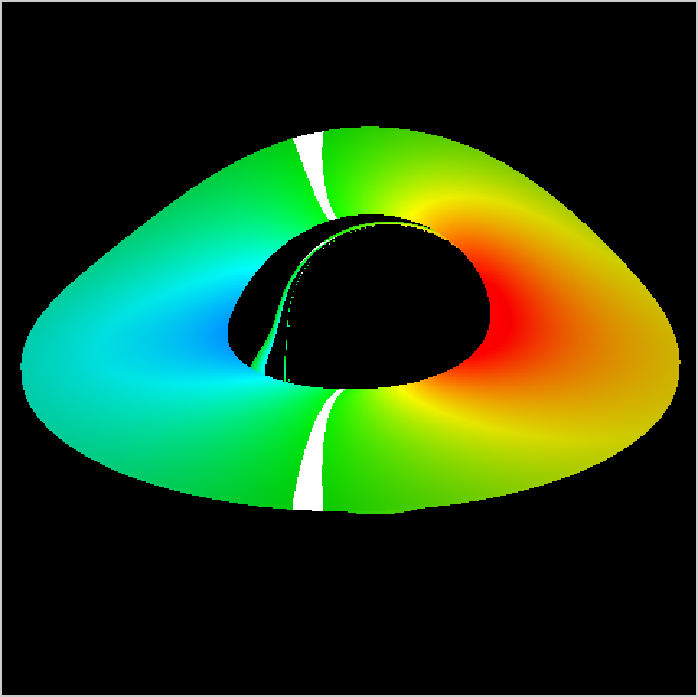}
\caption{Doppler redshift}
\end{subfigure}
\hfill
\begin{subfigure}[c]{0.32\textwidth}
\centering
\includegraphics[scale=0.2]{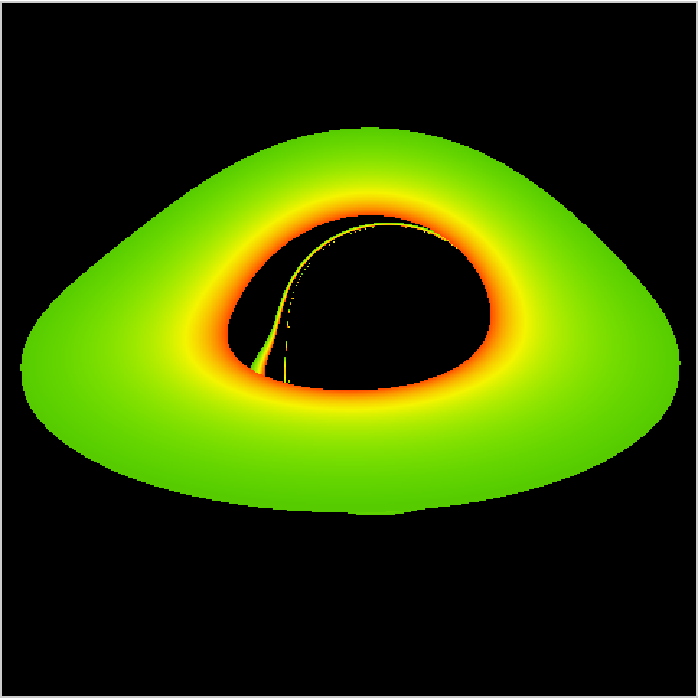}
\caption{Gravitational redshift}
\end{subfigure}
\hfill
\begin{subfigure}[c]{0.32\textwidth}
\centering
\includegraphics[scale=0.2]{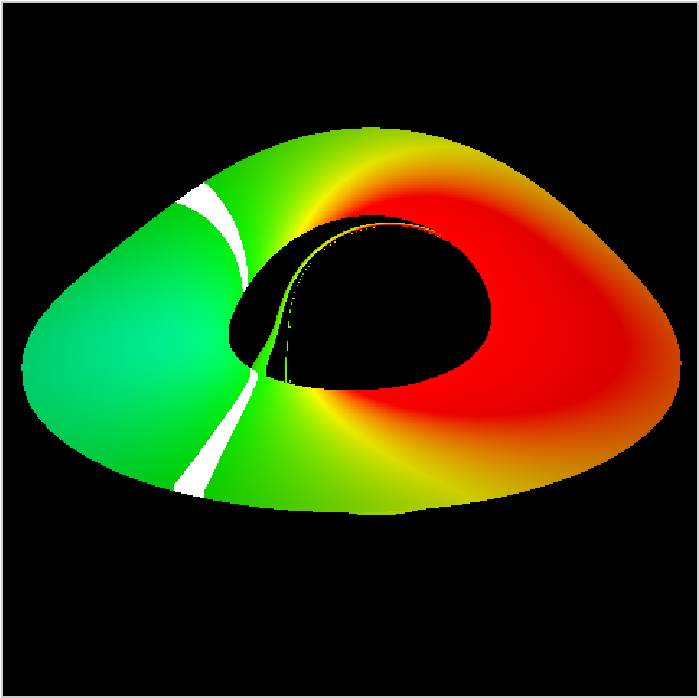}
\caption{Both}
\end{subfigure}
\caption{Different effects on the accretion disk, with $\mathrm{pix}=480^2$, $t_e\approx10^4s$, $i=3\pi/8$, $r_i=4M$, $r_e=12M$. The white strips correspond to where the shift is negligible. Compare with \cite[Figure 1]{bromley97}.}\label{effect2}
\end{figure}
\end{center}

\begin{center}
\begin{figure}[h!]
\begin{subfigure}[c]{0.45\textwidth}
\centering
\includegraphics[scale=0.22]{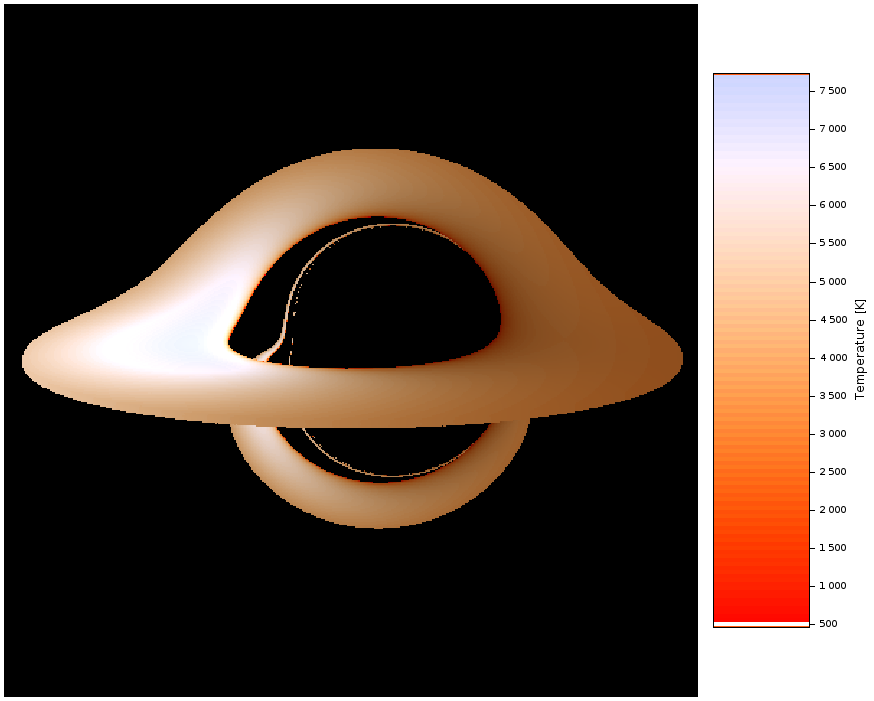}
\caption{$B_0=0$}
\end{subfigure}
\hfill
\begin{subfigure}[c]{0.45\textwidth}
\centering
\includegraphics[scale=0.22]{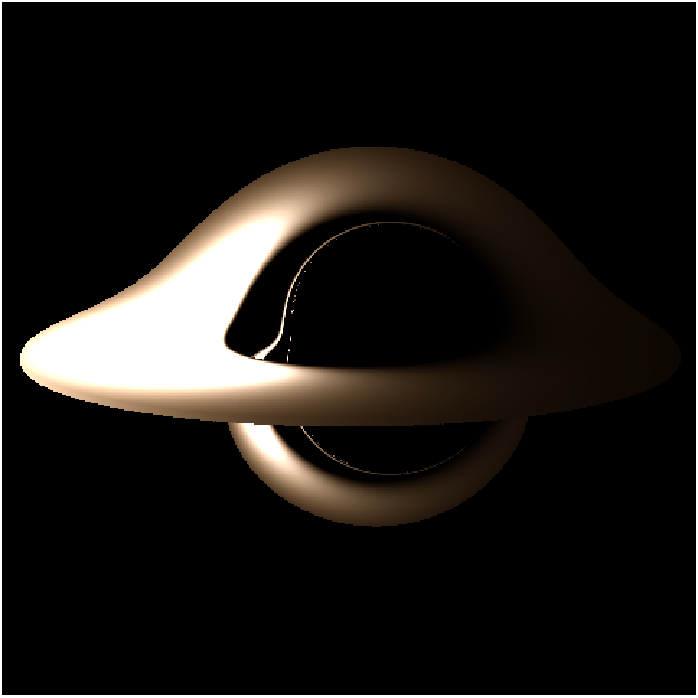}
\caption{$B_0=50$}
\end{subfigure}
\caption{Radiation temperatures, with $\mathrm{pix}=480^2$, $t_e\approx10^4s$, $i=4\pi/9$, $r_i=4M$, $r_e=12M$, $\dot{M}=90$.}\label{temper}
\end{figure}
\end{center}

%\vspace{-12mm}

\begin{center}
\begin{figure}[h!]
\begin{subfigure}[c]{0.292\columnwidth}
\centering
\includegraphics[width=1\linewidth]{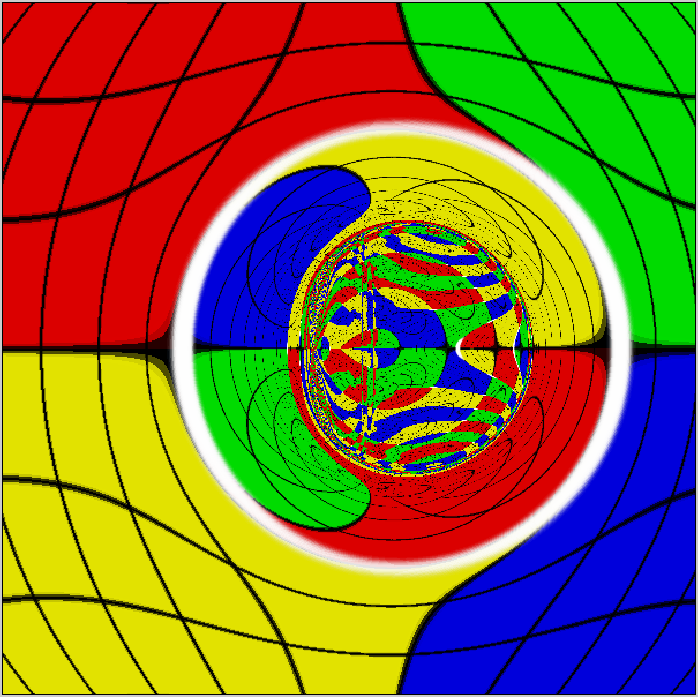}
\caption{$\Lambda=-1.2059\cdot10^{-2}$}
\end{subfigure}
\hfill
\begin{subfigure}[c]{0.292\columnwidth}
\centering
\includegraphics[width=1\linewidth]{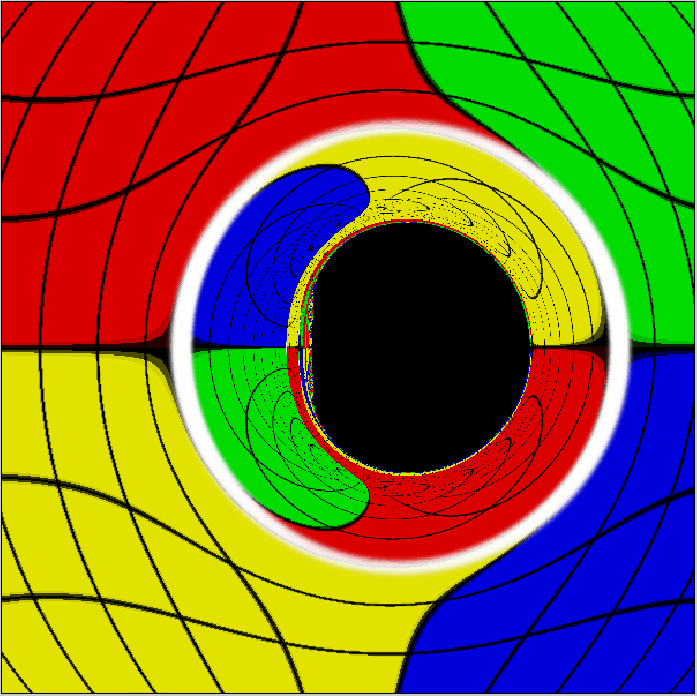}
\caption{$\Lambda=-1.2\cdot10^{-2}$}
\end{subfigure}
\hfill
\begin{subfigure}[c]{0.292\columnwidth}
\centering
\includegraphics[width=1\linewidth]{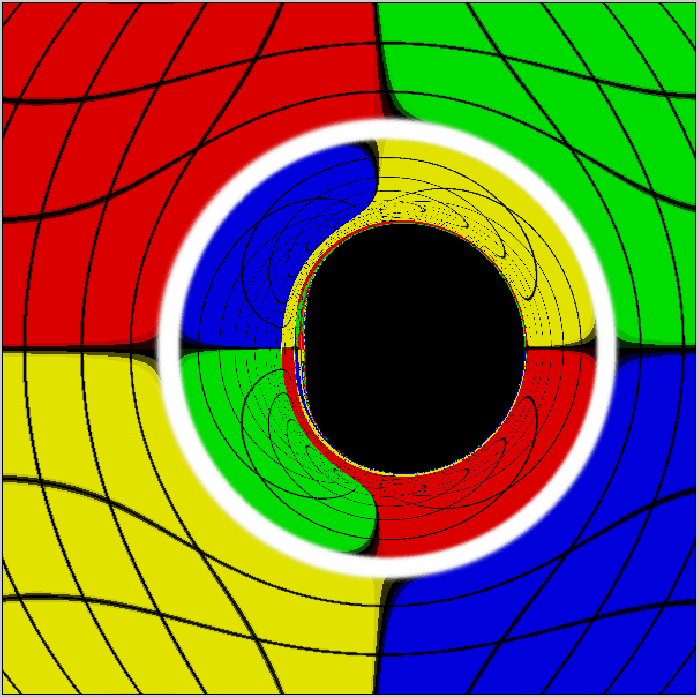}
\caption{$\Lambda=-6\cdot10^{-3}$}
\end{subfigure}\\[.1em]
\begin{subfigure}[c]{0.292\columnwidth}
\centering
\includegraphics[width=1\linewidth]{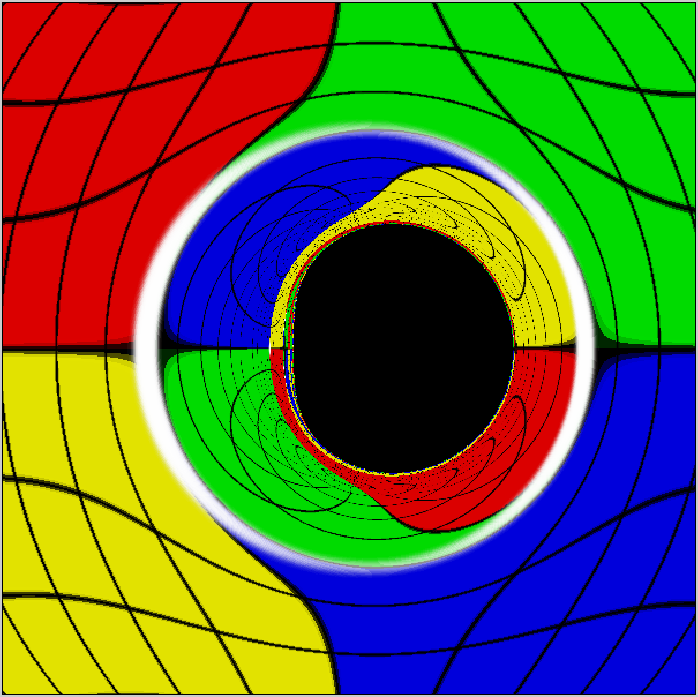}
\caption{$\Lambda=0$}
\end{subfigure}
\hspace{6.8mm}
\begin{subfigure}[c]{0.292\columnwidth}
\centering
\includegraphics[width=1\linewidth]{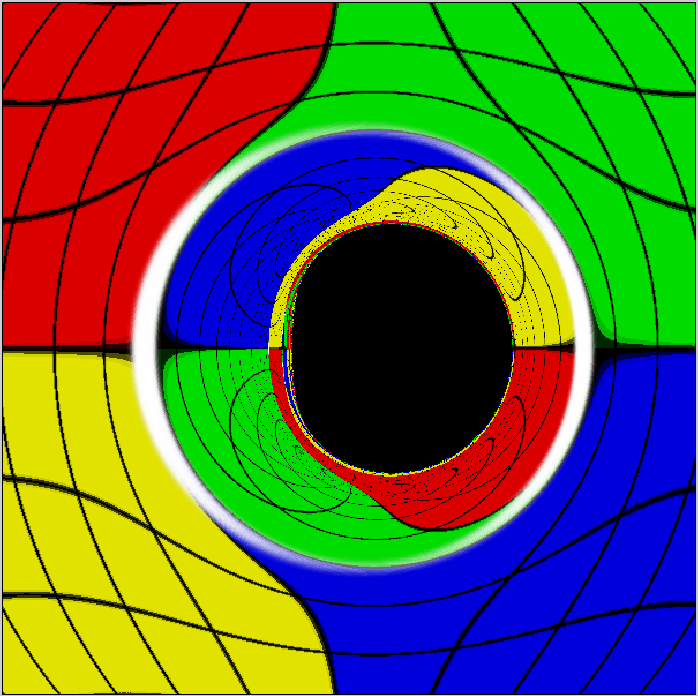}
\caption{$\Lambda=3\cdot10^{-4}$}
\end{subfigure}
\hspace{6.8mm}
\begin{subfigure}[c]{0.292\columnwidth}
\centering
\includegraphics[width=1\linewidth]{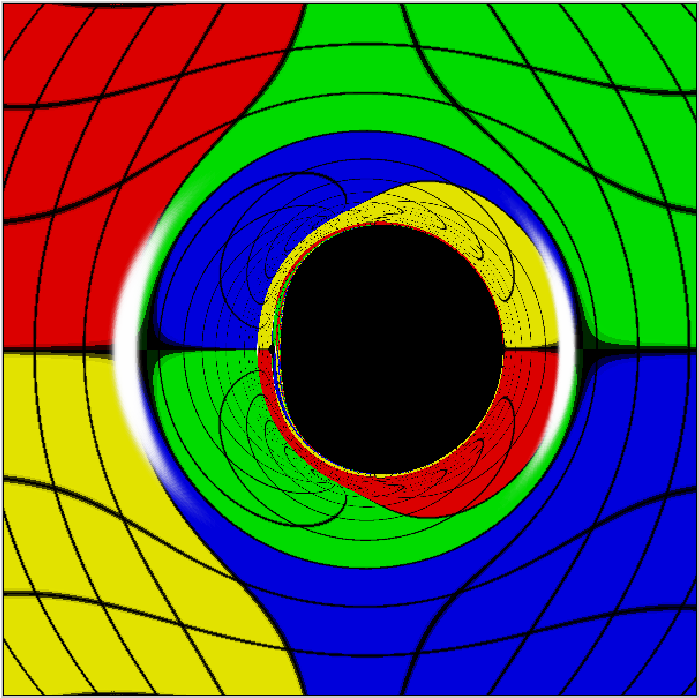}
\caption{$\Lambda=3\cdot10^{-3}$}
\end{subfigure}
\vspace{-.4cm}
\caption{Influence of $\Lambda$ (natural units) on shadows ($\mathrm{pix}=600^2$, $t_e\approx 5300s$ each). The upper left spacetime has no horizon since the quartic $\Delta_r=(1-\lambda r^2)(r^2+0.95^2)-2r+0.3^2$ has no real root for $\Lambda<-1.2034\cdot10^{-2}$.}
\label{lambda_on_shadows}
\end{figure}

%\vspace{-4mm}

\begin{figure}[h!]
\begin{subfigure}[c]{0.292\columnwidth}
\centering
\includegraphics[width=1\linewidth]{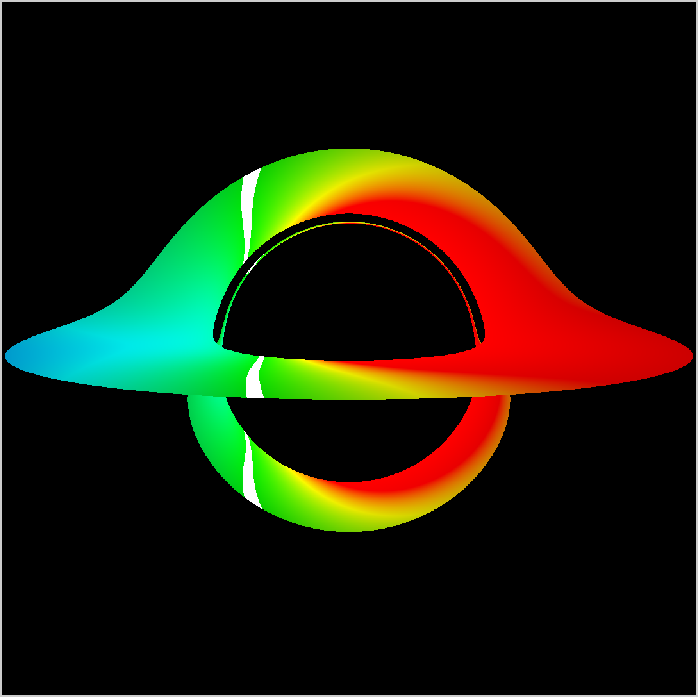}
\end{subfigure}
\hfill
\begin{subfigure}[c]{0.292\columnwidth}
\centering
\includegraphics[width=1\linewidth]{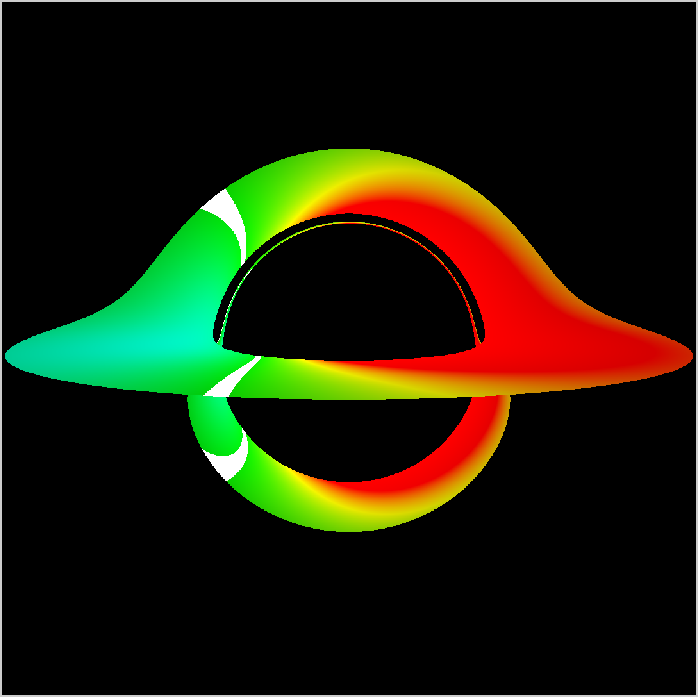}
\end{subfigure}
\hfill
\begin{subfigure}[c]{0.292\columnwidth}
\centering
\includegraphics[width=1\linewidth]{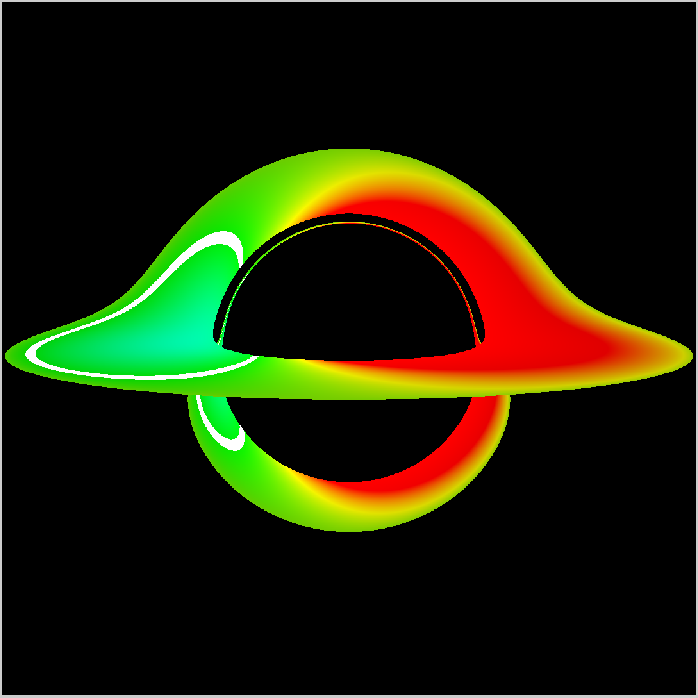}
\end{subfigure}\\[.1em]
\begin{subfigure}[c]{0.292\columnwidth}
\centering
\includegraphics[width=1\linewidth]{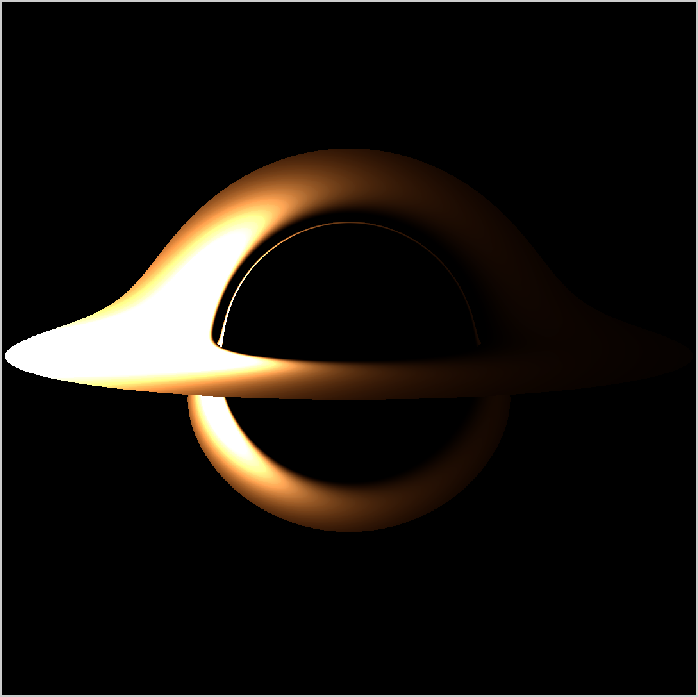}
\caption{$\Lambda=-1.5\cdot10^{-3}$}
\end{subfigure}
\hspace{6.8mm}
\begin{subfigure}[c]{0.292\columnwidth}
\centering
\includegraphics[width=1\linewidth]{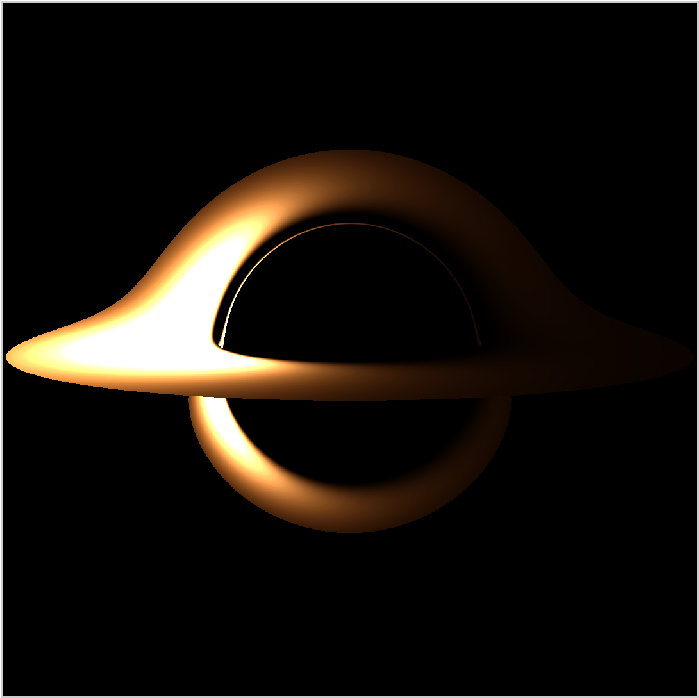}
\caption{$\Lambda=0$}
\end{subfigure}
\hspace{6.8mm}
\begin{subfigure}[c]{0.292\columnwidth}
\centering
\includegraphics[width=1\linewidth]{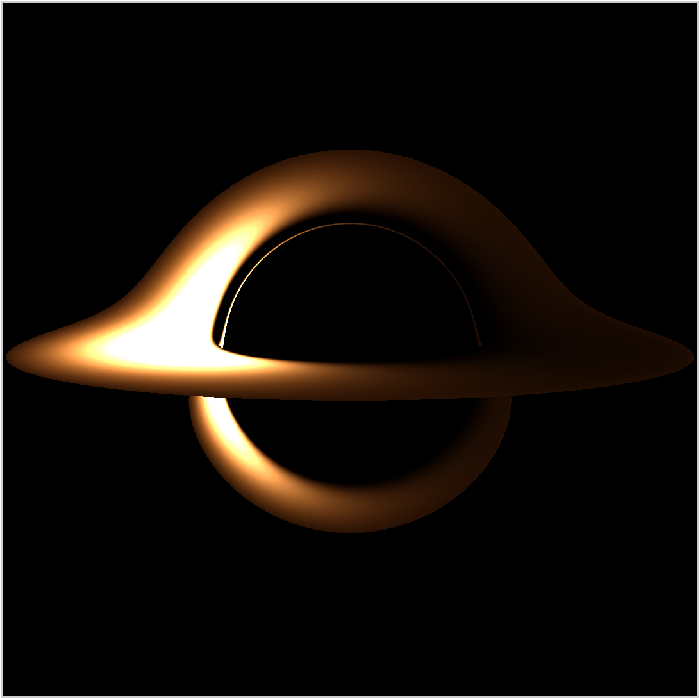}
\caption{$\Lambda=1.5\cdot10^{-3}$}
\end{subfigure}
\vspace{-.4cm}
\caption{Influence of $\Lambda$ (natural units) on the Doppler and gravitational redshift, and on the radiation temperature of an RNdS black hole ($Q=0$), with $\mathrm{pix}=1080^2$; $t_e\approx 2200s$, $i=13\pi/28$, $r_i=4M$, $r_e=12.5M$, $\dot{M}=20$ and $B_0=300$. Compare with \cite[Fig. 3]{gyoto} and \cite[Fig. 13]{osiris}.}
\label{lambda_on_accretion}
\end{figure}
\end{center}

\begin{center}
\begin{figure}[h!]
\begin{subfigure}[c]{0.45\textwidth}
\centering
\includegraphics[scale=0.225]{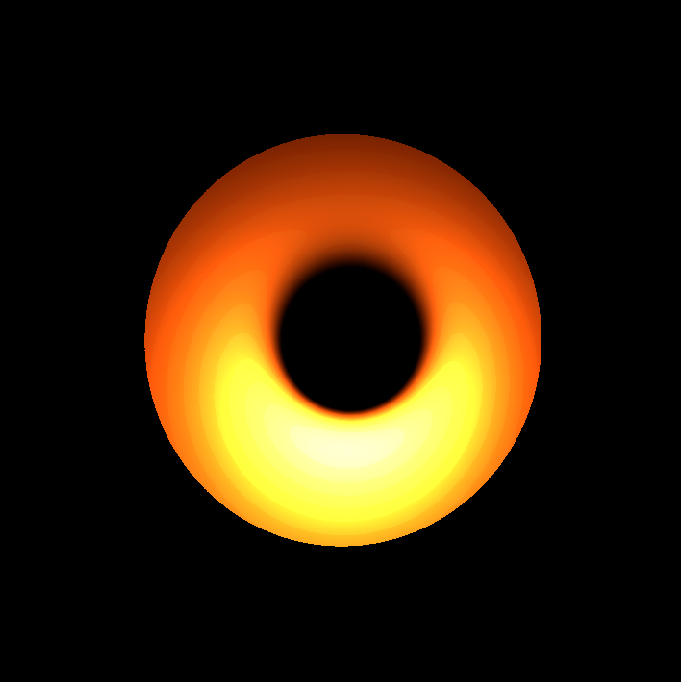}
\caption{$r_i=2.91M$}\label{m877}%5h40
\end{subfigure}
\hfill
\begin{subfigure}[c]{0.45\textwidth}
\centering
\includegraphics[scale=0.22]{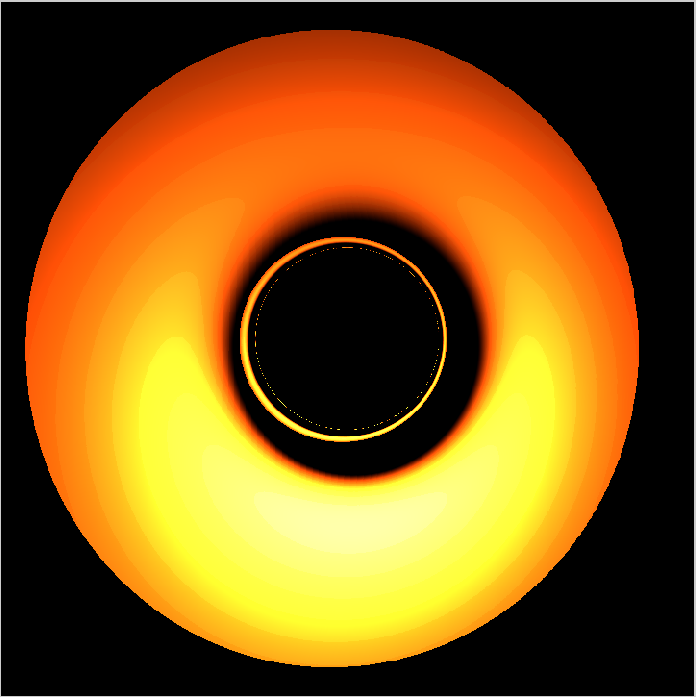}
\caption{$r_i=5.82M$}
\end{subfigure}
\caption{Simulations of M87*.}\label{m87_pict}
\end{figure}
\end{center}

\begin{center}
\vspace{-2cm}
\begin{figure}[h!]
\begin{subfigure}[c]{0.45\textwidth}
\centering
\includegraphics[scale=0.225]{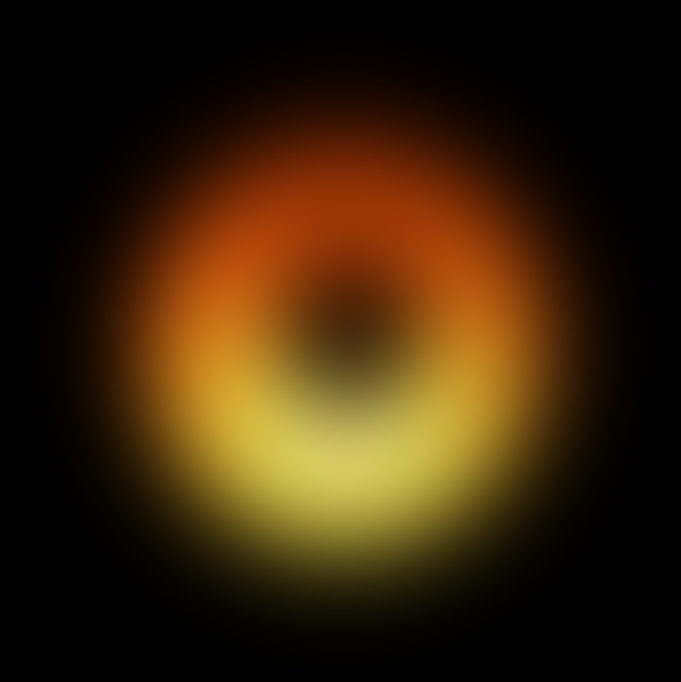}
\caption{Blur}
\end{subfigure}
\hfill
\begin{subfigure}[c]{0.45\textwidth}
\centering
\includegraphics[scale=0.225]{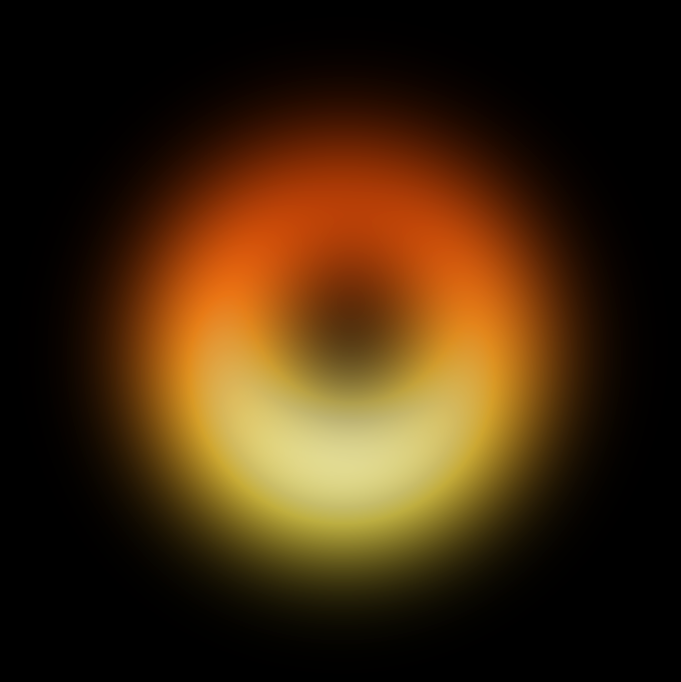}
\caption{Blur and 120\% brighter}
\end{subfigure}
\caption{Blurred versions ($\sigma=50$) of Figure \ref{m877}.}
\label{blurred}
\end{figure}
\end{center}

\newpage%\vspace{-1cm}
\section{Conclusions}

\subsection{Summary}
We present an open ray-tracing code, developed under \texttt{scilab 6.1.1}, for orbit tracing and shadowing of a Kerr--Newman(anti)--de Sitter spacetime\footnote{The package can be found at \url{https://github.com/arthur-garnier/knds_orbits_and_shadows.git}}. In particular and to the knowledge of the author, such a code, including the cosmological constant is new in the literature. It allows to qualitatively explore the visual effects of the cosmological background on spacetime and black hole shadows, as well as the quantitative study of these effects on a single massive or photon orbit. The code is designed so that the user may tune each spacetime parameter and draw the shadow of a KNdS black hole with a background picture of his choosing. The full ray-tracing process is discussed in detail in \S\ref{brt}.

The shadowing program includes a model for a thin accretion disk orbiting the black hole, radiating as a blackbody and taking the gravitational and Doppler redshift into account. Here again, the code is fully customizable: it lets the user choose the inner and outer radii of the disk, its accretion rate, its inclination and to rescale the brightness, or even to add a background picture. Details on the model are given in \S\ref{mod_acc}.

The single orbit integrator allows the user to select an integration method among the ones discussed in \S\ref{formulations} and \S\ref{mot_con} (including some symplectic integrators and Carter's equations). We illustrate its use and compare the different integration methods in Figure \ref{orbit} and gather the quantitative results in Table \ref{compa}.

As explained in \S\ref{comparison}, the best method for shadowing is the numerical integration of Carter's equations. In the case of a spherically symmetric spacetime (that is, a Reissner--Nordstr\"{o}m(anti)--de Sitter black hole), a simple analytic method detailed in \S\ref{wp} and using Weierstrass' elliptic functions is even more efficient for shadowing, as it doesn't require the computation of the full photon paths. This method is the preferred one for shadowing RNdS spacetimes. While well known, all the formulae are precisely stated and derived through the paper, except for some tedious proofs that can be found in the Appendix. This allows the reader to carefully follow the calculations and use these for future codes and works.

Our results and pictures are consistent with the existing literature: for instance, compare the differential systems (\ref{motion}) and (\ref{motion_with_momenta}) with \cite[\S 2]{balek} and \cite[\S 2.1]{PYYY}, or Figures \ref{leaves}, \ref{compa_shadows} and \ref{effect2} with \cite[Figs 14, 15]{perez-giz-levin}, \cite[Fig. 9]{wang-chen-jing} and \cite[Fig. 1]{bromley97}, respectively.

However, pictures displaying the qualitative effects of the cosmological constant on shadows and accretion disks such as in Figures \ref{lambda_on_shadows} and \ref{lambda_on_accretion}, are new. In particular, we numerically retrieve the algebraic fact that, at fixed mass, angular momentum and charge, a negative enough cosmological constant produces a naked singularity, and not a black hole. The approximate limit value given in Figure \ref{lambda_on_shadows} coincides with the algebraically derived one, validating our code again. We also observe in Figure \ref{lambda_on_accretion} that in an anti-de Sitter background, the redshift and radiation temperature are extrapolated and the brightness concentrates on the ``ingoing'' side of the disk. An opposite effect occurs in a de Sitter background: this is in accordance with the rough idea that an (anti-)de Sitter universe is expanding (retracting).

Finally, and as another validation and illustration of our code, in Figure \ref{m87_pict} we display some simulations of the M87 black hole, based on the data from \cite{M87}. The Figure \ref{blurred} depicts a blurred version of the previous one, imitating an actual black hole photograph, such as the well-known \cite[Fig. 3]{EHT_I}.

\subsection{Discussion and perspectives}
One of the aims of this work is to produce an open, transparent and easy-to-use code for customizable KNdS black hole shadowing, which can be useful for instance in educational purposes or to experience extreme cosmological conditions, without being a ray-tracing expert. In particular, we wanted the code to be easy to install and the \texttt{scilab} environment is a suitable choice, as it includes the handy \texttt{IPCV} package for image processing. This approach comes with a negative side: the code is not optimal and rather slow in comparison to other ray-tracing codes, such as \texttt{GRay} or \texttt{gyoto} (\cite{GRay,gyoto}). Hence, it would be suitable to parallelize the code using a GPU, but this is not an easy task on \texttt{scilab}. To this end, the code has to be translated into a more flexible language first, such as Python for instance. The author plans to do so in the future.

In our programs, we only considered a thin and steady accretion disk. It lacks of dynamical effects and it could be interesting to include models for vortices in the accretion disk and Rossby wave instability for example, such as in \cite[Fig. 4]{gyoto}. However, it would increase the length and execution times of the programs, which should be optimized first.

Also, as indicated in \S\ref{brt}, the shadowing program is designed for a common plane background picture and, in order do avoid distortions, we took the compromise of projecting the image only on a celestial hemisphere, while the other hemisphere is projected on by a mirrored version of the original image. Therefore, it could be of interest to adapt the code and to produce a panoramic version of it.

Finally, our approach dramatically relies on the fact that the KNdS metric is an exact solution of Einstein's field equation. However, there are physical situations where the metric can only be numerically approximated, such as the merging of neutron stars, the study of gravitational waves and binary pulsars or, more generally, the two-body problem in general relativity. Therefore, a nice goal for future works would be to extend the code to an approximated metric, using for instance a 3+1 decomposition.

%\subsection*{Data availability statement} All data that support the findings of this study are included within the article (and any supplementary files).

\subsection*{Acknowledgments} We are much grateful to the referees for their careful reading and their precise and constructive suggestions, which greatly improved the paper. We also deeply thank Daniel Juteau and Antoine Bourget for their help in the correction process.

%\newpage
\begin{appendix}
\section{Some proofs}\label{some_proofs}

\subsection{Proof of the Theorem \ref{EFE_ok}}\label{some_proofs_1}
\hfill

First, we express the metric in Kerr coordinates, using its \emph{principal null geodesics}. More precisely, consider the trajectory of a photon in the plane $\theta\equiv\pi/2$ with total energy $E=1$ and total (azimuthal) angular momentum $L=aE$. Using equations (\ref{motion}), we see that the corresponding four-velocity is given by
\[(\dot{t},\dot{r},\dot{\theta},\dot{\phi})=\left(\frac{\chi^2(r^2+a^2)}{\Delta_r},\pm\chi,0,\frac{a\chi^2}{\Delta_r}\right).\]
By rescaling the affine parameter $\lambda\leadsto \lambda/\chi^2$, and choosing the ingoing geodesic with $\dot{r}<0$, the velocity is given by
\[l^\mu=\left(\frac{r^2+a^2}{\Delta_r},\frac{-1}{\chi},0,\frac{a}{\Delta_r}\right).\]
Now, the coordinates $u$ and $\overline{\phi}$ replacing $t$ and $\phi$ respectively, should be chosen to be constant along this world line, that is, we want $\mathrm{d}u/\mathrm{d}r=\mathrm{d}\overline{\phi}/\mathrm{d}r=0$. Hence, we introduce
\[\left\{\begin{array}{c}
u:=t+T(r),\\[.5em]
\overline{\phi}:=\phi+\Phi(r),\end{array}\right.\]
where $T$ and $\Phi$ are respectively given by
\[T(r)=\chi\int_0^r\frac{\varrho^2+a^2}{\Delta_r(\varrho)}\mathrm{d}\varrho,~~\Phi(r)=a\chi\int_0^r\frac{\mathrm{d}\varrho}{\Delta_r(\varrho)}-\frac{\pi}{2}.\]
The constant in the definition of $\Phi$ ensures that the Kerr--Schild variables approach oblate spheroidal coordinates as $M$ and $\Lambda$ go to zero and doesn't change the metric. The 1-forms $\mathrm{d}t$ and $\mathrm{d}\phi$ can be expressed as
\[\mathrm{d}t=\mathrm{d}u-\mathrm{d}T=\mathrm{d}u-\frac{\chi(r^2+a^2)\mathrm{d}r}{\Delta_r},~~\mathrm{d}\phi=\mathrm{d}\overline{\phi}-\mathrm{d}\Phi=\mathrm{d}\overline{\phi}-\frac{a\chi\mathrm{d}r}{\Delta_r}.\]
Then, the metric in these new coordinates reads
\begin{align*}
\mathrm{d}s^2=&\frac{-\Delta_r}{\chi^2\Sigma}\left(\mathrm{d}u-\frac{\chi(r^2+a^2)\mathrm{d}r}{\Delta_r}-a\sin^2\theta\left(\mathrm{d}\overline{\phi}-\frac{a\chi\mathrm{d}r}{\Delta_r}\right)\right)^2
\\
&+\frac{\Delta_\theta\sin^2\theta}{\chi^2\Sigma}\left(a\left(\mathrm{d}u-\frac{\chi(r^2+a^2)\mathrm{d}r}{\Delta_r}\right)-(r^2+a^2)\left(\mathrm{d}\overline{\phi}-\frac{a\chi\mathrm{d}r}{\Delta_r}\right)\right)^2+\Sigma\left(\frac{\mathrm{d}r^2}{\Delta_r}+\frac{\mathrm{d}\theta^2}{\Delta_\theta}\right)
\\[1em]
=&\frac{a^2\sin^2\theta\Delta_\theta-\Delta_r}{\chi^2\Sigma}\mathrm{d}u^2+\frac{\sin^2\theta}{\chi^2\Sigma}\left[\Delta_\theta(r^2+a^2)^2-\Delta_ra^2\sin^2\theta\right]\mathrm{d}\overline{\phi}^2
\\
&+\frac{\Sigma\mathrm{d}\theta^2}{\Delta_\theta}+\frac{2\mathrm{d}u\mathrm{d}r}{\chi}+\frac{2a\sin^2\theta}{\chi^2\Sigma}(\Delta_r-\Delta_\theta(r^2+a^2))\mathrm{d}u\mathrm{d}\overline{\phi}-\frac{2a\sin^2\theta}{\chi}\mathrm{d}r\mathrm{d}\overline{\phi}.
\\[1em]
=&\frac{-1}{\chi^2\Sigma}\left[a^2\sin^4\theta\Delta_r\mathrm{d}\overline{\phi}^2+(\Delta_r-a^2\sin^2\theta\Delta_\theta)\mathrm{d}u^2-2a\sin^2\theta(\Delta_r-(r^2+a^2)\Delta_\theta)\mathrm{d}u\mathrm{d}\overline{\phi}\right]
\\
&+\frac{\Sigma\mathrm{d}\theta^2}{\Delta_\theta}+\frac{2\mathrm{d}u\mathrm{d}r}{\chi}-\frac{2a\sin^2\theta}{\chi}\mathrm{d}r\mathrm{d}\overline{\phi}
\\[1em]
=&\frac{-1}{\chi^2\Sigma}\left[\Delta_r(\mathrm{d}u-a\sin^2\theta\mathrm{d}\overline{\phi})^2-\sin^2\theta\Delta_\theta(a\mathrm{d}u-(r^2+a^2)\mathrm{d}\overline{\phi})^2\right]+\frac{\Sigma\mathrm{d}\theta^2}{\Delta_\theta}+\frac{2\mathrm{d}u\mathrm{d}r}{\chi}-\frac{2a\sin^2\theta}{\chi}\mathrm{d}r\mathrm{d}\overline{\phi}.
\end{align*}
Defining $\overline{t}:=u-r$, we obtain the KNdS metric in \emph{Kerr coordinates} $(\overline{t},r,\theta,\overline{\phi})$:
\begin{equation}\label{kerr_coords}
\mathrm{d}s^2=\frac{\Delta_\theta\sin^2\theta\left[a\mathrm{d}\overline{t}+a\mathrm{d}r-(r^2+a^2)\mathrm{d}\overline{\phi}\right]^2-\Delta_r\left[\mathrm{d}\overline{t}+\mathrm{d}r-a\sin^2\theta\mathrm{d}\overline{\phi}\right]^2}{\chi^2\Sigma}+\frac{\Sigma\mathrm{d}\theta^2}{\Delta_\theta}+\frac{2\mathrm{d}r}{\chi}\left[\mathrm{d}\overline{t}+\mathrm{d}r-a\sin^2\theta\mathrm{d}\overline{\phi}\right].
\end{equation}
Since $\Lambda a^2>-3$, we have $\Delta_\theta>0$ and therefore the above metric is well-defined everywhere except for $\Sigma=0$. However, computing the determinant of this metric, we find the same result as for Boyer--Lindquist coordinates, that is
\[\det(g_{\mu\nu})=-\chi^{-4}{\Sigma^2\sin^2\theta}\]
and thus it is not clear yet that the metric is Lorentzian (non-degenerate) because $\det(g_{\mu\nu})=0$ for $\theta=0,\pi$. Therefore, we still have to transform the metric.

Following \cite[II, (2.6)]{BL} we define the following ``Cartesian'' Kerr--Schild coordinates
\begin{equation}\label{KS_coords}
\left\{\begin{array}{l}
x:=\sqrt{r^2+a^2}\sin(\theta)\cos(\overline{\phi}+\arctan(a/r)),\\[.5em]
y:=\sqrt{r^2+a^2}\sin(\theta)\sin(\overline{\phi}+\arctan(a/r)),\\[.5em]
z:=r\cos(\theta).
\end{array}\right.
\end{equation}
The variable $r$ becomes a function of $(x,y,z)$ implicitly defined by the relation
\[\frac{x^2+y^2}{r^2+a^2}+\frac{z^2}{r^2}=1.\]
We differentiate (\ref{KS_coords}) and after some manipulations, we find the following relations
\[z\mathrm{d}z=-r^2\cos\theta\sin\theta\mathrm{d}\theta+r\cos^2\theta\mathrm{d}r,~~x\mathrm{d}x+y\mathrm{d}y+z\mathrm{d}z=a^2\cos\theta\sin\theta\mathrm{d}\theta+r\mathrm{d}r\]
as well as
\[x\mathrm{d}y-y\mathrm{d}x=(r^2+a^2)\sin^2\theta\mathrm{d}\overline{\phi}-a\sin^2\theta\mathrm{d}r.\]
Thus, we find the expressions
\[\left\{\begin{array}{l}
\mathrm{d}r=\frac{r^2(x\mathrm{d}x+y\mathrm{d}y)+(r^2+a^2)z\mathrm{d}z}{r\Sigma},\\[.5em]
\mathrm{d}\theta=\frac{x\mathrm{d}x+y\mathrm{d}y+z\mathrm{d}z-r\mathrm{d}r}{a^2\cos\theta\sin\theta}=\frac{\cos^2\theta(x\mathrm{d}x+y\mathrm{d}y+z\mathrm{d}z)-z\mathrm{d}z}{\Sigma\cos\theta\sin\theta},\\[.5em]
\mathrm{d}\overline{\phi}=\frac{a\sin^2\theta\mathrm{d}r+x\mathrm{d}y-y\mathrm{d}x}{(r^2+a^2)\sin^2\theta}.\end{array}\right.\]
These yield
\[a\mathrm{d}r-(r^2+a^2)\mathrm{d}\overline{\phi}=\frac{y\mathrm{d}x-x\mathrm{d}y}{\sin^2\theta},~~\mathrm{d}r-a\sin^2\theta\mathrm{d}\overline{\phi}=\frac{r(x\mathrm{d}x+y\mathrm{d}y)+a(y\mathrm{d}x-x\mathrm{d}y)}{r^2+a^2}+\frac{z\mathrm{d}z}{r}\]
and we may now express the metric in these new variables as (we keep one ``$\mathrm{d}r$'' for now in order to simplify the notation, but we'll give the full expression below)
\begin{align*}
\mathrm{d}s^2=&\left[\frac{2\mathrm{d}r}{\chi}-\frac{\Delta_r}{\chi^2\Sigma}\left(\mathrm{d}\overline{t}+\frac{r(x\mathrm{d}x+y\mathrm{d}y)+a(y\mathrm{d}x-x\mathrm{d}y)}{r^2+a^2}+\frac{z\mathrm{d}z}{r}\right)\right]\left(\mathrm{d}\overline{t}+\frac{r(x\mathrm{d}x+y\mathrm{d}y)+a(y\mathrm{d}x-x\mathrm{d}y)}{r^2+a^2}+\frac{z\mathrm{d}z}{r}\right)
\\
&+\frac{\Delta_\theta \sin^2\theta}{\chi^2\Sigma}\left(a\mathrm{d}\overline{t}+\frac{y\mathrm{d}x-x\mathrm{d}y}{\sin^2\theta}\right)^2+\frac{(\cos^2\theta(x\mathrm{d}x+y\mathrm{d}y+z\mathrm{d}z)-z\mathrm{d}z)^2}{\Sigma\Delta_\theta\cos^2\theta\sin^2\theta}.
\end{align*}
At this point we can formally compute the determinant of the metric and obtain
\[\det((g_{\mu\nu})_{\mathrm{KS}})=-\chi^{-4}\ne0,\]
so that the metric is Lorentzian where it is defined. The first line above indeed is a smooth differential 2-form except on $\{\Sigma=0\}$, so all we have to do is transform the second line. Consider the auxiliary spacial metric
\[\mathrm{d}\sigma^2:=\frac{\Delta_\theta(y\mathrm{d}x-x\mathrm{d}y)^2}{\chi^2\Sigma\sin^2\theta}+\frac{(\cos^2\theta(x\mathrm{d}x+y\mathrm{d}y+z\mathrm{d}z)-z\mathrm{d}z)^2}{\Sigma\Delta_\theta\cos^2\theta\sin^2\theta}.\]
Developing and factorizing this expression yields
\begin{align*}
\mathrm{d}\sigma^2=&\left(\frac{x^2\cos^2\theta}{\Delta_\theta}+\frac{y^2\Delta_\theta}{\chi^2}\right)\frac{\mathrm{d}x^2}{\Sigma\sin^2\theta}+\left(\frac{y^2\cos^2\theta}{\Delta_\theta}+\frac{x^2\Delta_\theta}{\chi^2}\right)\frac{\mathrm{d}y^2}{\Sigma\sin^2\theta}+\left(\frac{\cos^2\theta}{\Delta_\theta}-\frac{\Delta_\theta}{\chi^2}\right)\frac{2xy\mathrm{d}x\mathrm{d}y}{\Sigma\sin^2\theta}
\\
&+\frac{\sin^2\theta z^2\mathrm{d}z^2}{\Sigma\Delta_\theta\cos^2\theta}-\frac{2z\mathrm{d}z}{\Sigma\Delta_\theta}(x\mathrm{d}x+y\mathrm{d}y).
\end{align*}
First, we transform the coefficient of $\mathrm{d}x^2$. Using (\ref{KS_coords}), we have
\[\cos^2\theta=\frac{z^2}{r^2},~~\sin^2\theta=\frac{x^2+y^2}{r^2+a^2},~~\text{so that}~~\Delta_\theta=1+\lambda a^2\cos^2\theta=\frac{r^2+\lambda a^2z^2}{r^2},~\Sigma=\frac{r^4+a^2z^2}{r^2},\]
and recalling that $\chi=1+\lambda a^2$, we compute
\begin{align*}
\frac{1}{\Sigma\sin^2\theta}\left(\frac{x^2\cos^2\theta}{\Delta_\theta}+\frac{y^2\Delta_\theta}{\chi^2}\right)=&\frac{x^2\cos^2\theta(1+\lambda a^2)^2+y^2(1+\lambda a^2\cos^2\theta)^2}{\chi^2\Sigma\Delta_\theta\sin^2\theta}
\\[.5em]
=&\frac{x^2\cos^2\theta+y^2+\lambda^2 a^4\cos^2\theta(x^2+y^2\cos^2\theta)+2\lambda a^2\cos^2\theta(x^2+y^2)}{\chi^2\Sigma\Delta_\theta\sin^2\theta}
\\[.5em]
=&\frac{(x^2+y^2)(\Delta_\theta+\chi\lambda a^2\cos^2\theta)-\sin^2\theta(x^2+y^2\lambda^2 a^4\cos^2\theta)}{\chi^2\Sigma\Delta_\theta\sin^2\theta}
\\[.5em]
=&\frac{(r^2+a^2)(\Delta_\theta+\chi\lambda a^2\cos^2\theta)-x^2-y^2\lambda^2 a^4\cos^2\theta}{\chi^2\Sigma\Delta_\theta}
\\[.5em]
=&\frac{(r^2+a^2)(r^2+\lambda a^2z^2+\chi\lambda a^2 z^2)-r^2x^2-y^2\lambda^2 a^4z^2}{\chi^2r^2\Sigma\Delta_\theta}
\end{align*}
so that
\[
\frac{1}{\Sigma\sin^2\theta}\left(\frac{x^2\cos^2\theta}{\Delta_\theta}+\frac{y^2\Delta_\theta}{\chi^2}\right)=r^2\frac{(r^2+a^2)(r^2+\lambda a^2z^2(1+\chi))-r^2x^2-\lambda^2a^4y^2z^2}{\chi^2(r^4+a^2z^2)(r^2+\lambda a^2z^2)}.
\]
Exchanging $x$ and $y$ gives a similar expression for the coefficient of $\mathrm{d}y^2$. Now we treat the $\mathrm{d}x\mathrm{d}y$ term. We have
\begin{align*}
\frac{2xy}{\Sigma\sin^2\theta}\left(\frac{\cos^2\theta}{\Delta_\theta}-\frac{\Delta_\theta}{\chi^2}\right)=&\frac{2xy\left[\chi^2\cos^2\theta-\Delta_\theta^2\right]}{\chi^2\Sigma\Delta_\theta\sin^2\theta}=\frac{2xy\left[\cos^2\theta(1+\lambda a^2)^2-(1+\lambda a^2\cos^2\theta)^2\right]}{\chi^2\Sigma\Delta_\theta\sin^2\theta}
\\[.5em]
=&\frac{2xy\left[\cos^2\theta+\lambda^2a^4\cos^2\theta-1-\lambda^2a^4\cos^4\theta\right]}{\chi^2\Sigma\Delta_\theta\sin^2\theta}=\frac{2xy(\lambda^2a^4\cos^2\theta-1)}{\chi^2\Sigma\Delta_\theta}
\\[.5em]
=&\frac{2xyr^2(\lambda^2a^4z^2-r^2)}{\chi^2(r^4+a^2z^2)(r^2+\lambda a^2z^2)}.
\end{align*}
Finally, we compute the remaining terms:
\begin{align*}
\frac{\sin^2\theta z^2\mathrm{d}z^2}{\Sigma\Delta_\theta\cos^2\theta}-\frac{2z\mathrm{d}z}{\Sigma\Delta_\theta}(x\mathrm{d}x+y\mathrm{d}y)=&\frac{r^2(x^2+y^2)\mathrm{d}z^2}{(r^2+a^2)\Sigma\Delta_\theta}-\frac{2z\mathrm{d}z}{\Sigma\Delta_\theta}(x\mathrm{d}x+y\mathrm{d}y)
\\[.5em]
=&\frac{r^6(x^2+y^2)\mathrm{d}z^2}{(r^2+a^2)(r^4+a^2z^2)(r^2+\lambda a^2z^2)}-\frac{2zr^4\mathrm{d}z(x\mathrm{d}x+y\mathrm{d}y)}{(r^4+a^2z^2)(r^2+\lambda a^2z^2)}
\\[.5em]
=&\frac{r^4\mathrm{d}z\left[r^2(x^2+y^2)\mathrm{d}z-2z(r^2+a^2)(x\mathrm{d}x+y\mathrm{d}y)\right]}{(r^2+a^2)(r^4+a^2z^2)(r^2+\lambda a^2z^2)}.
\end{align*}
Gathering all, we obtain the full expression of the metric in Kerr--Schild coordinates:
\begin{align*}
\mathrm{d}s^2=&\left[\frac{2\mathrm{d}r}{\chi}-\frac{\Delta_r}{\chi^2\Sigma}\left(\mathrm{d}\overline{t}+\frac{r(x\mathrm{d}x+y\mathrm{d}y)+a(y\mathrm{d}x-x\mathrm{d}y)}{r^2+a^2}+\frac{z\mathrm{d}z}{r}\right)\right]\left(\mathrm{d}\overline{t}+\frac{r(x\mathrm{d}x+y\mathrm{d}y)+a(y\mathrm{d}x-x\mathrm{d}y)}{r^2+a^2}+\frac{z\mathrm{d}z}{r}\right)
\\
&+\mathrm{d}\sigma^2+\frac{a\Delta_\theta\mathrm{d}\overline{t}}{\chi^2\Sigma}(a\sin^2\theta\mathrm{d}\overline{t}+2(y\mathrm{d}x-x\mathrm{d}y)),
%%+\frac{a^2\sin^2\theta\Delta_\theta\mathrm{d}\overline{t}^2}{\chi^2\Sigma}+\frac{2a\Delta_\theta\mathrm{d}\overline{t}(y\mathrm{d}x-x\mathrm{d}y)}{\chi^2\Sigma},
\end{align*}
that is,
\begin{align*}%\label{In_KS_coords}
\stepcounter{equation}\tag{\theequation}\label{In_KS_coords} \mathrm{d}s^2=&\frac{2(r^2(x\mathrm{d}x+y\mathrm{d}y)+(r^2+a^2)z\mathrm{d}z)}{\chi(r^4+a^2z^2)}\left(r\mathrm{d}\overline{t}+\frac{r^2(x\mathrm{d}x+y\mathrm{d}y)+ar(y\mathrm{d}x-x\mathrm{d}y)}{r^2+a^2}+z\mathrm{d}z\right)
\\
&-\frac{\Delta_r}{\chi^2(r^4+a^2z^2)}\left(r\mathrm{d}\overline{t}+\frac{r^2(x\mathrm{d}x+y\mathrm{d}y)+ar(y\mathrm{d}x-x\mathrm{d}y)}{r^2+a^2}+z\mathrm{d}z\right)^2
\\
&+\frac{a(r^2+\lambda a^2z^2)\mathrm{d}\overline{t}}{\chi^2(r^4+a^2z^2)}\left(\frac{a(x^2+y^2)\mathrm{d}\overline{t}}{r^2+a^2}+2(y\mathrm{d}x-x\mathrm{d}y)\right)
\\
&+\frac{r^2}{(r^4+a^2z^2)(r^2+\lambda a^2z^2)}\left[\frac{(r^2+a^2)(r^2+\lambda a^2z^2(1+\chi))-r^2x^2-\lambda^2a^4y^2z^2}{\chi^2}\mathrm{d}x^2\right.\\
&+\left.\frac{(r^2+a^2)(r^2+\lambda a^2z^2(1+\chi))-r^2y^2-\lambda^2a^4x^2z^2}{\chi^2}\mathrm{d}y^2+\frac{2xy(\lambda^2a^4z^2-r^2)}{\chi^2}\mathrm{d}x\mathrm{d}y\right.\\
&+\left.\frac{r^2\mathrm{d}z}{r^2+a^2}(r^2(x^2+y^2)\mathrm{d}z-2z(r^2+a^2)(x\mathrm{d}x+y \mathrm{d}y))\right],
%%\\
%%&+\frac{r^2}{\chi^2(r^4+a^2z^2)(r^2+\lambda a^2z^2)}\left\{{(r^2+a^2)(r^2+\lambda a^2z^2(1+\chi))}(\mathrm{d}x^2+\mathrm{d}y^2)-{r^2}(x^2\mathrm{d}x^2+y^2\mathrm{d}y^2)\right.\\
%%&-{\lambda^2a^4z^2}(y^2\mathrm{d}x^2+x^2\mathrm{d}y^2)+{2xy(\lambda^2a^4z^2-r^2)}\mathrm{d}x\mathrm{d}y+\frac{\chi^2r^2\mathrm{d}z}{r^2+a^2}\left.\left[r^2(x^2+y^2)\mathrm{d}z-2z(r^2+a^2)(x\mathrm{d}x+y \mathrm{d}y)\right]\right\}
\end{align*}
where $r=r(x,y,z)$ is the positive solution\footnote{explicitly: $r=\tfrac{1}{\sqrt{2}}\sqrt{x^2+y^2+z^2-a^2+\sqrt{a^2(a^2-2x^2-2y^2+2z^2)+(x^2+y^2+z^2)^2}}$} of $\tfrac{x^2+y^2}{r^2+a^2}+\tfrac{z^2}{r^2}=1$. It is now manifest that this metric is well-defined everywhere except on the ring $\{\Sigma=0\}=\{z=0,~x^2+y^2=a^2\}$.

On the other hand, the Kretschmann scalar $K=R_{\alpha\beta\gamma\delta}R^{\alpha\beta\gamma\delta}$ for the KNdS metric has recently been computed by Kraniotis \cite[Theorem 1]{kraniotis22} and is given by
\begin{align*}
K=&\frac{8}{\Sigma^6}\left[3\lambda^2a^{12}\cos^{12}\theta+18\lambda^2a^{10}r^2\cos^{10}\theta+45\lambda^2a^8r^4\cos^8\theta+6\cos^6\theta(10\lambda^2a^6r^6-a^6M^2)\right.\\
&\left.+a^4\cos^4\theta(45\lambda^2r^8+90M^2r^2-60MQ^2r+7Q^4)\right.\\
&\left.+a^2r^2\cos^2\theta(18\lambda^2r^2-90M^2r^2+120MQ^2r-34Q^4) +3\lambda^2r^{12}+6M^2r^6-12MQ^2r^5+7Q^4r^4\right].
\end{align*}
It is singular exactly on $\{\Sigma=0\}$ and thus the analytic extension of the KNdS metric to $\mathcal{M}\setminus\{\Sigma=0\}$ is maximal.

To prove that the potential $\mathbf{A}=Qr\chi^{-1}\Sigma^{-1}(\mathrm{d}t-a\sin^2\theta\mathrm{d}\phi)$ solves the vacuum Maxwell equations and that the KNdS metric solves the associated EME, by smoothness of the metric and the potential, it suffices to check them on the dense open chart $\{\sin\theta(1-\cos\phi)\Delta_r\Sigma\ne0\}$ and this relies on tedious but elementary calculations. Details are given in Appendix \ref{proofEFE}.
\qed

\begin{rem}
By the Christodoulou--Ruffini mass formula (see \cite[\S 4, formula 57]{pradhan}), when $M\to0$, the irreducible mass approaches $\sqrt{-Q^2}/2$ and then also $Q\to0$. Using the notation of the previous proof, we have
\begin{align*}
\lim_{M\to0}\overline{\phi}=&\phi+\lim_{M\to0}\int_0^r\frac{a\chi\mathrm{d}\varrho}{\Delta_r(\varrho)}-\frac{\pi}{2}=\phi+\int_0^r\frac{a\chi\mathrm{d}\varrho}{(1-\lambda \varrho^2)(\varrho^2+a^2)}-\frac{\pi}{2} \\
=&\phi+\sqrt{\lambda}\mathrm{argth}(r\sqrt{\lambda})+\arctan\left(\frac{r}{a}\right)-\frac{\pi}{2}.
\end{align*}
Therefore, when $M\to0$, the Kerr--Schild coordinates (\ref{KS_coords}) read
\[\left\{\begin{array}{l}
x=\sqrt{r^2+a^2}\sin(\theta)\cos(\phi+\sqrt{\lambda}\mathrm{argth}(r\sqrt{\lambda})),\\[.5em]
y=\sqrt{r^2+a^2}\sin(\theta)\sin(\phi+\sqrt{\lambda}\mathrm{argth}(r\sqrt{\lambda})),\\[.5em]
z=r\cos(\theta),
\end{array}\right.\]
where it is understood that $\sqrt{\lambda}\mathrm{argth}(r\sqrt{\lambda})=-\sqrt{|\lambda|}\arctan(r\sqrt{|\lambda|})$ for $\lambda<0$. We also have
\[\lim_{M\to0}\overline{t}=t-r+\lim_{M\to0}\int_0^r\frac{\chi(\varrho^2+a^2)\mathrm{d}\varrho}{\Delta_r(\varrho)}=t-r+\frac{\chi}{\sqrt{\lambda}}\mathrm{argth}(r\sqrt{\lambda}),\]
with the convention $\mathrm{argth}(r\sqrt{\lambda})/\sqrt{\lambda}=\arctan(r\sqrt{|\lambda|})/\sqrt{|\lambda|}$ for $\lambda<0$. Hence, with our convention, the Kerr--Schild coordinates coincide with the usual oblate spheroidal coordinates only for $\lambda\to0$ and in this case, the Kerr--Schild and Boyer--Lindquist times agree.
\end{rem}

\begin{rem}
The formula (\ref{kerr_coords}) may be rewritten as
\[\mathrm{d}s^2=\mathrm{d}s^2_0+\frac{2Mr-Q^2}{\chi^2\Sigma}(\mathrm{d}\overline{t}+\mathrm{d}r-a\sin^2\theta\mathrm{d}\overline{\phi})^2,\]
where $\mathrm{d}s^2_0$ is the KNdS metric with $M=Q=0$. In Boyer--Lindquist coordinates, we have
\[\mathrm{d}s^2_0=-\frac{\widetilde{\Delta_r}}{\chi^2\Sigma}(\mathrm{d}t-a\sin^2\theta\mathrm{d}\phi)^2+\frac{\Delta_\theta\sin^2\theta}{\chi^2\Sigma}(a\mathrm{d}t-(r^2+a^2)\mathrm{d}\phi)^2+\Sigma\left(\frac{\mathrm{d}r^2}{\widetilde{\Delta_r}}+\frac{\mathrm{d}\theta^2}{\Delta_\theta}\right),\]
with $\widetilde{\Delta_r}=(1-\lambda r^2)(r^2+a^2)$. Following \cite[\S 4.2]{hoque}, define new coordinates $(T,R,\Theta,\Phi)$:
\[T:=t/\chi,~~R^2:=\tfrac{1}{\chi}(r^2\Delta_\theta+a^2\sin^2\theta),~~R\cos\Theta=r\cos\theta,~~\Phi=\phi-\tfrac{a\lambda}{\chi}T.\]
Then the metric $\mathrm{d}s_0^2$ becomes
\[\mathrm{d}s_0^2=-(1-\lambda R^2)\mathrm{d}T^2+\frac{\mathrm{d}R^2}{1-\lambda R^2}+R^2(\mathrm{d}\Theta^2+\sin^2\Theta\mathrm{d}\Phi^2),\]
which is the usual de Sitter metric. We obtain the \emph{Kerr--Schild form} of the KNdS metric: the flat de Sitter metric plus a perturbation term. In Kerr--Schild coordinates, we have
\[\mathrm{d}s^2=\mathrm{d}s^2_0+\frac{2Mr-Q^2}{\chi^2(r^4+a^2z^2)}\left(r\mathrm{d}\overline{t}+\frac{r^2(x\mathrm{d}x+y\mathrm{d}y)+ar(y\mathrm{d}x-x\mathrm{d}y)}{r^2+a^2}+z\mathrm{d}z\right)^2.\]
\end{rem}

\subsection{Proof of the Theorem \ref{carter_equations}}\label{some_proofs_2}
\hfill

First, we compute the Hamiltonian explicitly:
\begin{align*}
2\mathcal{H}(\gamma,p)\stackrel{\tiny{\text{df}}}=&g^{\mu\nu}(p_\mu-eA_\mu)(p_\nu-eA_\nu)=g^{tt}\left(p_t-\frac{eQr}{\chi\Sigma}\right)^2+2g^{t\phi}\left(p_t-\frac{eQr}{\chi\Sigma}\right)\left(p_\phi+\frac{eQra\sin^2\theta}{\chi\Sigma}\right)\\
&+g^{\phi\phi}\left(p_\phi+\frac{eQra\sin^2\theta}{\chi\Sigma}\right)^2+g^{rr}p_r^2+g^{\theta\theta}p_\theta^2 \\[.5em]
=&\frac{\chi^2}{\Sigma\Delta_r\Delta_\theta}(a^2\sin^2\theta\Delta_r-(r^2+a^2)^2\Delta_\theta)\left(E+\frac{eQr}{\chi\Sigma}\right)^2+\frac{\chi^2}{\Sigma}\left(\frac{1}{\sin^2\theta\Delta_\theta}-\frac{a^2}{\Delta_r}\right)\left(L+\frac{eQra\sin^2\theta}{\chi\Sigma}\right)^2 \\
&-\frac{2a\chi^2}{\Sigma\Delta_r\Delta_\theta}(\Delta_r-(r^2+a^2)\Delta_\theta)\left(L+\frac{eQra\sin^2\theta}{\chi\Sigma}\right)\left(E+\frac{eQr}{\chi\Sigma}\right)+\frac{\Delta_r}{\Sigma}p_r^2+\frac{\Delta_\theta}{\Sigma}p_\theta^2.
\end{align*}
Formally developing and factorizing, we find
\begin{align*}
\mu=&\frac{1}{\Sigma^3\Delta_r\Delta_\theta\sin^2\theta}\left\{-2\chi eQr\Sigma^2\sin^2\theta\Delta_\theta(E(r^2+a^2)-aL)-r^2e^2Q^2a^4\sin^6\theta\Delta_\theta\right. \\
&\left.+a^2\sin^4\theta(2e^2r^2Q^2(r^2+a^2)\Delta_\theta+E^2\chi^2\Sigma^2\Delta_r)+L^2\chi^2\Sigma^2\Delta_r\right. \\
&\left.+\sin^2\theta\left[\Sigma^2\Delta_\theta^2\Delta_rp_\theta^2+\Delta_\theta(\Sigma^2(\Delta_r^2p_r^2-\chi^2(E(r^2+a^2)-aL)^2)-e^2r^2Q^2(r^2+a^2)^2)-2ELa\chi^2\Sigma^2\Delta_r\right]\right\} \\[.5em]
=&\frac{2\chi eQr}{\Sigma\Delta_r}(aL-E(r^2+a^2))-\frac{e^2Q^2r^2a^4\sin^4\theta}{\Sigma^2\Delta_r}+\frac{2a^2e^2Q^2r^2\sin^2\theta(r^2+a^2)}{\Sigma^2\Delta_r}+\frac{a^2\chi^2E^2\sin^2\theta}{\Sigma\Delta_\theta}\\
&+\frac{L^2\chi^2}{\sin^2\theta\Sigma\Delta_\theta}-\frac{2ELa\chi^2}{\Sigma\Delta_\theta}-\frac{\chi^2}{\Sigma\Delta_r}(E(r^2+a^2)-aL)^2-\frac{e^2Q^2r^2(r^2+a^2)^2}{\Sigma^2\Delta_r}+\frac{\Delta_r}{\Sigma}p_r^2+\frac{\Delta_\theta}{\Sigma}p_\theta^2\\[.5em]
=&\frac{e^2Q^2r^2}{\Sigma^3\Delta_r}\underset{=-\Sigma^2}{\underbrace{(2a^2(r^2+a^2)\sin^2\theta-a^4\sin^4\theta-(r^2+a^2)^2)}}+\frac{L^2\chi^2}{\Sigma\Delta_\theta\sin^2\theta}+\frac{\Delta_r}{\Sigma}p_r^2+\frac{\Delta_\theta}{\Sigma}p_\theta^2\\
&+\frac{aL-E(r^2+a^2)}{\Sigma\Delta_r}(2\chi eQr+\chi^2(E(r^2+a^2)-aL))+\frac{aE\chi^2}{\Sigma\Delta_\theta}(aE\sin^2\theta-2L)
\end{align*}
and so
\begin{align*}
\mu=&-\frac{e^2Q^2r^2}{\Sigma\Delta_r}+\frac{\chi(aL-E(r^2+a^2)}{\Sigma\Delta_r}(2eQr+\chi(E(r^2+a^2)-aL))+\frac{a\chi^2E}{\Sigma\Delta_\theta}(aE\sin^2\theta-2L)\\
&+\frac{L^2\chi^2}{\Sigma\Delta_\theta\sin^2\theta}+\frac{\Delta_r}{\Sigma}p_r^2+\frac{\Delta_\theta}{\Sigma}p_\theta^2\\[.5em]
=&-\frac{\chi^2}{\Sigma\Delta_r}\left(E(r^2+a^2)-aL+\frac{eQr}{\chi}\right)^2+\frac{\chi^2(aE\sin^2\theta-L)^2}{\Sigma\Delta_\theta\sin^2\theta}+\frac{\Delta_r}{\Sigma}p_r^2+\frac{\Delta_\theta}{\Sigma}p_\theta^2 \\[.5em]
=&\frac{1}{\Sigma}\left(-\frac{W_r^2}{\Delta_r}+\frac{W_\theta^2}{\Delta_\theta}+\Delta_r p_r^2+\Delta_\theta p_\theta^2\right).
\end{align*}
Therefore, the Hamiltonian is given by
\begin{equation}\label{hamil}
\mathcal{H}(\gamma,p)=\frac{1}{2}\left(\frac{W_\theta^2}{\Sigma\Delta_\theta}-\frac{W_r^2}{\Sigma\Delta_r}+\frac{\Delta_r}{\Sigma}p_r^2+\frac{\Delta_\theta}{\Sigma}p_\theta^2\right).
\end{equation}
Consider now the action integral
\[S:=\int^\ell \mathcal{L}(\gamma,\dot{\gamma})\mathrm{d}l=\int^\ell p_\mu\dot{\gamma}^\mu-\mathcal{H}(\gamma,p)\mathrm{d}l.\]
Then, we have $p_\mu={\partial S}/{\partial {\gamma}^\mu}$ and the motion equations can be expressed as the \emph{Hamilton--Jacobi equation}
\begin{equation}\label{HJ}
\frac{\partial S}{\partial\ell}=\mathcal{H}\left(\gamma,\frac{\partial S}{\partial{\gamma}}\right).
\end{equation}
Using (\ref{hamil}), this amounts to say that
\[\mu=2\frac{\partial S}{\partial\ell}=\frac{W_\theta^2}{\Sigma\Delta_\theta}-\frac{W_r^2}{\Sigma\Delta_r}+\frac{\Delta_r}{\Sigma}\left(\frac{\partial S}{\partial r}\right)^2+\frac{\Delta_\theta}{\Sigma}\left(\frac{\partial S}{\partial \theta}\right)^2.\]
This equation may be rewritten in the separated form
\[\frac{W_\theta^2}{\Delta_\theta}+\Delta_\theta \left(\frac{\partial S}{\partial \theta}\right)^2-\mu a^2\cos^2\theta=\frac{W_r^2}{\Delta_r}-\Delta_r \left(\frac{\partial S}{\partial r}\right)^2+\mu r^2\]
so that each side of this equation is equal to some constant $\kappa\in\R$, as in the statement. But since we have $\Delta_\nu p_\nu=\Sigma \dot{\nu}$ for $\nu=r,\theta$, we obtain the equations for $\dot{r}^2$ and $\dot{\theta}^2$ as claimed.

Now, the equations for $\dot{t}$ and $\dot{\phi}$ may be derived as follows. We compute
\[g^{t\mu}p_\mu=g^{t\mu}(g_{\mu\nu}\dot{\gamma}^{\nu}+eA_\mu)=\dot{t}+eg^{t\mu}A_\mu=\dot{t}+\frac{eQr}{\chi\Sigma}(g^{tt}-a\sin^2\theta g^{t\phi})=\dot{t}-\frac{\chi eQr(r^2+a^2)}{\Sigma\Delta_r}\]
so that we get
\[\dot{t}=\frac{\chi eQr(r^2+a^2)}{\Sigma\Delta_r}+g^{tt}p_t+g^{t\phi}p_\phi=\frac{\chi eQr(r^2+a^2)}{\Sigma\Delta_r}-E g^{tt}+L g^{t\phi}=\frac{\chi W_r (r^2+a^2)}{\Sigma\Delta_r}-\frac{a\chi W_\theta\sin\theta}{\Sigma\Delta_\theta}\]
and we proceed in the same way for $\dot{\phi}$: we have $g^{\phi\mu}p_\mu=\dot{\phi}-\tfrac{eQra\chi}{\Sigma\Delta_r}$ and thus
\[\dot{\phi}=\frac{eQra\chi}{\Sigma\Delta_r}+g^{t\phi}p_t+g^{\phi\phi}p_\phi=\frac{eQra\chi}{\Sigma\Delta_r}-E g^{t\phi}+L g^{\phi\phi}=\frac{a\chi W_r}{\Sigma\Delta_r}-\frac{\chi W_\theta}{\Sigma\Delta_\theta\sin\theta}.\]
\qed

\begin{landscape}
\subsection{Proof that the KNdS metric solves the Einstein-Maxwell equations}\label{proofEFE}
The Christoffel symbols of the KNdS metric are as follows (their symmetry allows to lighten the notation with bullets):
\small{\[
{\Gamma^t}_{\mu\nu}=\tfrac{1}{\Sigma^2}\begin{pmatrix}0&{\frac{(r^2+a^2)\left(\tfrac{1}{2}\Sigma\Delta'_r+r(\Delta_\theta{a}^{2}\sin^2\theta-\Delta_r)\right)}{\Delta_r}}&-{\frac{{a}^{2}\left(\tfrac{1}{2}\sin^2\theta\Sigma\Delta'_\theta+\sin\theta\cos\theta(({a}^{2}+{r}^{2})\Delta_\theta-\Delta_r)\right)}{\Delta_\theta}}&0\\
\noalign{\medskip}\bullet&0&0&-{\frac{a\sin^2\theta\left(\tfrac{1}{2}({a}^{2}+{r}^{2})\Sigma\Delta'_r-r(\Delta_r(\Sigma+r^2+a^2)-({a}^{2}+{r}^{2})^{2}\Delta_\theta)\right)}{\Delta_r}}\\
\noalign{\medskip}\bullet&\bullet&0&{\frac{a\sin^2\theta\left(\tfrac{1}{2}({a}^{2}+{r}^{2})\Sigma\Delta'_\theta+{a}^{2}\sin\theta\cos\theta(({a}^{2}+{r}^{2})\Delta_\theta-\Delta_r)\right)}{\Delta_\theta}}\\
\noalign{\medskip}\bullet&\bullet&\bullet&0\end{pmatrix}
\]

\[
{\Gamma^r}_{\mu\nu}=\tfrac{1}{\Sigma}\begin{pmatrix}{\frac{\Delta_r\left(\tfrac{1}{2}\Sigma\Delta'_r+r(a^2\sin^2\theta\Delta_\theta-\Delta_r)\right)}{\Sigma^{2}{\chi}^{2}}}&0&0&-{\frac{a\sin^2\theta\Delta_r\left(\tfrac{1}{2}\Sigma\Delta'_r+r(a^2\sin^2\theta\Delta_\theta-\Delta_r)\right)}{\Sigma^{2}{\chi}^{2}}}\\
\noalign{\medskip}\bullet&r-{\frac{\Sigma\Delta'_r}{2\Delta_r}}&-{a}^{2}\sin\theta\cos\theta&0\\
\noalign{\medskip}\bullet&\bullet&-{\frac{\Delta_rr}{\Delta_\theta}}&0\\
\noalign{\medskip}\bullet&\bullet&\bullet&{\frac{\Delta_r\sin^2\theta\left(\tfrac{1}{2}{a}^{2}\sin^2\theta\Sigma\Delta'_r-r({a}^{2}\cos^2\theta(2\Delta_\theta({a}^{2}+{r}^{2})-\Delta_r)+(r^4-a^4)\Delta_\theta+{a}^{2}\Delta_r)\right)}{\Sigma^{2}{\chi}^{2}}}\end{pmatrix}
\]

\[
{\Gamma^\theta}_{\mu\nu}=\tfrac{1}{\Sigma}\begin{pmatrix}-{\frac{{a}^{2}\sin^2\theta\Delta_\theta\left(\tfrac{1}{2}\Sigma\Delta'_\theta+\cot\theta(({a}^{2}+{r}^{2})\Delta_\theta-\Delta_r)\right)}{\chi^2\Sigma^{2}}}&0&0&{\frac{a\sin^2\theta\Delta_\theta(r^2+a^2)\left(\tfrac{1}{2}\Sigma\Delta'_\theta+\cot\theta(({a}^{2}+{r}^{2})\Delta_\theta-\Delta_r)\right)}{\chi^2\Sigma^{2}}}\\
\noalign{\medskip}\bullet&{\frac{{a}^{2}\cos\theta\sin\theta\Delta_\theta}{\Delta_r}}&r&0\\
\noalign{\medskip}\bullet&\bullet&-{a}^{2}\cos\theta\sin\theta-\frac{\Sigma\Delta'_\theta}{2\Delta_\theta}&0\\
\noalign{\medskip}\bullet&\bullet&\bullet&-{\frac{\Delta_\theta\sin^2\theta\left(\tfrac{1}{2}({a}^{2}+{r}^{2})^{2}\Sigma\Delta'_\theta+\cot\theta(({a}^{2}+{r}^{2})^{3}\Delta_\theta-{a}^{2}\sin^2\theta\Delta_r(\Sigma+r^2+a^2))\right)}{\chi^2\Sigma^{2}}}\end{pmatrix}
\]

\[
{\Gamma^\phi}_{\mu\nu}=\tfrac{1}{\Sigma^2}\begin{pmatrix}0&{\frac{a\left(\tfrac{1}{2}\Sigma\Delta'_r+r(a^2\sin^2\theta\Delta_\theta-\Delta_r)\right)}{\Delta_r}}&-{\frac{a\left(\tfrac{1}{2}\sin\theta\Sigma\Delta'_\theta+\cos\theta(({a}^{2}+{r}^{2})\Delta_\theta-\Delta_r)\right)}{\Delta_\theta\sin\theta}}&0\\
\noalign{\medskip}\bullet&0&0&-{\frac{\tfrac{1}{2}{a}^{2}\sin^2\theta\Sigma\Delta'_r+r({a}^{2}+{r}^{2})(a^2\sin^2\theta\Delta_\theta-\Delta_r)}{\Delta_r}}\\
\noalign{\medskip}\bullet&\bullet&0&{\frac{\tfrac{1}{2}\sin\theta({a}^{2}+{r}^{2})\Sigma\Delta'_\theta+\cos\theta(\Delta_\theta((a^2+r^2)^2-a^2\sin^2\theta\Sigma)-{a}^{2}\sin^2\theta\Delta_r)}{\Delta_\theta\sin\theta}}\\
\noalign{\medskip}\bullet&\bullet&\bullet&0\end{pmatrix}
\]}
\end{landscape}

\begin{landscape}
From these we deduce that the only non-zero components of the Ricci tensor are the following:
\begin{align*}
R_{tt}=\tfrac{1}{\chi^2\Sigma^3}&\left[-\tfrac{1}{2}{a}^{2}\sin^2\theta\Sigma\Delta_\theta\Delta''_\theta+\tfrac{1}{2}\Sigma\Delta_r\Delta''_r-{a}^{2}\sin\theta\cos\theta\left(\left(\tfrac{1}{2}\Sigma+r^2+a^2\right)\Delta_\theta+\Delta_r\right)\Delta'_\theta\right.\\
&~\left.-r({a}^{2}\sin^2\theta\Delta_\theta+\Delta_r)\Delta'_r-{a}^{2}\sin^2\theta(a^2-r^2)\Delta_\theta^{2}-2a^2\cos^2\theta\Delta_\theta\Delta_r+\Delta_r^{2}\right]\\
=~~~~~~~&\hspace{-7.5mm}\frac{Q^2(\Delta_r+a^2\sin^2\theta\Delta_\theta)}{\chi^2\Sigma^3}-\frac{\Lambda(\Delta_r-a^2\sin^2\theta\Delta_\theta)}{\chi^2\Sigma},
\end{align*}

\begin{align*}
R_{t\phi}=R_{\phi t}=\tfrac{a\sin\theta}{\chi^2\Sigma^3}&\left[\tfrac{1}{2}\sin\theta({a}^{2}+{r}^{2})\Sigma\Delta_\theta\Delta''_\theta-\tfrac{1}{2}\sin\theta\Sigma\Delta_r\Delta''_r+\cos\theta\left(({a}^{2}+{r}^{2})\left(\tfrac{1}{2}\Sigma+r^2+a^2\right)\Delta_\theta+{a}^{2}\sin^2\theta\Delta_r\right)\Delta'_\theta\right.\\
&~\left.+\sin\theta(r\Delta'_r(\Delta_r+({a}^{2}+{r}^{2})\Delta_\theta)+({a}^{4}-{r}^{4})\Delta_\theta^{2}+\Delta_r\Delta_\theta({a}^{2}\cos^2\theta-{r}^{2})-\Delta_r^{2})\right]\\
=~~~~~~~&\hspace{-7.5mm}-\frac{a\sin^2\theta Q^2(\Delta_r+(r^2+a^2)\Delta_\theta)}{\chi^2\Sigma^3}-\frac{\Lambda a\sin^2\theta((r^2+a^2)\Delta_\theta-\Delta_r)}{\chi^2\Sigma},
\end{align*}

\[R_{rr}=\tfrac{1}{\Sigma\Delta_r}\left[-\tfrac{1}{2}\Sigma\Delta''_r+a^2(1+\cos^2\theta)\Delta_\theta+{a}^{2}\cos\theta\sin\theta\Delta'_\theta+r\Delta'_r-\Delta_r\right]=-\frac{Q^2}{\Sigma\Delta_r}+\frac{\Lambda\Sigma}{\Delta_r},\]

\[R_{\theta\theta}=\tfrac{1}{\sin\theta\Sigma\Delta_\theta}\left[-\tfrac{1}{2}\sin\theta\Sigma\Delta''_\theta-\cos\theta\left(\tfrac{1}{2}\Sigma+r^2+a^2\right)\Delta'_\theta-\sin\theta(r\Delta'_r+({a}^{2}-{r}^{2})\Delta_\theta-\Delta_r)\right]=\frac{Q^2}{\Sigma\Delta_\theta}+\frac{\Lambda\Sigma}{\Delta_\theta},\]

\begin{align*}
R_{\phi\phi}=\tfrac{\sin\theta}{\chi^2\Sigma^3}&\left[-\tfrac{1}{2}\sin\theta({a}^{2}+{r}^{2})^{2}\Sigma\Delta_\theta\Delta''_\theta+\tfrac{1}{2}{a}^{2}\sin^3\theta\Sigma\Delta_r\Delta''_r-\cos\theta\Delta'_\theta\left(({a}^{2}+{r}^{2})^{2}\left(\tfrac{1}{2}\Sigma+r^2+a^2\right)\Delta_\theta+{a}^{4}\sin^4\theta\Delta_r\right)\right.\\
&~\left.-\sin\theta(r\Delta'_r(({a}^{2}+{r}^{2})^{2}\Delta_\theta+{a}^{2}\sin^2\theta\Delta_r)+(a^2-r^2)({a}^{2}+{r}^{2})^{2}\Delta_\theta^{2}-\Delta_r\Delta_\theta(a^4\cos^4\theta+2\,{a}^{2}{r}^{2}+{r}^{4})-{a}^{2}\sin^2\theta\Delta_r^{2})\right]\\
=~~~~~~~&\hspace{-7.5mm}\frac{Q^2\sin^2\theta(a^2\sin^2\theta\Delta_r+(r^2+a^2)^2\Delta_\theta)}{\chi^2\Sigma^3}-\frac{\Lambda\sin^2\theta(a^2\Delta_r-(r^2+a^2)^2\Delta_\theta)}{\chi^2\Sigma}.
\end{align*}

Therefore, the Ricci tensor can be written as $R_{\mu\nu}=Q^2R_{\mu\nu}^{\rm ch}+\Lambda g_{\mu\nu}$, where
\[R_{\mu\nu}^{\rm ch}:=\tfrac{1}{\chi^2\Sigma}\begin{pmatrix}\frac{\Delta_r+a^2\sin^2\theta\Delta_\theta}{\Sigma^2}&0&0&-\frac{a\sin^2\theta(\Delta_r+(r^2+a^2)\Delta_\theta)}{\Sigma^2} \\ 0&-\frac{\chi^2}{\Delta_r}&0&0 \\ 0&0&\frac{\chi^2}{\Delta_\theta}&0 \\ -\frac{a\sin^2\theta(\Delta_r+(r^2+a^2)\Delta_\theta)}{\Sigma^2} &0&0& \frac{\sin^2\theta(a^2\sin^2\theta\Delta_r+(r^2+a^2)^2\Delta_\theta)}{\Sigma^2}\end{pmatrix}.\]
As $g^{\mu\nu}R_{\mu\nu}^{\rm ch}=0$, the Ricci scalar is $R=g^{\mu\nu}R_{\mu\nu}=4\Lambda$,
the Einstein tensor reads $G_{\mu\nu}=R_{\mu\nu}-\tfrac{1}{2}Rg_{\mu\nu}=Q^2R_{\mu\nu}^{\rm ch}-\Lambda g_{\mu\nu}$ and we obtain
\[G_{\mu\nu}+\Lambda g_{\mu\nu}=Q^2R_{\mu\nu}^{\rm ch}.\]
\end{landscape}

\begin{landscape}
Consider now the electromagnetic vector potential
\[A_\mu=\frac{rQ}{\chi\Sigma}(\mathrm{d}t-a\sin^2\theta\mathrm{d}\phi).\]
The associated electromagnetic field tensor $F=\mathrm{d}A$ has coordinates $F_{\mu\nu}=\nabla_\mu A_\nu-\nabla_\nu A_\mu=\partial_\mu A_\nu-\partial_\nu A_\mu$ and is given by
\[F_{\mu\nu}=\frac{Q}{\chi\Sigma^2}\begin{pmatrix}0 & r^2-a^2\cos^2\theta & -2ra^2\cos\theta\sin\theta & 0 \\ a^2\cos^2\theta-r^2 & 0 & 0 & a\sin^2\theta(r^2-a^2\cos^2\theta) \\ 2ra^2\cos\theta\sin\theta & 0 & 0 & -2ar\cos\theta\sin\theta(r^2+a^2) \\ 0 & a\sin^2\theta(a^2\cos^2\theta-r^2) & 2ar\cos\theta\sin\theta(r^2+a^2) & 0\end{pmatrix}\]
Raising the indices yields the associated contravariant tensor
\[F^{\mu\nu}=g^{\mu\alpha}F_{\alpha\beta}g^{\beta\nu}=\frac{\chi Q}{\Sigma^3}\begin{pmatrix}0 & (r^2+a^2)(a^2\cos^2\theta-r^2) & 2ra^2\cos\theta\sin\theta & 0 \\ (r^2+a^2)(r^2-a^2\cos^2\theta) & 0 & 0 & a(r^2-a^2\cos^2\theta) \\ -2ra^2\cos\theta\sin\theta & 0 & 0 & -2ra\cot\theta \\ 0 & a(a^2\cos^2\theta-r^2) & 2ra\cot\theta & 0\end{pmatrix}.\]
It is easy to check that for $\nu=t,\phi$, we have $\partial_r(\sqrt{-g}F^{r\nu})+\partial_\theta(\sqrt{-g}F^{\theta\nu})=0$ where $g=\det(g_{\mu\nu})=-\chi^{-4}\sin^2\theta\Sigma^2$. Because $\nabla_\mu F^{\mu\nu}=\tfrac{1}{\sqrt{-g}}\partial_\mu(\sqrt{-g}F^{\mu\nu})$, we obtain that for all $\nu$, we have 
\[\nabla_\mu F^{\mu\nu}=0.\]
Hence, the tensor $F$ is a vacuum solution of the Maxwell equations.

We compute the trace
\[F_{\alpha\beta}F^{\alpha\beta}=2\left(\frac{4Q^2a^2r^2\cos^2\theta}{\Sigma^4}-\frac{Q^2(r^2-a^2\cos^2\theta)}{\Sigma^4}\right)=2(B^2-E^2)\]
and the stress-energy tensor 
\[T_{\mu\nu}=\frac{1}{\mu_0}\left(g^{\alpha\beta}F_{\alpha\mu}F_{\beta\nu}-\tfrac{1}{4}g_{\mu\nu}F_{\alpha\beta}F^{\alpha\beta}\right)=\frac{Q^2}{8\pi\chi^2\Sigma}\begin{pmatrix}\frac{\Delta_r+a^2\sin^2\theta\Delta_\theta}{\Sigma^2}&0&0&-\frac{a\sin^2\theta(\Delta_r+(r^2+a^2)\Delta_\theta)}{\Sigma^2} \\ 0&-\frac{\chi^2}{\Delta_r}&0&0 \\ 0&0&\frac{\chi^2}{\Delta_\theta}&0 \\ -\frac{a\sin^2\theta(\Delta_r+(r^2+a^2)\Delta_\theta)}{\Sigma^2} &0&0& \frac{\sin^2\theta(a^2\sin^2\theta\Delta_r+(r^2+a^2)^2\Delta_\theta)}{\Sigma^2}\end{pmatrix}=\frac{Q^2}{8\pi}R_{\mu\nu}^{\rm ch}=\frac{G_{\mu\nu}+\Lambda g_{\mu\nu}}{8\pi}\]
and so the EME is indeed satisfied by the KNdS metric and the vector potential $A_\mu$. \qed
\end{landscape}
\end{appendix}

\subsection*{Data availability statement} All data that support the findings of this study are included within the article (and any supplementary files).

\printbibliography

\iffalse
\newcommand{\etalchar}[1]{$^{#1}$}

\fi
\end{document}